\documentclass{amsart}
\usepackage{graphicx} 
\usepackage{epsfig}
\usepackage{pstricks,pst-eucl}
\usepackage{mathptmx}
\usepackage{mathrsfs}

\usepackage{bbm,bm}

\usepackage{xcolor}

\let\text=\mbox

\catcode`\@=11
\let\ced=\c
\renewcommand{\a}{\alpha}
\renewcommand{\b}{\beta}
\renewcommand{\d}{\delta}
\newcommand{\g}{\lambda}
\renewcommand{\o}{\omega}
\newcommand{\q}{\quad}

\newcommand{\s}{\sigma}

\newcommand{\cal}{\mathcal}

\newcommand{\M}{{\cal M}}

\newcommand{\ty}{\infty}
\newcommand{\e}{\varepsilon}
\newcommand{\f}{\varphi}
\newcommand{\ov}[1]{\overline{#1}}

\renewcommand{\O}{\Omega}
\newcommand{\pa}{\partial}
\newcommand{\I}{\mathbbm{i}}
\newcommand{\D}{\mathrm{d}}

\newcount\br@j
\newcommand{\udesno}[1]{\unskip\nobreak\hfil\penalty50\hskip1em\hbox{}
             \nobreak\hfil{#1\unskip\ignorespaces}
                 \parfillskip=\z@ \finalhyphendemerits=\z@\par
                 \parfillskip=0pt plus 1fil}
\catcode`\@=11

\newcommand{\eR}{\mathbb{R}}
\newcommand{\eN}{\mathbb{N}}
\newcommand{\Ze}{\mathbb{Z}}

\newcommand{\Ce}{\mathbb{C}}

\newcommand{\re}{\mathop{\mathrm{Re}}}
\newcommand{\im}{\mathop{\mathrm{Im}}}

\newcommand{\po}{{\mathop{\mathcal P}}}

\newcommand{\res}{\operatorname{res}}

\newcommand{\sideremark}[1]{\ifvmode\leavevmode\fi\vadjust{\vbox to0pt{\vss 
      \hbox to 0pt{\hskip\hsize\hskip1em           
 \vbox{\hsize2cm\tiny\raggedright\pretolerance10000
 \noindent #1\hfill}\hss}\vbox to8pt{\vfil}\vss}}}%

\newcommand{\qs}{\q}

\newcommand{\ovb}[1]{\mkern 1.5mu\overline{\mkern-1.5mu#1\mkern-1.5mu}\mkern 1.5mu}

\newcommand{\E}{\mathrm{e}}
\newcommand{\di}{\,\mathrm{d}}

\address[lapidus@math.ucr.edu]{Michel L.\ Lapidus, University of California, Riverside, 900 University Avenue, Riverside, CA 92521}

\address[goran.radunovic@fer.hr]{Goran Radunovi\'c, University of Zagreb, Faculty of Electrical Engineering and Computing, Unska 3, 10000 Zagreb, Croatia}

\address[darko.zubrinic@fer.hr]{Darko \v Zubrini\'c, University of Zagreb, Faculty of Electrical Engineering and Computing, Unska 3, 10000 Zagreb, Croatia}

\newtheorem{theorem}{Theorem}[section]
\newtheorem{corollary}[theorem]{Corollary}

\newtheorem{proposition}[theorem]{Proposition}

\theoremstyle{definition}
\newtheorem{definition}[theorem]{Definition}
\newtheorem{example}[theorem]{Example}
\newtheorem{remark}[theorem]{Remark}

\numberwithin{equation}{section}

\title[Minkowski measurability criteria for compact sets and relative fractal drums]{Minkowski measurability criteria for compact sets and relative fractal drums in Euclidean spaces}

\author[M.\ L.\ Lapidus, G.\ Radunovi\'c and D.\ \v Zubrini\' c]{Michel L.\ Lapidus, Goran Radunovi\'c and Darko \v Zubrini\' c}

\thanks{The work of Michel L.\ Lapidus was partially supported by the US National Science Foundation (NSF) under the research grants DMS-0707524 and DMS-1107750, as well as by the Institut des Hautes Etudes Scientifiques (IHES) in Paris/Bures-sur-Yvette, France, where the first author was a visiting professor in the Spring of 2012 while part of this work was completed.}
\thanks{The research of Goran Radunovi\'c and Darko \v Zubrini\'c was supported by the Croatian Science Foundation under the project IP-2014-09-2285 and by the Franco-Croatian 
PHC-COGITO~project.}

\begin{document}

\begin{abstract}
We survey recent advances dealing with Minkowski measurability criterion for a large class of relative fractal drums (or, in short, RFDs), in Euclidean spaces of arbitrary dimensions in terms of their complex dimensions, which are defined as the poles of their associated fractal zeta functions.
Relative fractal drums represent a far-reaching generalization of bounded subsets of Euclidean spaces as well as of fractal strings studied extensively by the first author and his collaborators.
In fact, the Minkowski measurability criterion discussed here is a generalization of the corresponding one obtained for fractal strings by the first author and M.\ van Frankenhuijsen.
Similarly as in the case of fractal strings, the criterion is formulated in terms of the locations of the principal complex dimensions associated with the relative drum under consideration.
These complex dimensions are defined as poles or, more generally, singularities of the corresponding distance (or tube) zeta function associated.
We also discuss the notion of gauge-Minkowski measurability of RFDs and survey several results connecting it to the nature and location of the complex dimensions.
(This is especially useful when the underlying scaling law does not follow a classic power law.)
We illustrate our results and their applications by means of a number of interesting examples.
\end{abstract}



\maketitle

\tableofcontents

\section{Introduction}

For a nonempty bounded subset $A$ of the $N$-dimensional Euclidean space $\eR^N$, with $N\geq 1$, the attribute `fractal' is most readily associated with the notion of fractal dimension.
As is well known, there are several different notions of fractal dimensions, for example, the Minkowski (or box) dimension, the Hausdorff dimension and the packing dimension (see~\cite{falc}), but not any single one of them encompasses all of the sets that we would like to call fractal.
One only has to recall the example of the famous `devil's staircase'; namely, the graph of the Cantor function, for which all of the known fractal dimensions are trivial (more specifically, are equal to 1).
Nevertheless, one glance at the Cantor graph suffices to make us want to call it ``fractal''.

The recent development of the higher-dimensional theory of complex dimensions in [LapRa\v Zu1--8] 
provides us with a new tool which can be used, among many other things, to reveal some of the elusive `fractality' of sets for which the other well-known fractal dimensions fail.
The complex dimensions of a relative fractal drum or RFD (and, in particular, of a bounded set) in $\eR^n$ are defined as the multiset of poles (or of more general singularities) of the meromorphic extension of the associated fractal distance (or tube) zeta function.
These complex dimensions also generalize the notion of the Minkowski (or box) dimension of a given relative fractal drum.
Actually, even more importantly, they play a major role in determining the asymptotics of the volume of the $t$-neighborhoods of the given RFD (or bounded set) as $t\to0^+$ and therefore, can be considered as a footprint of its inner geometry.
(See \cite{ftf_A} and \cite[Chapter 5]{fzf}.)
 
Of course, being a generalization of the Minkowski dimension, the complex dimensions should also be connected to the property of Minkowski measurability.
Showing this in an explicit manner is one of the main goals of the present survey paper.

Recall that the value of the Minkowski content of a bounded subset $A$ of $\eR^N$ can be used as one of the equivalent ways to define the Minkowski dimension.
More precisely, for a nonempty bounded subset $A$ of $\eR^N$ and $0\leq r\leq N$, we denote the $r$-dimensional Minkowski content of $A$ by
\begin{equation}\label{mink_c}
{{\M}}^{r}(A):=\lim_{\d\to0^+}\frac{|A_\d|}{\d^{N-r}}, 
\end{equation}
whenever this limit exists as a value in $[0,+\ty]$.
Here, $|\cdot|$ denotes the $N$-dimensional Lebesgue measure in $\eR^N$ and
\begin{equation}\label{Adelta}
A_{\d}:=\{x\in\eR^N:d(x,A)<\d\}
\end{equation}
is the {\em $\d$-neighborhood} (or the {\em $\d$-parallel set}) of $A$, with $d(x,A):=\inf\{|x-a|:a\in A\}$ denoting the Euclidean distance from $x\in\eR^N$ to $A$.
The set $A$ is said to be {\em Minkowski measurable} (of dimension $r$) if $\M^r(A)$ exists and satisfies $0<\M^r(A)<\ty$. The Minkowski (or box) dimension of $A$ is then equal to~$r$.

It has been of considerable interest in the past to determine whether or not a set $A$ is Minkowski measurable.
One of the motivations can be found in Mandelbrot's suggestion in~\cite{Man} to use the Minkowski content as a characteristic for the texture of sets (see~\cite[\S X]{Man}).
Mandelbrot called the quantity $1/\M^r(A)$ the {\em lacunarity} of the set $A$ and made the heuristic observation that for subsets of $\eR^N$, small lacunarity corresponds to the spatial homogeneity of the set, which means that the set has small, uniformly distributed holes.
On the other hand, large lacunarity corresponds to the clustering of the set and to large holes between different clusters.
In other words, a small Minkowski content means that $A$ has large, uniformly distributed holes, while a large Minkowski content means that we have small, uniformly distributed holes.  This fits well with the intuitive meaning of $\M^r(A)$ as a `measure' of the `fractal content' of $A$.
More information on this subject can be found in~\cite{BedFi,Fr,lapidusfrank} and in~\cite[\S12.1.3]{lapidusfrank12}. We note that in the last reference (and in [Lap-vFr1]), complex dimensions and the associated residues are suggested to provide a more precise measure of the elusive notion of lacunarity. 

More directly relevant to our present work, prior to that, a lot of attention was devoted to the notion of Minkowski content in connection to the (modified) Weyl--Berry conjecture \cite{Lap1} which relates the spectral asymptotics of the Laplacian on a bounded open set and the Minkowski content of its boundary.
In dimension one, that is, for fractal strings (i.e., for one-dimensional drums with fractal boundary), this conjecture was resolved affirmatively in [LapPo1].\footnote{For the original Weyl--Berry conjecture and its physical applications, see Berry's papers~[Berr1--2].
Furthermore, early mathematical work on this conjecture and its applications can be found in~\cite{BroCar, Lap1, Lap2, Lap3, lapiduspom, LapPo2, FlVa}.
For a more extensive list of later work, see~\cite[\S12.5]{lapidusfrank12}.}  
A crucial part of this result was the characterization of Minkowski measurability of bounded subsets of $\eR$ obtained in~\cite{lapiduspom}.\footnote{A new proof of a part of this result was later given in \cite{Fal2} and, more recently, in \cite{winter}.}
In particular, this led to an important reformulation of the Riemann hypothesis, 
in terms of an inverse spectral problem for fractal strings; see~\cite{LapMa2}.
See the formulation given in~\cite{Lap1} and which was proved for subsets of $\eR$ in 1993 by M.\ L.\ Lapidus and C.\ Pomerance~\cite{lapiduspom}.

Our motivations for introducing relative fractal drums (in \cite{refds} and \cite[Chapter\ 4]{fzf}) and stating the results of this paper in such a generality are two-fold.
They are partly of a practical nature and are also partly due to the emergence of new and interesting phenomena arising in this general context, such as the possibility of negative Minkowski dimensions and of studying the local properties of fractals.
The practicality of working with relative fractal drums (RFDs, in short) reflects itself in the fact that it provides us with a unified approach to both the higher-dimensional theory of complex dimensions [LapRa\v Zu1] and to the well-known theory of complex dimensions for fractal strings [Lap-vFr1--2].
Furthermore, RFDs provide us with a more elegant and convenient way of computing some of the fractal zeta functions for bounded sets by subdividing them into appropriate relative fractal subdrums.

The main goal of this paper is to provide an overview of fractal tube formulas and especially, of the corresponding Minkowski measurability criteria, in the general context of relative fractal drums (and, in particular, of bounded sets) in $\eR^N$, with $N\ge1$ arbitrary, as well as to illustrate these results via a variety of examples. We next give a very brief and informal discussion of some of those results, in a simple situation.

The (pointwise or distributional) {\em fractal tube formula} obtained in [LapRa\v Zu8] (see also [LapRa\v Zu1, \S5.3]) is of the following kind, when expressed in terms of the distance zeta function $\zeta_{A,\O}=\zeta_{A,\O}(s)$ of an RFD $(A,\O)$ in $\eR^N$ and in the case of simple complex dimensions (i.e., when the poles $\o$ of $\zeta_{A,\O}$ are simple, with associated residues denoted by $\res(\zeta_{A,\O},\o)$):

Under appropriate hypotheses on the growth of $\zeta_{A,\O}$, 
\begin{equation}\label{Vt}
\begin{aligned}
V(t)=V_{A,\O}(t)=\sum_{\o\in {\mathcal D}_{A,\O}}\res(\zeta_{A,\O},\o)\frac{t^{N-\o}}{N-\o}+R(t),
\end{aligned}
\end{equation}
where ${\mathcal D}_{A,\O}={\mathcal P}(\zeta_{A,\O})$ is the set of (visible) complex dimensions of the RFD $(A,\O)$ (i.e., the visible poles of $\zeta_{A,\O}$) and $R(t)$ is an error term which can be explicitly estimated when $t\to0^+$. Here, the tube function $V(t)$ denotes the volume of the $t$-neighborhood of $(A,\O)$.\footnote{For the special case of a bounded subset of $A$ of $\eR^N$ and given $t>0$, $V(t)$ is simply the $N$-dimensional measure $|A_t|$ of $A_t$, i.e., of the set of points of $\eR^N$ within a distance less than $t$ from $A$, as given by~\eqref{Adelta}. For a general $RFD$ $(A,\O)$ in $\eR^N$, we let $V(t):=|A_t\cap\O|$.} (See \S2 below for the precise definitions.) Furthermore, under stronger growth hypotheses on $\zeta_{A,\O}$ (which are typically satisfied, for example, in self-similar situations), we can let $R(t)\equiv0$ in Equation~\eqref{Vt} and we therefore obtain an {\em exact} fractal tube formula.

Assuming for simplicity that we could easily separate the {\em principal} complex dimensions (i.e., the poles of $\zeta_{A,\O}$ with maximal real $D$, necessarily equal to the Minkowski dimension of $(A,\O)$) from the other complex dimensions (i.e., the poles of $\zeta_{A,\O}$ with real part strictly less than $D$), we would then conclude that
\begin{equation}\label{Wt}
{\mathcal W}(t):=\frac{V(t)}{t^{N-D}}
\end{equation}
has a (nontrivial) limit as $t\to0^+$ if and only if the only complex dimension with maximal real part $D$ is $D$ itself. Indeed, in light of the fractal tube formula \eqref{Vt}, the presence of {\em nonreal} principal complex dimensions would give rise to actual oscillations in the leading asymptotic term for ${\mathcal W}(t)$ in \eqref{Wt}, and hence, prevent ${\mathcal W}(t)$ from having a limit as $t\to0^+$.\footnote{Naturally, the actual proof of the Minkowski measurability criterion (which is described in \S4 below) is significantly more complicated than is suggested here.} 

In fact, our {\em Minkowski measurability criterion} for relative fractal drums, which extends to arbitrary dimension $N\ge1$ the corresponding criterion for bounded fractal strings obtained in [Lap-vFr1--2, Chapter 8], can be stated as follows (see Theorem \ref{criterion_p} or Theorem \ref{criterion} and Corollary \ref{5.4.20.1/4} below for the precise statement and hypotheses; see also \cite[\S5.4.3]{fzf}).

\begin{theorem}
Let $(A,\O)$ be an RFD in $\eR^N$, with $N\ge1$ arbitrary. Then, under suitable hypotheses on its distance zeta function $\zeta_{A,\O}$, the following conditions are equivalent$:$
\medskip

$(i)$ The RFD $(A,\O)$ is Minkowski measurable, with Minkowski dimension $D$.
\smallskip

$(ii)$ The only principal complex dimension of $(A,\O)$ is $D$, and it is simple $($i.e., it is a simple pole of $\zeta_{A,\O})$.
\end{theorem}

\medskip

Of course, as a corollary, an entirely analogous result holds for bounded subsets $A$ of $\eR^N$, still for any $N\ge1$.
\bigskip

In closing the main part of this introduction, we mention that works involving the notion of Minkowski measurability (or of Minkowski content) in a related context include [BroCar, CaLapPe-vFr1--2, Fal1--2, Fed2, FlVa, Fr, Gat, HamLap, HeLap, HerLap1--3, KeKom, Kom, KomPeWi, Lap1--7, LapL\'eRo, LapLu, LapLu-vFr1--2, LapMa, LapPe1--2, LapPeWi1--2, LapPo1--2, LapRa\v Zu1--8, Lap-vFr1--2, Man2, Pe, PeWi, Ra1--2, RatWi, Res, Sta, Tri1--2, Wi, WiZ\"a, Z\"a3, \v Zu] and the many relevant references therein. 

Furthermore, we refer the interested reader to the end of \S2 for references on tube formulas, including, especially, the fractal tube formulas which play a key role in the proofs of our main results.
\bigskip

The plan of the paper is as follows:

In \S2, we recall some of the main definitions, including the notion of a relative fractal drum (or RFD), as well as of the associated fractal zeta functions (i.e., the distance and tube zeta functions) and their poles in a suitable region (i.e., the visible complex dimensions). We also recall the notion of (relative) Minkowski dimensions and content. Furthermore, we provide some of the basic properties of the fractal zeta functions.

In \S3, we state our main result (Theorem \ref{criterion_p}), according to which, under suitable assumptions, an RFD is Minkowski measurable if and only if the only complex dimension with real part $D$ (the Minkowski dimension of the RFD) is $D$ itself, and it is simple. 

We also recover earlier results about self-similar fractal strings obtained in [Lap-vFr2] (Corollary \ref{5.4.22.1/4C} in \S3.2) and discuss a variety of higher-dimensional examples illustrating our main results; see \S3.3, \S3.4 and \S3.5, respectively, about the Sierpi\'nski gasket and the 3-dimensional carpet, families of fractal nests and unbounded geometric chirps, and self-similar sprays (higher-dimensional analogs of self-similar strings, but of a more general nature than those considered in [LapPo2, Lap2--3, Lap-vFr1--2, LapPe2, LapPeWi1--2]).

In \S4, we give the essence of the proof of the main result which we want to emphasize in this paper (Theorem \ref{criterion_p}, restated more specifically in Theorems \ref{criterion} and \ref{tilde_criterion}, along with Corollary \ref{5.4.20.1/4}), the Minkowski measurability criterion for RFDs (and, in particular, for bounded sets) in $\eR^N$, with $N\ge1$ arbitrary. In the process, we state and establish new results, including a sufficient condition for Minkowski measurability (Theorem \ref{mink_char}), whose proof makes use of the Wiener--Pitt Tauberian theorem, and a necessary condition for Minkowski measurability (Theorem \ref{necess}), the proof of which makes use, in particular, of the distributional tube formula for RFDs obtained in [LapRa\v Zu8] and the uniqueness theorem for almost periodic distributions.

Moreover, in \S5, we construct and study a family of $h$-Minkowski measurable RFDs with respect to a suitable (nontrivial) gauge function $h$ and therefore not obeying a standard power scaling law.  (See, especially, Theorems 5.4, 5.7 and 5.9.) In the process, we show, in particular, that the hypotheses of some of our earlier results in [LapRa\v Zu3--4] about the existence of meromorphic continuations of fractal zeta functions (of bounded sets and RFDs in $\eR^N$) are optimal; see Theorem \ref{tubealpha}.

Finally, in \S6, we propose several future research directions in this area and briefly discuss joint work in progress by the authors [LapRa\v Zu9,10] about various extensions of the notion of complex dimensions which involve essential singularities (in addition to poles) of fractal zeta functions, along with the connections between the non-pole singularities, generalized Minkowski contents and generalized (i.e., non-power like) scaling laws.

\section{Preliminaries: Minkowski content and fractal zeta functions}

We begin by introducing some necessary definitions and results from \cite{refds} (see also~\cite{fzf}) that are needed in order to establish the main results of this paper.
We will always assume throughout this paper that all the sets $A$ and $\O$ are nonempty, in order to avoid dealing with trivial cases.

\begin{definition}\label{zeta_r}\label{drum} 
Let $\Omega$ be a Lebesgue measurable subset of $\eR^N$, not necessarily bounded, but of finite $N$-dimensional Lebesgue measure (or ``volume'').
Furthermore, let $A\subseteq\eR^N$, also possibly unbounded, be such that $\Omega$ is contained in $A_\delta$ for some $\delta>0$.
Then, the {\em distance zeta function $\zeta_{A,\O}$ of $A$ relative to $\Omega$} (or the {\em relative distance 
zeta function} of $(A,\O)$)
 is defined by the following Lebesgue integral:
\begin{equation}\label{rel_dist_zeta}
\zeta_{A,\O}(s):=\int_{\Omega} d(x,A)^{s-N}\D x,
\end{equation}
for all $s\in\Ce$ with $\re s$ sufficiently large.

The ordered pair $(A,\Omega)$, appearing in Definition~\ref{zeta_r} is called   
a {\em relative fractal drum} or RFD, in short. 
We will also use the phrase {\em zeta functions of relative fractal drums} instead of relative zeta functions.
\end{definition}

\medskip

\begin{remark}\label{holo_diff}
In the above definition, we may replace the domain of integration $\O$ in~\eqref{rel_dist_zeta} with $A_\d\cap\O$ for some fixed $\d>0$; that is, we may let
\begin{equation}\label{rel_dist_zeta_d}
\zeta_{A,\O}(s;\d):=\int_{A_\d\cap\Omega} d(x,A)^{s-N}\D x.
\end{equation}
Indeed, the difference $\zeta_{A,\O}(s)-\zeta_{A,\O}(s;\d)$ is then an entire function (see \cite{refds} or \cite{fzf}) so that the above change of the domain of integration does not affect the principal part of the distance zeta function in any way.
Therefore, we can alternatively define the relative distance function of $(A,\O)$ by~\eqref{rel_dist_zeta_d},
since in the theory of complex dimensions we are mostly interested in poles (or, more generally, in singularities) of meromorphic extensions of (various) fractal zeta functions.
Then, in light of the principle of analytic continuation, the dependence of $\zeta_{A,\O}(\,\cdot\,;\d)$ on $\d$ is inessential.
  
The condition that $\O\subseteq A_\d$ for some $\d>0$ is of a technical nature and ensures that the map $x\mapsto d(x,A)$ is bounded for $x\in\O$.
If $\O$ does not satisfy this condition, we can still use the alternative definition given by Equation~\eqref{rel_dist_zeta_d}.\footnote{Since then, $\O\setminus A_\d$ and $A$ are a positive distance apart, this replacement will not affect the relative box dimension of $(A,\O)$ introduced just below or any other fractal properties of $(A,\O)$ that will be introduced later on.}

\end{remark}

\begin{remark}\label{setrfd}
As was already pointed out in the introduction, the notion of a relative fractal drum generalizes the notion of a bounded subset.
Indeed, any bounded subset $A$ of $\eR^N$ may be identified with the relative fractal drum $(A,A_{\d_0})$, for some fixed $\d_0>0$, or even more practically, with $(A,\O)$, where $\O$ is any bounded open set containing $A_{\d_0}$, for some $\d_0>0$.
Of course, in light of Equation \eqref{rel_dist_zeta}, one can then choose the most convenient $\O$.
\end{remark}

Entirely analogous remarks can be made about the tube zeta function of a relative fractal drum, which we now introduce.

\begin{definition}\label{tube_zeta_deff}
Let $(A,\O)$ be an RFD in $\eR^N$ and fix $\d>0$. We define the {\em tube zeta function} $\widetilde{\zeta}_{A,\O}$ {\em of $A$ relative to $\O$} (or the {\em relative tube zeta function}) by the Lebesgue integral
\begin{equation}\label{401/2}
\widetilde{\zeta}_{A,\O}(s):=\widetilde{\zeta}_{A,\O}(s;\d):=\int_0^{\delta}t^{s-N-1}|A_t\cap\O|\,\D t,
\end{equation}
for all $s\in\Ce$ with $\re s$ sufficiently large.
Here,  $|A_t\cap\O|:=|A_t\cap\O|_N$ denotes the $N$-dimensional Lebesgue measure (volume) of $A_t\cap\O\subseteq\eR^N$.
\end{definition} 

The distance and tube zeta functions of relative fractal drums belong to the class of Dirichlet-type integrals (or, in short, DTIs), and as such, have a well-defined {\em abscissa of $($absolute$)$ convergence}.
The abscissa of convergence of a DTI $f\colon E\to\Ce$, where $E\subseteq\Ce$ is a domain, is defined as the infimum of all real numbers $\a$ for which the integral $f(\a)$ is absolutely convergent and we denote it by $D(f)$.\footnote{For more information on DTIs, as well as for their generalizations, we refer the interested reader to~\cite[Appendix~A]{fzf}.} 
It then follows that the Lebesgue integral defining $f(s)$ is convergent (and hence, absolutely convergent) for any $s\in\Ce$ such that $\re s>D(f)$.
A basic result about a DTI $f$ is the fact that it is a holomorphic function in the open half-plane to the right of its abscissa of convergence; that is, on the {\em half-plane of $($absolute$)$ convergence} $\Pi(f):=\{\re s>D(f)\}$.\footnote{Here, and in the rest of this paper we will abbreviate in the following way open half-planes and vertical lines; that is, given $\alpha\in\eR$, subsets of $\Ce$ of the type $\{s\in\Ce:\re s> \a\}$ and $\{s\in\Ce:\re s= \a\}$ are denoted $\{\re s>\alpha\}$ and $\{\re s=\alpha\}$, respectively.}

The relative distance and tube zeta functions are connected by the functional equation
\begin{equation}\label{equ_tilde}
\zeta_{A,\O}(s;\d)=\delta^{s-N}|A_\delta\cap\O|+(N-s)\widetilde\zeta_{A,\O}(s;\d),
\end{equation}
which is valid on any open connected subset $U$ of $\Ce$ to which any of these two zeta functions has a meromorphic continuation (see \cite{refds} or \cite{fzf}).
Observe that in the above functional equation, we use the modified definition of the relative distance zeta function $\zeta_{A,\O}$ given by \eqref{rel_dist_zeta_d}.
If we want to use the definition given by \eqref{rel_dist_zeta}, the above functional equality is valid only if we choose $\delta$ large enough so that $\Omega\subseteq A_\delta$; hence, $A_\delta\cap\O=\O$.
Otherwise, we must add to the right-hand side of \eqref{equ_tilde} another term which is an entire function.
This result is very important since the distance zeta function is much more practical to calculate in concrete examples, as opposed to the tube zeta function for which we need information about the {\em tube function} $t\mapsto|A_t\cap\O|$ itself.
On the other hand, the tube zeta function has an important theoretical value and gives us a way to obtain the information about the asymptotics of the tube function via an application of the inverse Mellin transform directly from the distribution of the complex dimensions of the relative fractal drum under consideration.
This connection enabled us to obtain our {\em pointwise and distributional fractal tube formulas} in \cite{ftf_A} (see also \cite[Chapter 5]{fzf}) which will play a central part in obtaining the aforementioned Minkowski measurability criterion emphasized in the present paper.
These fractal tube formulas extend to any dimension $N\geq 1$ the corresponding pointwise and distributional fractal tube formulas for fractal strings (i.e., when $N=1$) obtained in [Lap-vFr1--3]; 
see \cite[Chapter\ 8, esp.\ \S8.1]{lapidusfrank12}.

We now introduce the notions of Minkowski content and Minkowski (or box) dimension of a relative fractal drum $(A,\O)$ and relate them to the distance (and tube) zeta functions of $(A,\O)$.
For any {\em real} number $r$, we define the {\em upper $r$-dimensional Minkowski content of $A$ relative to $\Omega$}
(or {\em the upper relative Minkowski content}, or {\em the upper Minkowski content of the relative fractal drum $(A,\Omega)$}) by
\begin{equation}\label{minkrel}
{{\mathcal{M}}}^{*r}(A,\Omega):=\limsup_{t\to0^+}\frac{|A_t\cap\Omega|}{t^{N-r}}, 
\end{equation}
and we then proceed in the usual way (except for considering $r\in\eR$ instead of $r\ge0$):
\begin{equation}\label{dimrel}
\begin{aligned}
\ov\dim_B(A,\Omega)&=\inf\{r\in\eR:{{\mathcal{M}}}^{*r}(A,\Omega)=0\} \\
&=\sup\{r\in\eR:{{\mathcal{M}}}^{*r}(A,\Omega)=+\ty\}.
\end{aligned}
\end{equation}
We call $\ov{\dim}_B(A,\O)$ the {\em relative upper box dimension}
 $($\rm{or} {\em relative Minkowski dimension}$)$ of $A$ with respect to $\Omega$ (or else the {\em relative upper box dimension of $(A,\Omega)$}).
Note that $\ov\dim_B(A,\Omega)\in[-\ty,N]$, and that the values can indeed be negative, even equal to $-\ty$; 
see~\cite[\S4.1.2]{fzf}.
Also note that for these definitions to make sense, it is sufficient that $|A_\d\cap\O|<\ty$ for some $\d>0.$

The value of the {\em lower $r$-dimensional Minkowski content} of $(A,\Omega)$, denoted by ${\mathcal{M}}_*^{r}(A,\Omega)$, is defined as in \eqref{minkrel}, except for a lower instead of an upper limit.
Analogously as in \eqref{dimrel}, we define the {\em relative lower box $(${\rm or} Minkowski$)$ dimension} of $(A,\Omega)$ by using the lower instead of the upper $r$-dimensional Minkowski content of $(A,\Omega)$.
Furthermore, in the case when $\underline\dim_B(A,\Omega)=\ov\dim_B(A,\Omega)$, we denote by
$
\dim_B(A,\Omega)
$ 
this common value and call it the {\em relative box $(${\rm or} Minkowski$)$ dimension\label{rel_box_dim}}.
If there exists a real number~$D$ such that
\begin{equation}
0<{\mathcal{M}}_*^D(A,\Omega)\le{\mathcal{M}}^{*D}(A,\Omega)<\ty,
\end{equation}
we say that the relative 
fractal drum $(A,\O)$ is {\em Minkowski nondegenerate}.\label{nondeg_rel}
It then follows that $\dim_B(A,\O)$ exists and is equal to $D$. Otherwise, we say that the RFD $(A,\O)$ is {\em Minkowski degenerate}.

If ${\mathcal{M}}_*^D(A,\Omega)={\mathcal{M}}^{*D}(A,\Omega)$, we denote this common value by $\mathcal{M}^D(A,\Omega)$ and call it the {\em relative 
Minkowski content} of $(A,\O)$.
If $\mathcal{M}^D(A,\Omega)$ exists for some $D\in\eR$ and 
\begin{equation}
0<\mathcal{M}^D(A,\Omega)<\ty
\end{equation}
 (in which case $\dim_B(A,\Omega)$ exists and then, necessarily, $D=\dim_B(A,\Omega)$), we say that the relative fractal drum $(A,\Omega)$ is {\em Minkowski 
measurable}.
Various examples and properties of the relative box dimension can be found in [Lap1--3], [LapPo1--3], \cite{lapidushe}, [Lap--vFr1--3], \cite{rae}, [LaPe2--3], [LapPeWi1--2] and, in full generality, in [LapRa\v Zu1--8].

Theorems \ref{an_rel}, \ref{pole1} and \ref{pole1mink_tilde} below, cited from [LapRa\v Zu1,4], provide some basic results from the theory of zeta functions of relative fractal drums.

\begin{theorem}\label{an_rel}
Let $(A,\O)$ be a relative fractal drum in $\eR^N$. Then the following properties hold$:$

\bigskip 

$(a)$ The relative distance zeta function $\zeta_{A,\O}(s)$ is holomorphic in the half-plane $\{\re s>\overline{\dim}_B(A,\Omega)\}$.
More precisely,
\begin{equation}
D(\zeta_{A,\O})=\overline{\dim}_B(A,\Omega).
\end{equation}

\bigskip

$(b)$ If the relative box $($or Minkowski$)$ dimension $D:=\dim_B(A,\O)$ exists, $D<N$, and ${\M}_*^{D}(A,\O)>0$, then $\zeta_{A,\O}(s)\to+\ty$ as $s\in\eR$
converges to $D$ from the right.   
\end{theorem}

We always have that $\ov{\dim}_B(A,\O)\in[-\infty,N]$.
However, if $A\subset\eR^N$ is a bounded set viewed as an RFD, as in Remark \ref{setrfd}, then $\ov{\dim}_B(A,\O)\in[0,N]$.

\begin{remark}\label{tube_holo}
If $\ovb{\dim}_B(A,\O)<N$, then in light of the functional equation~\eqref{equ_tilde}, Theorem~\ref{an_rel} is also valid if we replace the relative distance zeta function with the relative tube zeta function in its statement.
Moreover, it can be shown directly (i.e., without the use of the functional equation), that in the case of the tube zeta function, Theorem~\ref{an_rel} is also valid when $\ovb{\dim}_B(A,\O)=N$; see \cite{fzf}.
\end{remark}

\begin{theorem}\label{pole1}
Assume that $(A,\O)$ is a nondegenerate RFD in $\eR^N$, 
that is, $0<{\M}_*^{D}(A,\O)\le{\M}^{*D}(A,\O)<\ty$ $($in particular, $\dim_B(A,\O)=D)$, 
and $D<N$.
If $\zeta_{A,\O}(s)$ can be meromorphically extended to a connected open neighborhood of $s=D$,
then $D$ is necessarily a simple pole of $\zeta_{A,\O}(s)$.
Furthermore, the residue $\res(\zeta_{A,\O},D)$ of $\zeta_{A,\O}(s)$ at $s=D$ is independent of the choice of $\d>0$ and satisfies the inequalities
\begin{equation}\label{res}
(N-D){\M}_*^{D}(A,\O)\le\res(\zeta_{A,\O},D)\le(N-D){\M}^{*D}(A,\O).
\end{equation}
Moreover, if $(A,\O)$ is Minkowski measurable, then 
\begin{equation}\label{pole1minkg1=}
\res(\zeta_{A,\O}, D)=(N-D)\M^D(A,\O).
\end{equation}
\end{theorem}

Also, much as in Remark \ref{tube_holo} above, one can reformulate the above theorem in terms of the relative tube zeta function and, in that case, we can remove the condition $\dim_B(A,\O)<N$.

\begin{theorem}\label{pole1mink_tilde}
Assume that $(A,\O)$ is a nondegenerate RFD in $\eR^N$ $($so that $D:=\dim_B(A,\O)$ exists$)$, 
and that for some $\d>0$ there exists a meromorphic extension of $\widetilde\zeta_{A,\O}:=\widetilde\zeta_{A,\O}(\,\cdot\,;\d)$ to a connected open neighborhood of $D$.
Then,  $D$ is a simple  pole,
and $\res(\widetilde{\zeta}_{A,\O},D)$ is independent of $\delta$.
Furthermore, we have
\begin{equation}\label{zeta_tilde_M}
{\M}_*^{D}(A,\O)\le\res(\widetilde\zeta_{A,\O}, D)\le {\M}^{*D}(A,\O).
\end{equation}
In particular, if $(A,\O)$ is Minkowski measurable, then 
\begin{equation}\label{zeta_tilde_Mm}
\res(\widetilde\zeta_{A,\O}, D)=\M^D(A,\O).
\end{equation}
\end{theorem}

Let us now introduce the notion of complex dimensions of a relative fractal drum.

\begin{definition}[{\it Complex dimensions of an RFD} {[LapRa\v Zu1,4]}] 
Let $(A,\O)$ be a relative fractal drum in $\eR^N$. 
Assume that the relative distance zeta function ${\zeta}_{A,\O}$ has a meromorphic extension to a connected open neighborhood $U$ of the critical line $\{\re s=\ov{\dim}_B(A,\O)\}$. Then, the set of {\em visible complex dimensions of $(A,\O)$ $($with respect to $U)$} is the set of poles of ${\zeta}_{A,\O}$ that are contained in $U$ and we denote it by
\begin{equation}
\po({\zeta}_{A,\O},U):=\{\omega\in U:\omega\textrm{ is a pole of }{\zeta}_{A,\O}\}.
\end{equation}
If $U=\Ce$, we say that $\po({\zeta}_{A,\O},\Ce)$ is the set of {\em complex dimensions} of $(A,\O)$ and denote it by $\dim_{\Ce}(A,\O)$.

Furthermore, we call the set of poles located on the critical line $\{\re s=\ov{\dim}_B(A,\O)\}$ the set of {\em principal complex dimensions of $(A,\O)$} and denote it by
\begin{equation}
\dim_{PC}(A,\O):=\{\omega\in \po({\zeta}_{A,\O},U):\re \omega=\ov{\dim}_B(A,\O)\}.
\end{equation}

\medskip

Finally, for each (visible or principal) complex dimension $\omega$, we define its {\em multiplicity} as the order of the pole $\omega$ of ${\zeta}_{A,\O}$.
In light of this, we will sometimes refer to the set of (visible or principal) complex dimensions as a multiset (that is, a set with integer multiplicities).
Observe that the multiset of principal complex dimensions is clearly independent of the domain $U$ chosen as above. 
\end{definition}

\begin{remark}
We may also define the complex dimensions in terms of the tube zeta function instead of the distance zeta function; that is, analogously as in the above definition but with ${\zeta}_{A,\O}$ replaced by $\widetilde{\zeta}_{A,\O}:=\widetilde{\zeta}_{A,\O}(\,\cdot\,;\delta)$ for some $\delta>0$.\footnote{It does not matter which $\delta>0$ we choose since it does not affect the singularities of the tube zeta function.}
Furthermore, if $\ov{\dim}_B(A,\O)<N$, then in light of the functional equation \eqref{equ_tilde}, the complex dimensions do not depend on the choice of the fractal zeta function.
This is also true in the special case when $\ov{\dim}_B(A,\O)=N$ except that care must be taken in this situation since then, it may happen that the tube zeta function has a pole at $s=N$ while the distance function does not because the pole at $s=N$ may be canceled by the factor $(N-s)$.
In order to avoid this problem, we will always assume that $\ov{\dim}_B(A,\O)<N$ when working with the distance zeta function.
\end{remark}
 
In order to be able to formulate our main result, we will need to introduce several definitions connected to the growth properties of the distance (or tube) zeta functions.
These definitions are adapted from~[Lap-vFr1--2] where they were used in the setting of generalized fractal strings and their geometric zeta functions. 

\begin{definition}\label{window_def_5}
The {\em screen} $\bm{S}$ is the graph of a bounded, real-valued Lipschitz continuous function $S(\tau)$, with the horizontal and vertical axes interchanged:
\begin{equation}\label{g_screen}
\bm{S}:=\{S(\tau)+\I \tau\,:\,\tau\in\eR\}.
\end{equation}
The Lipschitz constant of $S$ is denoted by $\|S\|_{\mathrm{Lip}}$; so that
$$
|S(x)-S(y)|\leq\|S\|_{\mathrm{Lip}}|x-y|,\quad\textrm{ for all }x,y\in\eR.
$$
Furthermore, the following quantities are associated to the screen:
\begin{equation}\nonumber
\inf S:=\inf_{\tau\in\eR}S(\tau)\quad\textrm{ and }\quad\sup S:=\sup_{\tau\in\eR}S(\tau).
\end{equation}
\end{definition}

From now on, given an RFD $(A,\O)$ in $\eR^N$, we denote its upper relative box dimension by $\ov{D}:=\ov{\dim}_B(A,\O)$; recall that $\ov{D}\leq N$.
We always assume, additionally, that $\ov{D}>-\ty$ and the screen $\bm{S}$ lies to the left of the {\em critical line} $\{\re s=\ov{D}\}$, i.e., that $\sup S\leq\ov{D}$.
Also, in the sequel, we assume that $\inf S>-\ty$; thus, we have that
\begin{equation}\label{5.1.1.1/2-}
-\ty<\inf S\leq\sup S\leq\ov{D}.
\end{equation}

Moreover, the {\em window} $\bm{W}$ is defined as the part of the complex plane to the right of $\bm{S}$; that is,
\begin{equation}
\bm{W}:=\{s\in\Ce:\re s\geq S(\im s)\}.
\end{equation}
(Note that $\bm{W}$ is a closed subset of $\Ce$.)
We say that the relative fractal drum $(A,\O)$ is {\em admissible} if its relative tube (or distance) zeta function can be meromorphically extended (necessarily uniquely) to an open connected neighborhood of some window $\bm{W}$, defined as above.


The next definition adapts~\cite[Definition~5.2]{lapidusfrank12} to the case of relative fractal drums in $\eR^N$ (and, in particular, to the case of bounded subsets of $\eR^N$).

\begin{definition}[{$d$-languidity}]\label{languid}
An admissible relative fractal drum $(A,\O)$ in $\eR^N$ is said to be {\em $d$-languid} if for some fixed $\d>0$, its distance zeta function ${\zeta}_{A,\O}:={\zeta}_{A,\O}(\,\cdot\,;\d)$ satisfies the following growth conditions:

\medskip

There exists a real constant $\kappa_d$ and a two-sided sequence $(T_n)_{n\in\Ze}$ of real numbers such that $T_{-n}<0<T_n$ for all $n\geq 1$ and
\begin{equation}\label{seq_T_n}
\lim_{n\to\ty}T_n=+\ty,\quad\lim_{n\to\ty}T_{-n}=-\ty
\end{equation}
satisfying the following two hypotheses, {\bf L1} and {\bf L2}:
\newline

{\bf L1}\ \ For a fixed real constant $c>N$, there exists a positive constant $C>0$ such that for all $n\in\Ze$ and all $\s\in (S(T_n),c)$,
\begin{equation}\label{L1}
|{\zeta}_{A,\O}(\s+\I T_n;\d)|\leq C(|T_n|+1)^{\kappa_d}.
\end{equation}

{\bf L2}\ \ For all $\tau\in\eR$, $|\tau|\geq 1$,
\begin{equation}\label{L2}
|{\zeta}_{A,\O}(S(\tau)+\I \tau;\d)|\leq C|\tau|^{\kappa_d},
\end{equation}
where $C$ is a positive constant which can be chosen to be the same one as in condition~{\bf L1}.
\end{definition}

Note that these are (at most) polynomial growth conditions on ${\zeta}_{A,\O}(\,\cdot\,;\d)$ along a sequence of horizontal segments and along the vertical direction of the screen.
We call the exponent $\kappa_d\in\eR$ appearing in the above definition the {\em $d$-languidity exponent} of $(A,\O)$ (or of ${\zeta}_{A,\O}$). Here, the index $d$ in $\kappa_d$ stands for the Euclidean distance appearing in the definition of the relative distance zeta function $\zeta_{A,\O}$ defined by~\ref{rel_dist_zeta}.

We will also use the notion of {\em languid} relative tube zeta function (or {\em languid} RFD) if the analogous conditions as in Definition \ref{languid} are satisfied for the tube zeta function.
(See \cite[Definition 2.12 and 2.13]{ftf_A} or \cite{fzf}.)
In that case, we denote the {\em languidity exponent} by $\kappa\in\eR$.
There exists also conditions of {\em strong $d$-languidity} and {\em strong languidity} in which we assume that the screen may be ``pushed'' to $-\ty$; these conditions were needed in order to obtain the exact fractal tube formulas in \cite{ftf_A}.
Since we will not need them explicitly in this paper, we refrain from defining these conditions rigorously.

\medskip

It turns out that all of the geometrically interesting examples of RFDs (and, in particular, of bounded sets) in $\eR^N$ considered here are $d$-languid (relative to a suitable screen).
This illustrates the fact that the results of this paper can be applied to a large class of RFDs.
For instance the so-called ``self-similar'' RFDs with generators that are nice ``enough'' (including, but not limited to, monophase and pluriphase generators in the sense of [LapPe2, LapPeWi1--2]), 
are in this class.
Furthermore, a wide range of RFDs belonging to this class and which do not exhibit any kind of self-similarity can be constructed. 
Hence, the theory developed here can be also applied to them. 

Although, as was already explained earlier, the dependence of the distance zeta function $\zeta_{A,\O}(\,\cdot\,;\d)$ on $\d>0$ is inessential, it is not clear if this is also true for the $d$-languidity conditions.
More precisely, it is true that changing the parameter $\d>0$ will preserve $d$-languidity but possibly with a different exponent $\kappa_d$.
This is the statement of the next proposition, which was proved in \cite[Proposition 5.27]{ftf_A} and which we partly state next, for the sake of exposition. 

\begin{proposition}\label{propB}
Let $(A,\O)$ be a relative fractal drum in $\eR^N$.
If the relative distance zeta function ${\zeta}_{A,\O}(\,\cdot\,;\d)$ satisfies the languidity conditions {\bf L1} and {\bf L2} for some $\d>0$ and $\kappa_d\in\eR$, then so does
${\zeta}_{A,\O}(\,\cdot\,;\d_1)$ for any $\d_1>0$ and with ${(\kappa_d)}_{\d_1}:=\max\{\kappa_d,0\}$.
\end{proposition}

We note that tube formulas in convex, `smooth' and `fractal' geometry, as well as in integral geometry and geometric measure theory, have a long history, going back to the work of Steiner [Stein], Minkowski [Mink], Weyl [Wey] and Federer [Fed1], in particular. Works on these topics include [BergGos, Bla, DemDenKo\"U, DemKo\"O\"U, DeK\"O\"U, Gra, HamLap, HuLaWe, KlRot, LapLu, LapLu-vFr1--2, LapPe1--2, LapPeWi1--2, LapRa\v Zu1,7--8, Lap-vFr1--2, Ol1--2, Ra1, Schn, Wi, WiZ\"a, Z\"a1--3] and the many other relevant references therein.

\section{Statement of the main result and applications}

In this section, we state our main result in Theorem \ref{criterion_p} and illustrate its applications in a variety of examples and situations, including the recovery of the corresponding result for self-similar fractal strings in Corollary \ref{5.4.22.1/4C}, the 
 Cantor set in Example \ref{ecant}, the $a$-string in Example \ref{ex_a}, as well as the higher-dimensional examples of the Sierpi\'nski gasket and $3$-carpet in \S\ref{subsec_sier}, the fractal nests and unbounded geometric chirps in \S\ref{subsec_nestch}, and self-similar sprays in \S\ref{subsec_self_similar_sp}.
We will give a sketch of the proof of Theorem \ref{criterion_p} in \S\ref{main_proof}.

\subsection{The Minkowski measurability criterion}

We begin by stating the following theorem, in which we give a characterization of the Minkowski measurability of a large class of relative fractal drums in terms of the distribution of the location of their complex dimensions with maximal real part (i.e., of their principal complex dimensions).

\begin{theorem}[Minkowski measurability criterion]\label{criterion_p}
Let $(A,\O)$ be a relative fractal drum in $\eR^N$ such that $D:=\dim_B(A,\O)$ exists and $D<N$.
Furthermore, assume that $(A,\O)$ is $d$-languid for a screen passing {\rm strictly} between the critical line $\{\re s=D\}$ and all the complex dimensions of $(A,\O)$ with real part strictly less than $D$.
Then the following statements are equivalent$:$

\medskip

$(a)$ The RFD $(A,\O)$ is Minkowski measurable.

\medskip

$(b)$ $D$ is the only pole of the relative distance zeta function ${\zeta}_{A,\O}$ located on the critical line $\{\re s=D\}$ and it is simple.

\medskip

Furthermore, if $(a)$ or $(b)$ is satisfied, then the Minkowski content of $(A,\O)$ is given by
\begin{equation}
\mathcal{M}^D(A,\O)=\frac{\res(\zeta_{A,\O},D)}{N-D}.
\end{equation}
\end{theorem}

\begin{remark}\label{5.4.19.1/2}
Theorem \ref{criterion_p} extends to RFDs in $\eR^N$, with $N\geq 1$ arbitrary, the Minkowski measurability criterion for fractal strings obtained in \cite[Theorem~8.15 of \S8.3]{lapidusfrank12}.
The latter criterion corresponds to the $N=1$ case of Theorem \ref{criterion}.
We note that in \cite{lapidusfrank12}, the criterion was formulated in terms of the principal complex dimensions of the underlying fractal string (interpreted as the poles of the associated geometric zeta function with real part equal to $D$, the Minkowski dimension of the string).
However, in light of the results of \S\ref{subsec_frstr} below (see, especially, Proposition \ref{geo_dist}), they can now be restated equivalently in terms of the principal poles of the distance zeta function of the corresponding relative fractal drum $(\pa\O,\O)$, where $\O$ is any geometric realization of the fractal string.
We mention that in the statement of \cite[Theorem~8.15 of \S8.3]{lapidusfrank12}, the fact that the fractal string was weakly (rather than strongly) Minkowski measurable should have been stressed more explicitly.\footnote{See Definition \ref{5.4.16.1/2} below for the precise definition of strong and weak Minkowski measurability.}
On the other hand, in \cite[Theorems~8.23 and 8.36 of \S8.4]{lapidusfrank12}, the characterization of Minkowski measurability obtained for self-similar strings was stated in terms of the strong (i.e., ordinary) Minkowski measurability of the fractal strings.\footnote{Namely, a self-similar string (of Minkowski dimension $D\in(0,1)$) is (strongly) Minkowski measurable if and only if it is a nonlattice string (i.e., the multiplicative group generated by its distinct scaling ratios is not of rank $1$) or, equivalently, if its Minkowski dimension $D$ is the only complex dimension of real part $D$.
(It is known from \cite[Theorems~2.16 and~3.6]{lapidusfrank12}) that for a self-similar string, $D$ is always simple.)}
We do not consider the $N$-dimensional counterpart of such a situation here, although this would certainly be interesting.
Our approach allows us to unify (as well as significantly extend) these two types of results and to obtain results that are stated in terms of the strong (i..e., ordinary) Minkowski measurability.

Beside the obvious difficulty of computing the distance (or another fractal) zeta function of a (suitable) `self-similar RFD' in $\eR^N$ (and, in particular, of a compact self-similar set in $\eR^N$ satisfying the open set condition, say), an important remaining issue is to remove (as was done in \cite[\S8.4]{lapidusfrank12} when $N=1$) the condition concerning the existence of an appropriate screen $\bm{S}$ (as is assumed in Theorem \ref{criterion}, as well as in Theorem \ref{tilde_criterion} and Corollary \ref{5.4.20.1/4} below).
Indeed, in the lattice case, this condition is clearly always satisfied as long as the `base RFD' is `nice enough' (see \cite[\S4.3 and esp., Theorem 4.21]{refds} or \cite[\S4.2.2 and esp., Theorem 4.2.7]{fzf} where the distance zeta function for self-similar RFDs was computed).
On the other hand, it is shown in \cite{lapidusfrank12} that in the nonlattice case, there are examples of nonlattice self-similar strings (and hence, of self-similar RFDs and sets in $\eR^N$) for which it is not fulfilled.
(See \cite[Example~5.32]{lapidusfrank12} showing that for a given self-similar string, it is not always possible to choose a screen $\bm{S}$ passing strictly between the critical line $\{\re s=D\}$ and the complex dimensions to the left of this line and along which the RFD is languid.)

However, we also stress that this issue regarding the nonlattice case is, a priori, occurring only in establishing one direction of the desired (Minkowski measurability) characterization theorem; that is, in the direction which aims at proving that, under suitable hypotheses, a Minkowski measurable self-similar RFD (with a `nice enough' base RFD or generator) is always nonlattice.\footnote{Actually, this is only a problem at first glance.
In fact, a moment's reflection shows that it suffices to reason by contradiction in order to resolve this problem. 
Indeed, in the lattice case, we can always find a suitable screen satisfying the required hypotheses of Theorem \ref{criterion}
and, consequently, conclude that a lattice RFD is not Minkowski measurable and reach a contradiction after an application of Theorem \ref{criterion}.
We will proceed exactly in this manner in the proof of Corollary \ref{5.4.22.1/4C} (where $N=1$).} 

Indeed, for the other direction, we do not need the restrictive assumption about the existence of a suitable screen along which the RFD is languid since the desired Minkowski measurability conclusion would follow from the sufficient condition provided in Theorem \ref{mink_char}.
More specifically, 
under appropriate hypotheses,\footnote{In particular, we assume that the generators of the associated self-similar tilings (or sprays) are pluriphase (in the sense of \cite{lappe2} and \cite{lappewi1}, along with \cite[\S13.1]{lapidusfrank12}) and 
that the Minkowski dimension of their boundary is strictly smaller than their similarity dimension (or, equivalently, than the similarity dimension of the associated self-similar tiling or spray).} we expect to conclude that a nonlattice self-similar RFD is always Minkowski measurable since the only pole on the critical line is its Minkowski dimension (which is equal to the maximum of the inner Minkowski dimension of the boundary of the generator and the unique real solution of their associated complexified Moran equation, 
and it is simple.\footnote{We would then recover via our general theory a result of \cite{gatzouras} establishing one direction of the geometric part of \cite[Conjecture 3]{Lap3} (the converse direction having recently been established in \cite{KomPeWi}).}
Of course, the above potential ``theorem'' should be more precisely stated, with the expressions of `self-similar RFD', `nice enough' and `base RFD' (or `generator') being unambiguously specified.
We leave this task for a future work or to the interested readers.
(See also the end of \S\ref{subsec_self_similar_sp} below where these issues are addressed and essentially resolved in the important special case of self-similar sprays, under mild assumptions.)
\end{remark}

\subsubsection{Self-similar strings: Minkowski measurability and the lattice/nonlattice dichotomy}

In the next result (Corollary \ref{5.4.22.1/4C}), which follows from a combination of Theorem \ref{mink_char} and Theorem \ref{criterion}, we recover the aforementioned characterization of the Minkowski measurability of self-similar fractal strings (with possibly multiple gaps,\footnote{We note that the general case of multiple gaps precisely corresponds to the general case of multiple generators for self-similar sprays.} in the sense of \cite[Chapters~2 and~3]{lapidusfrank12}), obtained in \cite[\S8.4, esp., Theorems 8.23, 8.25 and 8.36, along with Corollary 8.27]{lapidusfrank12}.
In the proof of the next corollary, we will use the known fact (established in \cite[Theorems 2.16 and 3.6]{lapidusfrank12}) that a self-similar string is nonlattice if and only if its only scaling complex dimension (i.e., the only pole of its geometric or scaling zeta function) with real part $\sigma_0$ (the similarity dimension of the string) is $\sigma_0$ itself.
Note that this last statement also uses the fact (established in a part of \cite[Corollary 8.27]{lapidusfrank12}) according to which a lattice self-similar string (with multiple gaps) has infinitely many principal scaling complex dimensions (i.e., potential poles of the geometric or scaling zeta function with real part $\sigma_0$ and with nonzero residues) of the form $\sigma_0+\I k\mathbf{p}$, where $\mathbf{p}:=2\pi/\log r^{-1}$, and $r\in(0,1)$ is the single generator of the multiplicative group (of rank one) generated by the underlying distinct scaling ratios; therefore, it has at least one nonreal principal complex dimension.
(This latter fact is easy to check in the case of a single gap; then, the set of principal scaling complex dimensions is all of $\sigma_0+\mathbf{p}\I\Ze$.)
Finally, we recall from \cite[Chapters 2 and 3]{lapidusfrank12} that $\sigma_0$ (and hence, each of the other principal scaling complex dimensions) is always simple, either in the lattice case or in the nonlattice case.
(This latter fact is also easy to check directly from the definitions.)

In the following corollary of Theorem 4.2 and Theorem 4.15, 
combined with the results in \cite[Chapters 2 and 3]{lapidusfrank12}, $\O$ denotes an arbitrary geometric realization of a (nontrivial) bounded self-similar fractal string $\mathcal{L}:=(\ell_j)_{j=1}^{\ty}$, as a bounded open subset of $\eR$.
Furthermore, $\pa\O$ denotes its boundary (in $\eR$) and $(\pa\O,\O)$ is the associated relative fractal drum (or RFD in $\eR$).
We note that in \cite{lapidusfrank12}, the term RFD (or `relative fractal drum') was not used but that an equivalent notion was used instead in the present situation of fractal strings.

\begin{corollary}[Characterization of the Minkowski measurability of self-simi\-lar strings, {\cite[\S8.4]{lapidusfrank12}}]\label{5.4.22.1/4C}
Let $(\pa\O,\O)$ $($or $\mathcal{L}$$)$ be a $($nontrivial, bounded$)$ self-similar fractal string, with $($upper$)$ Minkowski dimension $D<1$; so that $D=\sigma_0$, its similarity dimension.
Then the following statements are equivalent$:$

\medskip

$(i)$ The RFD $(\pa\O,\O)$ is Minkowski measurable.

\medskip

$(ii)$ The self-similar string $\mathcal{L}$ $($or, equivalently, the self-similar RFD $(\pa\O,\O)$$)$ is nonlattice.

\medskip

$(iii)$ The only principal scaling complex dimension of $(\pa\O,\O)$ is $\sigma_0$.
\end{corollary}

\begin{proof}
We already know from the discussion preceding the statement of the corollary that $(ii)$ and $(iii)$ are equivalent, based on the results of \cite[Chapters 2, 3 and 8]{lapidusfrank12}; see, especially, \cite[Theorems 8.23 and 8.36]{lapidusfrank12}.

Next, we show that $(i)$ and $(iii)$ are equivalent.
Note that since $D=\sigma_0$ and $\sigma_0>0$, we have that $D\in(0,1)$.

First, observe that (since $\sigma_0$ is always simple) the fact that $(iii)$ implies $(i)$ follows from Theorem \ref{mink_char}, the sufficient condition for the Minkowski measurability of an RFD.
Note that in order to verify that the hypotheses of Theorem \ref{mink_char} are satisfied by $(\pa\O,\O)$ and its distance zeta function $\zeta_{\pa\O,\O}$, we use the fact that the scaling (and hence, geometric) zeta function $\zeta_{\mathfrak{S}}$ of a self-similar string is strongly languid with exponent $\kappa:=0$ (and hence, also for any $\kappa\geq 0$), as is shown in \cite[\S6.4, just above Remark 6.12]{lapidusfrank12}, combined with the functional equation (recalled in Proposition 3.7 below)
\begin{equation}\label{5.4.52.1/2E}
\zeta_{\pa\O,\O}(s)=\frac{2^{1-s}\zeta_{\mathcal{L}}(s)}{s}
\end{equation}
connecting the distance zeta function of $(\pa\O,\O)$ and the geometric zeta function of the fractal string $\mathcal{L}$.
For the proof of the functional equation \eqref{5.4.52.1/2E}, we refer to \cite[Proposition 6.3 and Remark 6.4]{ftf_A} or \cite[\S5.5.2]{fzf}.
Therefore, the distance zeta function $\zeta_{\pa\O,\O}$ is strongly $d$-languid with exponent $\kappa_d:=-1$.
Consequently (and assuming that $(iii)$ holds), the hypotheses of Theorem \ref{mink_char} are satisfied and so, it follows that $(\pa\O,\O)$ is Minkowski measurable; i.e., $(i)$ holds, as desired.

Now, all that remains to show is that $(i)$ implies $(iii)$.
More explicitly, we need to show that the fact that $(\pa\O,\O)$ is Minkowski measurable implies that $\mathcal{L}$ (or, equivalently, $(\pa\O,\O)$) is nonlattice.
For this purpose, we reason by contradiction.
Namely, we assume that $(i)$ holds (i.e., $(\pa\O,\O)$ is Minkowski measurable) but that $\mathcal{L}$ is a {\em lattice} (self-similar) string.
Since $\mathcal{L}$ is lattice,\footnote{In the case of a self-similar string with multiple gaps, this means that the multiplicative group generated by both the distinct scaling ratios and gaps is of rank 1; see \cite[Chapters 2 and~3]{lapidusfrank12}.} its scaling complex dimensions are located (and periodically distributed) on finitely many vertical lines (possibly on a single such line), the right most of which is the vertical line $\{\re s=\sigma_0\}$, the critical line (since $\sigma_0=D$).
Therefore, we can obviously choose, as is required by the hypotheses of Theorem \ref{criterion} (the Minkowski measurability criterion), a screen $\bm{S}$ passing strictly between the vertical line $\{\re s=\sigma_0=D\}$ and all the complex dimensions (i.e., the poles of $\zeta_{\pa\O,\O}$) with real part strictly less than $D=\sigma_0$.
In light of Equation \eqref{5.4.52.1/2E}, it suffices to let $S$ be any vertical line $\{\re s=\Theta\}$, where $\Theta\in(\max\{0,\sigma_1\},\sigma_0)$ and $\sigma_1$ is the abscissa of the second to last (right most) vertical line on which the scaling complex dimensions of $\mathcal{L}$ (or of $(\pa\O,\O)$) are located.  
(If $\sigma_1$ does not exist, then we can choose any $\Theta\in(0,\sigma_0)$.)\footnote{The fact that the strong $d$-languidity assumption is satisfied by $(\pa\O,\O)$ is explained in the previous step of the proof.
The corresponding argument is valid for any self-similar string.}

We may therefore apply Theorem \ref{criterion} and deduce from the fact that the RFD $(\pa\O,\O)$ is Minkowski measurable and that $D=\sigma_0$ must be its only complex dimension of real part $D=\sigma_0$.\footnote{Note that, in light of \eqref{5.4.52.1/2E} and since $\sigma_0>0$, it follows that the principal complex dimensions of $(\pa\O,\O)$ and the principal scaling complex dimensions coincide: $\po_c(\zeta_{\pa\O,\O})=\po_c(\zeta_{\mathcal{L}})$.}
This contradicts the fact that $\mathcal{L}$ is a lattice string, and hence has infinitely many (and thus at least two complex conjugate) nonreal principal scaling complex dimensions.
We deduce from this contradiction that $\mathcal{L}$ must be a nonlattice string (i.e., $(ii)$ holds) and hence (since $(ii)$ and $(iii)$ are equivalent), that $(iii)$ holds, as desired.

This concludes the proof of the corollary.
\end{proof}

\begin{remark}
Naturally, it follows from Corollary \ref{5.4.22.1/4C} that a (nontrivial, bounded) {\em self-similar string is} not {\em Minkowski measurable if and only if it is lattice}. Recall that a self-similar string is always Minkowski nondegenerate.
\end{remark}

\begin{remark}\label{5.4.22.1/2R}
In \S\ref{subsec_self_similar_sp}, we will extend Corollary \ref{5.4.22.1/4C} to higher dimensions, that is, to a large class of self-similar sprays in $\eR^N$, with $N\geq 1$ arbitrary.
In the general case (and under some mild assumptions), Minkowski measurability will have to be replaced by `possibly subcritical Minkowski measurability' in a sense to be explained there.
\end{remark}

In the rest of this section, we illustrate the results of this paper by applying them to several examples of bounded (fractal) sets and relative fractal drums.
These examples include the line segment, the Cantor string, the Sierpi\'nski $N$-gasket and the $3$-dimensional Sierpi\'nski carpet,
with $N\geq 3$ (\S\ref{subsec_sier}), fractal nests and (unbounded) geometric chirps (\S\ref{subsec_nestch}), as well as, finally, the recovery and significant extensions of the known fractal tube formulas (from [LapPe2--3, LapPeWi1--2]) 
for self-similar sprays (\S\ref{subsec_self_similar_sp}), with  applications to the $1/2$-square and $1/3$-square fractals.

\subsection{Application to general fractal strings}\label{subsec_frstr}

In the present subsection, we illustrate how we can apply the Minkowski measurability citerion stated in Theorem \ref{criterion} (or Theorem \ref{criterion_p}) to the case of general (not necessarily self-similar) fractal strings.

We begin by discussing the prototypical example of the Cantor string (viewed as an RFD), in Example \ref{ecant}, and further illustrate our results by means of the well-known example of the $a$-string (in Example \ref{ex_a}).
Along the way, we discuss the case of general fractal strings as well as the associated fractal tube formulas.

\begin{example}\label{ecant}({\em The standard ternary Cantor set and string}).
Let $C$ be the standard ternary Cantor set in $[0,1]$ and fix $\d\geq 1/6$.
Then, the `absolute' distance zeta function of $C$ is meromorphic in all of $\Ce$ and given by
\begin{equation}\label{cant}
\zeta_{C,C_{\d}}(s)=\frac{2^{1-s}}{s(3^s-2)}+\frac{2\d^s}{s},\quad\textrm{for all}\ s\in\Ce;
\end{equation}
see \cite[Example 3.4]{dtzf} or \cite[Example 2.1.18]{fzf}.
Similarly, the relative distance zeta function of the relative fractal drum $(C,(0,1))$ is also meromorphic on all of $\Ce$ and given by
\begin{equation}\label{Cant_zeta}
\zeta_{C,(0,1)}(s)=\frac{2^{1-s}}{s(3^s-2)},\quad\textrm{for all}\ s\in\Ce.
\end{equation}
Furthermore, the sets of complex dimensions of the Cantor set $C$ and of the Cantor string $(C,(0,1))$, viewed as an RFD, coincide:
\begin{equation}\label{comp_dim_CC}
\po({\zeta}_C)=\po({\zeta}_{C,(0,1)})=\{0\}\cup\left(\log_32+\frac{2\pi}{\log3}\I\Ze\right).
\end{equation}
In \eqref{comp_dim_CC}, each of the complex dimensions is simple.
Furthermore, $D:=\dim_B(C,(0,1))$, the Minkowski dimension of the Cantor string, exists and $D=\log_32$, the Minkowski dimension of the Cantor set, which also exists.
Moreover, $\mathbf{p}:={2\pi}/{\log 3}$ is the oscillatory period of the Cantor set (or string), viewed as a  {\em lattice} self-similar set (or string); see \cite[Chapter 2, esp., \S2.3.1 and \S2.4]{lapidusfrank12}.

It is easy to check that $(C,(0,1))$ is $d$-languid where the screen may be choosen to coincide with any vertical line $\{\re s=\sigma\}$ for $\sigma\in(0,\log_32)$.
Therefore, the assumptions of Theorem \ref{criterion} are satisfied but since $D$ is not the only pole on the critical line $\{\re s=D\}$ we conclude that, as is well-known, the relative fractal drum $(C,(0,1))$, i.e., the Cantor string is not Minkowski measurable.

In fact, from \cite[Example 6.2]{ftf_A} we have the following exact pointwise fractal tube formula for the `inner' $t$-neighborhood of $C$, valid for all $t\in(0,1/2)$:
\begin{equation}\label{CC_compute}
\begin{aligned}
|C_t\cap(0,1)|&=t^{1-D}G\left(\log_{3}(2t)^{-1}\right)-2t,
\end{aligned}
\end{equation}
where $G$ is the positive, nonconstant $1$-periodic function, which is bounded away from zero and infinity and given by the following Fourier series expansion:
\begin{equation}\label{5.5.10.1/4}
G(x):=\sum_{k\in\Ze}\frac{2^{-D}\E^{2\pi\I kx}}{\omega_k(1-\omega_k)\log 3},
\end{equation}
where $\omega_k:=D+\mathbf{p}\I k$ and $\mathbf{p}$ is the oscillatory period of the Cantor string.

As expected, the above exact pointwise fractal tube formula \eqref{CC_compute} coincides with the one obtained by direct computation for the Cantor string (see~\cite[\S1.1.2]{lapidusfrank12}) or from the general theory of fractal tube formulas for fractal strings (see \cite[Chapter~8, esp., \S8.1 and \S8.2]{lapidusfrank12}) and, in particular, for self-similar strings (see, especially, \cite[\S8.4.1, Example 8.2.2]{lapidusfrank12}).\footnote{{\em Caution}: in \cite[\S8.4]{lapidusfrank}, the Cantor string is defined slightly differently, and hence, $C$ is replaced by $3^{-1}C$.}

Finally, observe that, in agreement with the lattice case of the general theory of self-similar strings developed in \cite[Chapters~2--3 and \S8.4]{lapidusfrank12}, we can rewrite the pointwise fractal tube formula \eqref{CC_compute} as follows (with $D:=\dim_BC=\log_23$):
\begin{equation}\label{5.5.10.3/4}
t^{-(1-D)}V_{C,(0,1)}(t)=t^{-(1-D)}|C_t\cap(0,1)|=G\left(\log_{3}(2t)^{-1}\right)+o(1),
\end{equation}
where $G$ is given by \eqref{5.5.10.1/4}.
Therefore, since $G$ is periodic and nonconstant, it is clear that $t^{-(1-D)}V_{C,(0,1)}(t)$ cannot have a limit as $t\to 0^+$.
It also follows from the above formula \eqref{5.5.10.3/4} that the Cantor string RFD $(C,(0,1))$ (or, equivalently, the Cantor string ${\mathcal L}_{CS}$) is {\em not} Minkowski measurable but (since $G$ is also bounded away from zero and infinity) is Minkowski nondegenerate.
(This was first proved in [LapPo1--2] 
via a direct computation, leading to the precise values of $\mathcal{M}_*$ and $\mathcal{M^*}$, and reproved in \cite[\S8.4.2]{lapidusfrank12} by using either the pointwise fractal tube formulas or a self-similar fractal string analog of the Minkowski measurability criterion; i.e., of the $N=1$ case of Theorem \ref{criterion}; see Remark \ref{5.4.19.1/2} and, especially, Corollary \ref{5.4.22.1/4C}.)
\end{example}

The above example demonstrates how the theory developed in this chapter generalizes (to arbitrary dimensions $N\geq 1$) the corresponding one for fractal strings developed in~\cite[Chapter~8]{lapidusfrank12}.\footnote{One should slightly qualify this statement, however, because the higher-dimensional counterpart of the theory of fractal tube formulas for self-similar strings developed in \cite[\S8.4]{lapidusfrank12} is not developed here in the general case of self-similar RFDs (and, for example, of self-similar sets satisfying the open set condition), except in the special case of self-similar sprays discussed in \S\ref{subsec_self_similar_sp} below.}
More generally, the following result gives a general connection between the geometric zeta function of a nontrivial fractal string $\mathcal{L}=(\ell_j)_{j\geq 1}$ and the (relative) distance zeta function of the bounded subset of $\eR$ given by
\begin{equation}
A_{\mathcal{L}}:=\left\{a_k:=\sum_{j\geq k}\ell_j:k\geq 1\right\}
\end{equation}
or, more specifically, of the RFD $(A_{\mathcal L},(0,\ell))$.
For the proof of this proposition, we refer the reader to \cite[Proposition 6.3]{ftf_A} or \cite[Proposition 5.5.4]{fzf}.

\begin{proposition}\label{geo_dist}
Let $\mathcal{L}=(\ell_j)_{j\geq 1}$ be a bounded nontrivial fractal string and let $\ell:=\zeta_{\mathcal{L}}(1)=\sum_{j=1}^{\ty}\ell_j$ denote its total length.
Then, for every $\delta\geq \ell_1/2$, we have the following functional equation for the distance zeta function of the relative fractal drum $(A_{\mathcal{L}},(0,\ell))$$:$
\begin{equation}\label{geo_equ}
\zeta_{A_{\mathcal{L}},(0,\ell)}(s;\d)=\frac{2^{1-s}\zeta_{\mathcal{L}}(s)}{s},
\end{equation}
valid on any connected open neighborhood $U\subseteq\Ce$ of the critical line $\{\re s=\ov{\dim}_B(A_{\mathcal{L}},(0,\ell))\}$ to which any of the two fractal zeta functions $\zeta_{{A_{\mathcal{L}}},(0,\ell)}$ and $\zeta_{\mathcal{L}}$ possesses a meromorphic continuation.\footnote{If we do not require that $\delta\geq \ell_1/2$, then we have that $\zeta_{A_{\mathcal{L}}}(s;\d)={2^{1-s}s^{-1}\zeta_{\mathcal{L}}(s)}+v(s)$, where $v$ is holomorphic on $\{\re s>0\}$.
On the other hand, in order to apply the theory, we may restrict ourselves to the case when $\delta\geq \ell_1/2$.}

Furthermore, if $\zeta_{\mathcal{L}}$ is languid for some $\kappa_{\mathcal{L}}\in\eR$, then $\zeta_{A_{\mathcal{L}},(0,\ell)}(\,\cdot\,;\d)$ is $d$-languid for $\kappa_d=\kappa_{\mathcal{L}}-1$, with any $\d\geq \ell_1/2$.
%
\end{proposition}

We note that in the above proposition, $(A_{\mathcal{L}},(0,\ell))$ can be replaced by any geometric realization $(\pa\O,\O)$ of the fractal string $\mathcal{L}$.

\begin{example}\label{ex_a}({\em The $a$-string}).
For a given $a>0$, the $a$-string $\mathcal{L}_{a}$ can be realized as the bounded open set $\O_a\subset\eR$ obtained by removing the points $j^{-a}$ for $j\in\eN$ from the interval $(0,1)$; that is, 
\begin{equation}\label{O_a}
\O_a=\bigcup_{j=1}^{\ty}\big((j+1)^{-a},j^{-a}\big),
\end{equation}
so that the sequence of lengths of $\mathcal{L}_a$ is defined by
\begin{equation}\label{lj}
\ell_j=j^{-a}-(j+1)^{-a},\ \textrm{ for }\ j=1,2,\ldots,
\end{equation}
and $\pa\O_a=\{j^{-a}\,:\,j\geq 1\}\cup\{0\}=A_{\mathcal{L}_a}\cup\{0\}$.
Hence, its geometric zeta function is given (for all $s\in\Ce$ such that $\re s>\dim_B\mathcal{L}_a$) by
$$
\zeta_{{\mathcal{L}}_a}(s)=\sum_{j=1}^{\ty}\ell_j^s=\sum_{j=1}^{\ty}\big(j^{-a}-(j+1)^{-a}\big)^s
$$
and it then follows from Proposition \ref{geo_dist} that for $\d>(1-2^{-a})/2$, its distance zeta function is given by
\begin{equation}\label{a-dist}
\zeta_{A_{\mathcal{L}_a,(0,1)}}(s;\d)=\frac{\zeta_{\mathcal{L}_a}(s)}{2^{s-1}s}=\frac{1}{2^{s-1}s}\sum_{j=1}^{\ty}\big(j^{-a}-(j+1)^{-a}\big)^s,
\end{equation}
where the second equality holds for all $s\in\Ce$ such that $\re s>\dim_B\mathcal{L}_a$ while the first equality holds for all $s\in\Ce$ (since, as will be recalled just below, $\zeta_{\mathcal{L}_a}$ and hence also $\zeta_{A_{\mathcal{L}_a},(0,1)}$, admits a meromorphic extension to all of $\Ce$).

Furthermore, the properties of the geometric zeta function $\zeta_{\mathcal{L}_a}$ of the $a$-string are well-known (see \cite[Theorem~6.21]{lapidusfrank12}).
Namely, $\zeta_{\mathcal{L}_a}$ has a meromorphic continuation to the whole of $\Ce$ and its poles in $\Ce$ are located at 
\begin{equation}
D:=\dim_B\mathcal{L}_a=\dim_BA_{\mathcal{L}_a}=\frac{1}{a+1}
\end{equation}
and at (a subset of) $\{-\frac{m}{a+1}:m\in\eN\}$.
Moreover, all of its poles are simple and we have that $\res(\zeta_{\mathcal{L}_a},D)=Da^D$.\footnote{In \cite[Theorem~6.21]{lapidusfrank12}, it is stated that $\res(\zeta_{\mathcal{L}_a},D)=a^D$, which is a misprint.
More specifically, in the proof of that theorem, the source of the misprint is the fact that the residue of $\zeta((a+1)s)$ at $s=1/(a+1)$ is equal to $1/(a+1)$ and not to $1$. Here, $\zeta$ is the Riemann zeta function.\label{25}}
In addition, for any screen $\bm{S}$ not passing through a pole, the function $\zeta_{\mathcal{L}}$ satisfies {\bf L1} and {\bf L2} with $\kappa=\frac{1}{2}-(a+1)\inf S$, if $\inf S\leq 0$ and $\kappa=\frac{1}{2}$ if $\inf S\geq 0$.
From these facts and Equation \eqref{a-dist}, we conclude that the set $A_{\mathcal{L}_a}$ is $d$-languid with $\kappa_d=-\frac{1}{2}-(a+1)\inf S$ if $\inf S\leq 0$ and with $\kappa_d=-\frac{1}{2}$ if $\inf S\geq 0$.

By choosing the screen $\bm{S}$ to be some vertical line $\{\re s=\sigma\}$ for any value of $\sigma\in(-1/(a+1),1/(a+1))$, we conclude that the assumptions of Theorem \ref{criterion} are satisfied. 
Therefore, since $D=1/(a+1)$ is the only pole on the critical line $\{\re s=D\}$ and, additionally, since it is simple, we conclude that the $a$-string is Minkowski measurable and
%
%
with Minkowski content given by
\begin{equation}
\mathcal{M}^D(A_{\mathcal{L}_a})=\frac{2^{1-D}a^D}{1-D},
\end{equation}
as was first established in \cite[Example~5.1]{Lap1} and later reproved in [LapPo1--2] 
via a general Minkowski measurability  criterion for fractal strings (expressed in terms of the asymptotic behavior of $(\ell_j)_{j=1}^{\ty}$, here, $\ell_j\sim aj^{-1/D}$ as $j\to\ty$) and then, in [Lap--vFr1--3] 
(via the the theory of complex dimensions of fractal strings, specifically, via the special case of Theorem~\ref{criterion} when $N=1$).
\end{example}


\begin{remark}
References on or related to fractal strings include [CaLapPe-vFr1--2, Fal2, DemDenKo\"U, DemKo\"O\"U, DubSep, ElLapMacRo, EsLi1--2, Fal2, Fr, HamLap, HeLap, HerLap1--3, KeKom, KomPeWi, LaLap1--2, Lap1--7, LapL\'eRo, LapLu, LapLu-vFr1--2, LapMa, LapPe1--2, LapPeWi1--2, LapPo1--2, LapRa\v Zu1,8, Lap-vFr1--2, L\'eMen, MorSep, MorSepVi, Ol1--2, Pe, PeWi, Ra1, Tep1--2], with applications to (or motivations from) a variety of subjects, such as number theory, fractal geometry, dynamical systems, harmonic analysis, spectral theory and mathematical physics.
\end{remark}

\subsection{The Sierpi\'nski gasket and $3$-carpet}\label{subsec_sier}

In this subsection, we provide an exact, pointwise fractal tube formula for the Sierpi\'nski gasket (Example \ref{gsk_fract}) and for a three-dimensional analog of the Sierpi\'nski carpet (Example \ref{ex2}).

We leave it to the interested reader to carry out the corresponding detailed computations and to imagine other (two- or higher-dimensional) examples of self-similar fractal sets or self-similar RFDs which can be dealt with explicitly within the present general theory of (higher-dimensional) fractal tube formulas.
The example of the $3$-carpet discussed in detail in Example \ref{ex2} below should give a good idea as to how  to proceed in other, related situations, including especially for the higher-dimensional inhomogeneous $N$-gasket RFDs (with $N\geq 3$), which is treated in [LapRa\v Zu1, Example 4.2.26].

\begin{example}({\em The Sierpi\'nski gasket}).\label{gsk_fract}
Let $A$ be the Sierpi\'nski gasket in $\eR^2$, constructed in the usual way inside the unit triangle.
Furthermore, we assume without loss of generality that $\d>1/4\sqrt{3}$, so that $A_\d$ be simply connected.
Then, the distance zeta function $\zeta_A$ of the Sierpi\'nski gasket is meromorphic on the whole complex plane and is given by
\begin{equation}\label{sierp}
\zeta_A(s;\d)=\frac{6(\sqrt3)^{1-s}2^{-s}}{s(s-1)(2^s-3)}+2\pi\frac{\d^s}s+3\frac{\d^{s-1}}{s-1},
\end{equation}
for all $s\in\Ce$ (see \cite[Example 4.12]{mefzf} and also \cite[Proposition 3.2.3]{fzf}). In particular, the set of complex dimensions of the Sierpi\'nski gasket is given by
\begin{equation}
\po({\zeta}_{A}):=\po(\zeta_A,\Ce)=\{0,1\}\cup\left(\log_23+\frac{2\pi}{\log 2}\I\Ze\right),
\end{equation}
with each complex dimension being simple.
It is now easy to check that $\zeta_A$ is $d$-languid for any screen $\bm{S}$ chosen to be some vertical line ${\re s=\sigma}$, where $\sigma\in(1,\log_23)$.
Therefore, the assumptions of Theorem \ref{criterion} are satisfied and since, clearly, $D$ is not the only pole on the critical line $\{\re s=D\}$, we conclude that $A$ is not Minkowski measurable, a fact which is well known.
%
%
%
Furthermore, from the fractal tube formula for $A$ which was obtained in \cite{lappe2} and \cite{lappewi1} and, more recently, via a different (but related) technique in~\cite{DeKoOzUr}, but also now follows from our general fractal tube formulas in [LapRa\v Zu1,8], one can see that $A$ is actually Minkowski nondegenerate.
\end{example}

%


The scaling properties of fractal zeta functions are an important tool for the efficient computation of distance zeta functions of various compact sets and relative fractal drums, especially those possessing a self-similar structure, like the Cantor set or the Sierpi\'nski gasket and the corresponding RFDs. 

\begin{theorem} {\rm (Scaling property of relative distance zeta functions; \cite[Theorem~4.1.40]{fzf})}.\label{scalingp}
Let $\zeta_{A,\O}(s)$ be the relative distance zeta function. Then, for any positive real number $\g$, we have that $D(\zeta_{\g A,\g\O})=D(\zeta_{A,\O})=\ov\dim_B(A,\O)$ and
\begin{equation}
\zeta_{\g A,\g\O}(s)=\g^s\zeta_{A,\O}(s).
\end{equation}
\end{theorem}

We illustrate our method in the case of the Sierpi\'nski gasket. More specifically, we derive Equation~\eqref{sierp} by using a natural relative fractal drum associated with the gasket; see Case $(b)$ just below, especially Equations~\eqref{rfdeq}--\eqref{zetaeqs}.

\setlength{\unitlength}{.6mm}
\begin{figure}[ht]
\begin{center}
\begin{pspicture}(-4,0)(4,4.3)
\psset{algebraic,unit=.75cm}
\psline(-3,0)(3,0)(0,5.196)(-3,0)
\psline(0,0)(1.5,2.598)(-1.5,2.598)(0,0)

\psline(1.5,0)(0.75,1.299)(2.25,1.299)(1.5,0)
\psline(-1.5,0)(-0.75,1.299)(-2.25,1.299)(-1.5,0)
\psline(0,2.598)(-0.75,3.897)(0.75,3.897)(0,2.598)

\end{pspicture}%
\end{center}
\caption{The Sierpi\'nski gasket $A$, where only the first two generations of the triangles are shown.}\label{01sierpinski}
\end{figure}
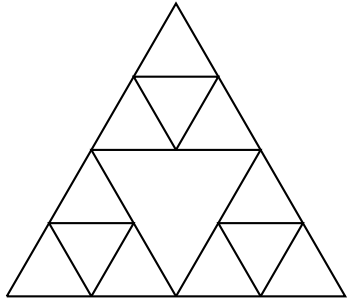

We show that the distance zeta function of the Sierpi\'nski gasket (see Figure~\ref{01sierpinski}) is given by Equation~\eqref{sierp},
where $\d$ is a positive real number larger than $1/4\sqrt3$, so that the $\d$-neighborhood of $A$ be simply connected.
We first consider the $\d$-neighborhood of the Sierpi\'nski gasket $A$, as shown in Figure~\ref{01sierpd}. We proceed in the following two steps, provided in Case~$(a)$ and Case~$(b)$ below.
\medskip

\setlength{\unitlength}{.6mm}
\begin{figure}[ht]
\begin{center}
\begin{pspicture}(-4,-1.2)(4,5)
\psset{algebraic,unit=.75cm}
\psline*[linecolor=lightgray](-3,0)(3,0)(0,5.196)(-3,0)
\psline*[linecolor=lightgray](-3,0)(3,0)(3,-1.3)(-3,-1.3)(-3,0)
\psline*[linecolor=lightgray](-3,0)(0,5.196)(-1.126,5.846)(-4.123,0.65)(-3,0)
\psline*[linecolor=lightgray](3,0)(0,5.196)(1.126,5.846)(4.123,0.65)(3,0)

\put(-3,0){\pscircle*[linecolor=lightgray]{1.3}}
\put(3,0){\pscircle*[linecolor=lightgray]{1.3}}
\put(0,5.196){\pscircle*[linecolor=lightgray]{1.3}}

\psline(-3,0)(3,0)(0,5.196)(-3,0)
\psline(0,0)(1.5,2.598)(-1.5,2.598)(0,0)

\psline(1.5,0)(0.75,1.299)(2.25,1.299)(1.5,0)
\psline(-1.5,0)(-0.75,1.299)(-2.25,1.299)(-1.5,0)
\psline(0,2.598)(-0.75,3.897)(0.75,3.897)(0,2.598)

\end{pspicture}%
\end{center}
\caption{The Sierpi\'nski gasket $A$ and its $\d$-neighborhood indicated in gray color, with $\d>1/4\sqrt3$, so that the $\d$-neighborhood be simply connected.}\label{01sierpd}
\end{figure}
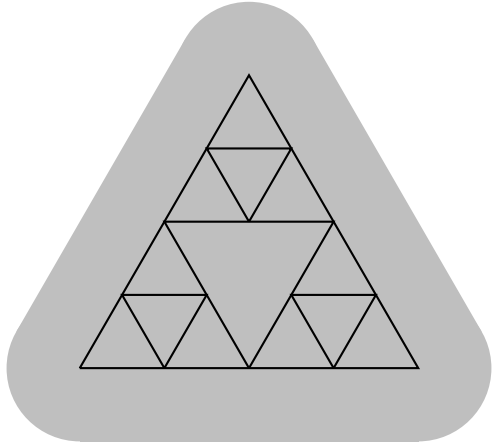

{\em Case $(a)$.} Let us first consider the part of the neighborhood $A_\d$  outside of the triangle, which is an easier case. It consists of three congruent sectorial sets near the vertices of the outer triangle of the gasket, and of three congruent rectangles; see Figure~\ref{01sierpds}. Note that if $x$ is any point in that set, then the distance $d(x,A)$ is equal to the distance $d(x,x')$, where $x'$ is either one of the vertices of the unit triangle $\O$, if $x$ belongs to any of the sectorial sets, and $x'$ is the orthogonal projection of $x$ onto one of the sides of the rectangles, if $x$ is in any of the rectangles.

\setlength{\unitlength}{.6mm}
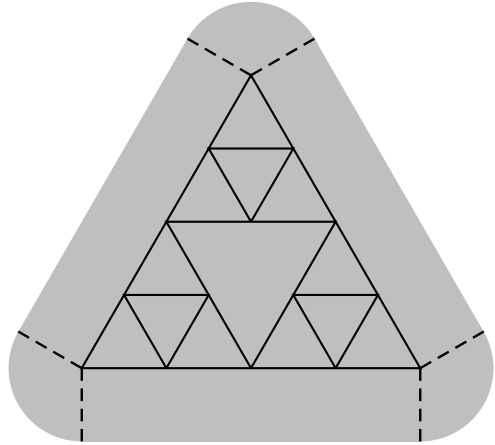
\begin{figure}[ht]
\begin{center}
\begin{pspicture}(-4,-1.2)(4,5)
\psset{algebraic,unit=.75cm}
\psline*[linecolor=lightgray](-3,0)(3,0)(0,5.196)(-3,0)
\psline*[linecolor=lightgray](-3,0)(3,0)(3,-1.3)(-3,-1.3)(-3,0)
\psline*[linecolor=lightgray](-3,0)(0,5.196)(-1.126,5.846)(-4.123,0.65)(-3,0)
\psline*[linecolor=lightgray](3,0)(0,5.196)(1.126,5.846)(4.123,0.65)(3,0)

\put(-3,0){\pscircle*[linecolor=lightgray]{1.3}}
\put(3,0){\pscircle*[linecolor=lightgray]{1.3}}
\put(0,5.196){\pscircle*[linecolor=lightgray]{1.3}}

\psline[linestyle=dashed,linewidth=0.9pt](-3,0)(-3,-1.3)
\psline[linestyle=dashed,linewidth=0.9pt](3,0)(3,-1.3)

\psline[linestyle=dashed,linewidth=0.9pt](-3,0)(-4.126,0.65)
\psline[linestyle=dashed,linewidth=0.9pt](0,5.196)(-1.126,5.846)

\psline[linestyle=dashed,linewidth=0.9pt](3,0)(4.126,0.65)
\psline[linestyle=dashed,linewidth=0.9pt](0,5.196)(1.126,5.846)

\psline(-3,0)(3,0)(0,5.196)(-3,0)
\psline(0,0)(1.5,2.598)(-1.5,2.598)(0,0)

\psline(1.5,0)(0.75,1.299)(2.25,1.299)(1.5,0)
\psline(-1.5,0)(-0.75,1.299)(-2.25,1.299)(-1.5,0)
\psline(0,2.598)(-0.75,3.897)(0.75,3.897)(0,2.598)

\end{pspicture}%
\end{center}
\caption{In Case $(a)$, we consider the set $A_\d\setminus \O$, by splitting it into six parts, as indicated by dashed lines. In Case $(b)$, we consider the RFD $(A,\O)$.}\label{01sierpds}
\end{figure}

The disjoint union of the three sectorial RFDs can be identified with an RFD consisting of a point and an open disk of radius~$\d$ around it, as shown in Figure~\ref{01sierpsect} on the right. Its relative distance zeta function is easy to compute in polar coordinates:
\begin{equation}\label{0Bd0}
\zeta_{\{0\},B_\d(0)}=\iint_{B_\d(0)}r^{s-2}r\,\D r\,\D\f=\int_0^{2\pi}\D \f\int_0^\d r^{s-1}\D r=2\pi\frac{\d^s}s
\end{equation}
(where $B_\d(0)$ is the unit open disk of radius $\d$, centered at the origin),
and it coincides with the second term on the right-hand side of Equation~\eqref{sierp}.

\setlength{\unitlength}{.6mm}
\begin{figure}[ht]
\begin{center}%
\begin{minipage}{0.5\textwidth}
\begin{pspicture}(-4.2,-1)(2,1.3)
\psset{algebraic,unit=.75cm}

\put(0,0){\pscircle*[linecolor=lightgray]{1.3}}
\put(0,0){\circle*{0.2}}

\psline[linestyle=dashed,linewidth=0.9pt](0,0)(0,-1.3)
\psline[linestyle=dashed,linewidth=0.9pt](0,0)(1.126,0.65)
\psline[linestyle=dashed,linewidth=0.9pt](0,0)(-1.126,0.65)

\end{pspicture}%
\end{minipage}%
\begin{minipage}{0.5\textwidth}
\begin{pspicture}(-2.5,-1)(2,1.3)
\psset{algebraic,unit=.75cm}

\put(0,0){\pscircle*[linecolor=lightgray]{1.3}}

\put(0,0){\circle*{0.2}}

\end{pspicture}%
\end{minipage}
\end{center}
\caption{On the left is depicted the disjoint union of three sectorial parts of $\d$-neighborhood of the Sierpi\'nski gasket from the preceding Figure~\ref{01sierpds}. It can be viewed as an RFD represented on the right, consisting of a point and of the open disk of radius $\d$ arround it.}\label{01sierpsect}
\end{figure}
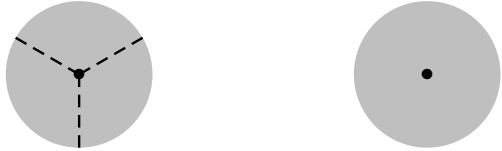

Furthermore, we have three congruent copies of rectangular RFDs, and each of them can be identified with the one shown in Figure~\ref{01rectangle}, consisting of a closed interval of length~$1$ and of the open rectangle of height~$\d$.

\setlength{\unitlength}{.6mm}
\begin{figure}[ht]
\begin{center}
\begin{pspicture}(-4,0)(4,1.3)
\psset{algebraic,unit=.75cm}
\psline*[linecolor=lightgray](-3,0)(3,0)(3,1.3)(-3,1.3)(-3,0)
\psline[linewidth=1pt](-3,0)(3,0)
\put(3.2,1){\small $(I,R)$}
\end{pspicture}%
\end{center}
\caption{In Figure \ref{01sierpds} above, we can see three copies of rectangular RFDs, each of them congruent to the RFD $(I,R)$, where $I=[0,1]\times \{0\}$ and $R=(0,1)\times(0,\d)$, as represented in the present figure.}\label{01rectangle}
\end{figure}
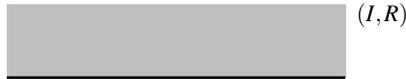

The corresponding relative zeta function is given by
\begin{equation}\label{IR}
3\zeta_{I,R}(s)=3\int_0^1\D x\int_0^\d y^{s-2}\D y=3\frac{\d^{s-1}}{s-1},
\end{equation}
which coincides with the last term on the right-hand side of Equation~\eqref{sierp}.
\medskip

{\em Case $(b)$.}\label{caseb} Here, we consider the most interesting object: an RFD $(A,\O)$ consiting of the Sierpi\'nski gasket $A$ and of the open set $\O$ which is just the interior of the convex hull of $A$, as shown in Figure~\ref{01serpinskirfd}.

\setlength{\unitlength}{.6mm}
\begin{figure}[ht]
\begin{center}
\begin{pspicture}(-4,0)(4,4.3)
\psset{algebraic,unit=.75cm}
\psline*[linecolor=lightgray](-3,0)(3,0)(0,5.196)(-3,0)
\psline(-3,0)(3,0)(0,5.196)(-3,0)
\psline(0,0)(1.5,2.598)(-1.5,2.598)(0,0)

\psline(1.5,0)(0.75,1.299)(2.25,1.299)(1.5,0)
\psline(-1.5,0)(-0.75,1.299)(-2.25,1.299)(-1.5,0)
\psline(0,2.598)(-0.75,3.897)(0.75,3.897)(0,2.598)

\end{pspicture}%
\end{center}
\caption{The Sierpi\'nski gasket RFD $(A,\O)$.}\label{01serpinskirfd}
\end{figure}
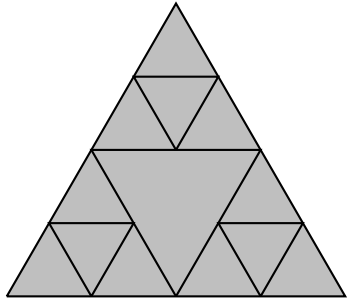

The Sierpi\'nski gasket $A$ is self-similar; more specifically, it is equal to the disjoint union (denoted by $\sqcup$) of three of its scaled translates, using the scaling factor $\g=1/2$.
Therefore, we can write
\begin{equation}\label{rfdeq}
(A,\O)=\frac12(A,\O)\,\sqcup\,\frac12(A,\O)\,\sqcup\,\frac12(A,\O)\,\sqcup\,(\pa \O_0,\O_0),
\end{equation}
where $\O_0$ is the inner open isosceles triangle, with sides of lengths $1/2$; i.e., the first deleted open triangle during the construction of the Sierpi\'nski gasket.
From Equation~\eqref{rfdeq}, by using the additivity and scaling properties of relative distance zeta functions (see, in particular, Theorem~\ref{scalingp} above), we deduce that
\begin{equation}\label{zetaeq}
\begin{aligned}
\zeta_{A,\O}(s)&=3\zeta_{\frac12(A,\O)}(s)+\zeta_{\pa\O_0,\O_0}(s)\\
&=3\cdot\frac1{2^s}\zeta_{A,\O}(s)+\zeta_{\pa \O_0,\O_0}(s),
\end{aligned}
\end{equation}
and hence,
\begin{equation}\label{zetaeqs}
\zeta_{A,\O}(s)=\frac{2^s}{2^s-3}\cdot \zeta_{\pa\O_0,\O_0}(s),
\end{equation}
for all $s$ with $\re s>\log_23$.

It remains to compute $\zeta_{\pa \O_0,\O_0}(s)$. The triangle RFD $(\pa\O_0,\O_0)$ can be split into six congruent parts (more precisely, into six congruent RFDs), as shown in Figure~\ref{6triangles}.

\setlength{\unitlength}{.6mm}
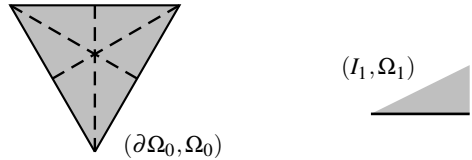
\begin{figure}[ht]
\begin{center}
\begin{minipage}{0.5\textwidth}
\begin{pspicture}(-4,0)(1,2)
\psset{algebraic,unit=.75cm}
\psline*[linecolor=lightgray](0,0)(1.5,2.598)(-1.5,2.598)(0,0)
\psline(0,0)(1.5,2.598)(-1.5,2.598)(0,0)
\psline[linestyle=dashed,linewidth=0.9pt](0,0)(0,2.598)
\psline[linestyle=dashed,linewidth=0.9pt](1.5,2.598)(-0.75,1.299)
\psline[linestyle=dashed,linewidth=0.9pt](-1.5,2.598)(0.75,1.299)

\put(0.5,0){\small$(\pa\O_0,\O_0)$}
\end{pspicture}%
\end{minipage}%
\begin{minipage}{0.5\textwidth}
\begin{pspicture}(-2,-1)(1,2)
\psset{algebraic,unit=.75cm}
\psline*[linecolor=lightgray](-0.875,0)(0.875,0.866)(0.875,0)(-0.875,0)
\psline[linewidth=1pt](-0.875,0)(0.875,0)
\put(-1.4,0.7){\small$(I_1,\O_1)$}

\end{pspicture}%
\end{minipage}
\end{center}
\vskip-2mm
\caption{On the left is the triangle RFD $(\pa\O_0,\O_0)$, where $\O_0$ is the isosceles triangle with sides of lengths $1/2$, decomposed into six congruent RFDs, indicated by dashed vertices.
On the right is the triangle RFD $(I_1,\O_1)$ consisting of the side $I_1$ of length $1/4$ and of the open right-angle triangle $\O_1$ with angles $30^\circ$ and $60^\circ$.}\label{6triangles}
\end{figure}

By the additivity property of the relative zeta function, we have that $\zeta_{\pa\O_0,\O_0}(s)=6\zeta_{I_1,\O_1}(s)$, and $\zeta_{I_1,\O_1}(s)$ is very easy to compute by placing the origin of the plane at the left vertex of the triangle. We have (see Figure~\ref{6triangles}) that
\begin{equation}\label{I1O1}
\begin{aligned}
\zeta_{\pa\O_0,\O_0}(s)&=6\zeta_{I_1,\O_1}(s)=6\iint_{\O_1}y^{s-2}\D x\,\D y\\
&=6\int_0^{1/4} \D x\int_0^{x/\sqrt3}y^{s-2}\,\D y,
\end{aligned}
\end{equation}
where the last integral is easy to compute. Substituting into Equation~\eqref{zetaeqs}, we deduce that $\zeta_{A,\O}(s)$ coincides with the first term on the righ-hand side of Equation~\eqref{sierp}. This completes the verification of the closed formula~\eqref{sierp} for the distance zeta function of the Sierpi\'nski gasket.

We therefore obtain formula~\eqref{sierp}, at first valid for all $s\in\Ce$ with $\re s$ large enough. However, in light of the principle of analytic continuation (i.e., here, meromorphic continuation), we deduce that Equation~\eqref{sierp} holds for all~$s\in\Ce$.

It is interesting to note that the computation of the distance zeta function of the Sierpi\'nski gasket was reduced to the computation of the distance zeta functions of three very simple RFDs: $(\{0\},B_\d(0))$ appearing in Equation~\eqref{0Bd0}, $(I,R)$ in Equation~\eqref{IR} and $(I_1,\O_1)$ in Equation~\eqref{I1O1}. A list containing about two dozens of basic RFDs, with the corresponding distance zeta functions and principal complex dimensions, can be found in~\cite[Appendix~C]{fzf}.

\medskip


\begin{example}({\em The $3$-carpet}).\label{ex2}
Let $A$ be a three-dimensional analog of the Sierpi\'nski carpet.
More specifically, we construct $A$ by dividing the closed unit cube of $\eR^3$ into $27$ congruent cubes and remove the open middle cube.
Then, we iterate this step with each of the $26$ remaining smaller closed cubes; and so on, ad infinitum.
By choosing $\d>1/6$, we have that $A_\d$ is simply connected.

From \cite[Example 6.9]{ftf_A} (see also \cite[Example 5.5.13]{fzf}), we have that $\zeta_A$ is meromorphic on all of $\Ce$ and given by
\begin{equation}
\zeta_A(s):=\zeta_A(s,\d)=\frac{48\cdot 2^{-s}}{s(s-1)(s-2)(3^s-26)}+\frac{4\pi\d^s}{s}+\frac{6\pi\d^{s-1}}{s-1}+\frac{6\d^{s-2}}{s-2},
\end{equation}
for every $s\in\Ce$.

It follows that the set of complex dimensions of the $3$-carpet $A$ is given by
\begin{equation}\label{3-carp_po}
\po(\zeta_A):=\po({\zeta}_A,\Ce)=\{0,1,2\}\cup\big(\log_326+\mathbf{p}\I\Ze\big),
\end{equation}
where $D:=\log_326\ (=D(\zeta_A))$ is the Minkowski (or box) dimension of the $3$-carpet $A$ and $\mathbf{p}:=2\pi/\log 3$ is the oscillatory period of $A$ (viewed as a lattice self-similar set).
In \eqref{3-carp_po}, each of the complex dimensions is simple.

In particular, we conclude that $D:=\dim_BA=\log_326$ (as was noted before) and, by Theorem\ \ref{criterion}, that the three-dimensional Sierpi\'nski carpet is not Minkowski measurable, which is expected (see \cite{Lap3}).
%
%
%
\end{example}

\subsection{Fractal nests and unbounded geometric chirps}\label{subsec_nestch}

In this subsection, we apply our Minkowski measurability criterion to examples of fractal nests (Example \ref{ex_nest}) and of (unbounded) geometric chirps (Example \ref{unb_chirp}).
Both of these families of sets are examples of ``fractal'' sets which are {\em not} self-similar or, more generally, `self-alike' in any sense.
By carefully examining the example of the fractal nest we acquire new insights into the situation when the fractal zeta function has a pole of second order at $s=D$, where $D$ is the Minkowski dimension of the bounded set or RFD under consideration.
This situation will be further investigated in \S\ref{gauge_mink} below where we will obtain some general results about the gauge-Minkowski measurability (for a specific family of gauge functions) in terms of the presence of complex dimensions of higher order on the critical line.

\begin{example}({\em Fractal nests}).\label{ex_nest}
We let $\mathcal{L}=(\ell_j)_{j\geq 1}$ be a bounded fractal string and, as before, let $A_{\mathcal{L}}=\{a_k:k\in\eN\}\subset\eR$, with $a_k:=\sum_{j\geq k}\ell_j$ for each $k\geq 1$.
Furthermore, consider now $A_{\mathcal{L}}$ as a subset of the (horizontal) $x_1$-axis in $\eR^2$ and let $A$ be the planar set obtained by rotating $A_{\mathcal{L}}$ around the origin; i.e., $A$ is a union of concentric circles of radii $a_k$ and center at the origin.
For $\d>\ell_1/2$, the distance zeta function of $A$ is then given by
\begin{equation}
\zeta_{A}(s)=\frac{2^{2-s}\pi}{s-1}\sum_{j=1}^{\ty}\ell_j^{s-1}(a_j+a_{j+1})+\frac{2\pi\d^{s}}{s}+\frac{2\pi a_1\d^{s-1}}{s-1};
\end{equation}
see \cite[Example 3.5.1]{fzf}.
The last two terms in the above formula correspond to the annulus $a_1<r<a_1+\d$ and we will neglect them; that is, we will consider only the relative distance zeta function $\zeta_{A,\O}$, with $\O:=B_{a_1}(0)$.\footnote{Here, $B_r(x)$ denotes the open ball of radius $r$ with center at $x$.}

\medskip

We will now consider a special case of the fractal nest above; that is, the relative fractal drum $(A_a,\O)$ corresponding to the $a$-string $\mathcal{L}:=\mathcal{L}_a$, with $a>0$; so that for all $j\ge1$, we have $\ell_j:=j^{-a}-(j+1)^{-a}$ and hence, $a_j=j^{-a}$.
In this case, it is not difficult to see that
\begin{equation}\label{zeta_nest}
\zeta_{A_a,\O}(s)=\frac{2^{3-s}\pi}{s-1}\sum_{j=1}^{\ty}j^{-a}\ell_j^{s-1}-\frac{2^{2-s}\pi}{s-1}\zeta_{\mathcal{L}}(s),
\end{equation}
where $\zeta_{\mathcal{L}}$ is the geometric zeta function of the $a$-string.
(See \cite[Example 6.12]{ftf_A} for details.)

Furthermore, the above zeta function given by \eqref{zeta_nest} has been analyzed in \cite[Example 6.12]{ftf_A} and shown to possess a meromorphic continuation to all of $\Ce$. Furthermore for $a\neq 1$ it was shown that the set of complex dimensions of $(A_a,\O)$ satisfies the following inclusion
\begin{equation}\label{pozeta}
\begin{aligned}
\po({\zeta}_{A_a,\O})&:=\po({\zeta}_{A_a,\O},\Ce)\subseteq\left\{1,\frac{2}{a+1},\frac{1}{a+1}\right\}\cup\left\{-\frac{m}{a+1}:m\in\eN\right\}.
\end{aligned}
\end{equation}
Moreover, all of the above (potential) complex dimension are simple and we are certain that ${2}/{(a+1)}$ is always a complex dimension of $({A_a},\O)$, i.e., by letting $D:={2}/{(a+1)}$, we have for all positive $a\neq 1$ that
\begin{equation}
\res\left({\zeta}_{A_a,\O},D\right)=\frac{2^{2-D}D\pi}{D-1}a^{D-1}.
\end{equation}
It was also shown in \cite{ftf_A} that $(A_a,\O)$ is $d$-languid for any screen chosen to be a vertical line not passing through the (potential) poles of ${\zeta}_{A_a,\O}$ given on the right-hand side of~\eqref{pozeta}.

It is not difficult to show directly that for $a>0$ the Minkowski dimension $\dim_B(A_a,\O)$ exists and is equal to $\max\{1,2/(a+1)\}$.  
We now conclude from Theorem~\ref{criterion} that if $a\in(0,1)$, $\dim_B({A_a},\O)=D(\zeta_{A_a,\O})=D$ and $(A_a,\O)$ is Minkowski measurable with Minkowski content given by
\begin{equation}
\mathcal{M}^D({A_a},\O)=\frac{\res\left({\zeta}_{A_a,\O},D\right)}{2-D}=\frac{2^{2-D}D\pi}{(2-D)(D-1)}a^{D-1}.
\end{equation}

Furthermore, if $a>1$, we have that $\dim_B({A_a},\O)=1$ and it also follows from Theorem \ref{criterion} that $({A_a},\O)$ is Minkowwski measurable with
\begin{equation}
\mathcal{M}^1({A_a},\O)=\frac{\res\left({\zeta}_{A_a,\O},1\right)}{2-1}=4\pi\zeta(a)-2\pi,
\end{equation}
where $\zeta$ is the Riemann zeta function.
Note also that since $\zeta(a)>1$ for $a>1$, we have $$2\pi<\mathcal{M}^1({A_a},\O)<\ty.$$
We stress here that the computation of $\res\left({\zeta}_{A_a,\O},1\right)$ above is not trivial and we refer the interested reader to \cite[Example 6.12]{ftf_A} for the details.

We now turn our attention to the critical case when $a=1$.
In this case, we have that $s=1$ is a pole of second order (i.e., of multiplicity two) of ${\zeta}_{A_1,\O}$ since it is a simple pole of $\zeta_1(s):=\sum_{j=1}^{\ty}j^{-a}\ell_j^{s-1}$ appearing in \eqref{zeta_nest}.

By results from \cite{ftf_A} (see, especially \cite[Example 6.12]{ftf_A}), we can obtain in this particular case the following fractal tube formula for $(A_1,\O)$:
\begin{equation}
\begin{aligned}
|(A_1)_t\cap\O|&=2\pi t\log t^{-1}+\mathrm{const.}\cdot t+o(t)\quad \textrm{as}\ t\to0^+.
\end{aligned}
\end{equation}

As we can clearly see, $\dim_B(A_1,\O)=1$; furthermore, the RFD $(A_1,\O)$ is Minkowski degenerate with $\mathcal{M}^1(A_1,\O)=+\ty$.
On the other hand, if we consider the notion of gauge-Minkowski measurability with the gauge function $h(t):=\log t^{-1}$ defined for all $t\in(0,1)$, we have that $(A_1,\O)$ is gauge-Minkowski measurable, meaning, more specifically, that its generalized Minkowski content exists and is given by
\begin{equation}
\mathcal{M}^1(A_1,\O,h):=\lim_{t\to0^+}\frac{|(A_1)_t\cap\O|}{t^{2-1}h(t)}=2\pi.
\end{equation}

We will rigorously define the notion of gauge-Minkowski measurability in \S\ref{gauge_mink} and obtain a result (see Theorems 5.4 and 5.9) explaining how this notion is connected to the presence of complex dimensions of higher order on the critical line.
Here, we have considered the motivating example of $(A_1,\O)$ for which the asymptotics of its tube function do not obey a power law but, nevertheless, are such that this property is reflected in the nature of its distance zeta function; i.e., in the fact that it possesses a (unique) pole of order two on the critical line.
One may also view this situation as a kind of ``merging'' of two simple complex dimensions of $(A_a,\O)$ when $a\neq 1$, namely, $1$ and $2/(a+1)$, into one complex dimension of order two as $a\to 1$. 
\end{example}

Next, we consider an example of an unbounded geometric chirp, considered as a relative fractal drum $(A,\O)$ where both of the sets $A$ and $\O$ are unbounded.
A standard {\em geometric $(\a,\b)$-chirp}, with positive parameters $\a$ and $\b$, is a simple geometric approximation of the graph of the chirp function $f(x)=x^{\a}\sin(\pi x^{-\b})$, for all $x\in(0,1)$.
(See \cite[Example 4.4.1 and Proposition 4.4.2]{fzf}.)

If we choose parameters $-1<\a<0<\beta$, we obtain an example of an unbounded chirp function $f$ which we can approximate by the unbounded geometric $(\a,\b)$-chirp.

\begin{example}({\em Unbounded geometric chirps}).\label{unb_chirp}
Let $A_{\a,\b}$ be the union of vertical segments with abscissae $x=j^{-1/\b}$ and of lengths $j^{-\a/\b}$, for every $j\in\eN$.
Furthermore, define $\O$ as a union of the rectangles $R_j$ for $j\in\eN$, where $R_j$ has a base of length $j^{-1/\b}-(j+1)^{-1/\b}$ and height $j^{-\a/\b}$.
The relative distance zeta function of $(A,\O)$ is computed in \cite[Example 4.4.1]{fzf} and is given by
\begin{equation}\label{zetaab}
\begin{aligned}
\zeta_{A_{\a,\b},\O}(s)&=\frac{2^{2-s}}{s-1}\sum_{j=1}^{\ty}j^{-\a/\b}\left(j^{-1/\b}-(j+1)^{-1/\b}\right)^{s-1}.
\end{aligned}
\end{equation}
From \cite[Example 6.15]{ftf_A} we have that $\zeta_{A_{\a,\b},\O}$ has a meromorphic continuation to all of $\Ce$ and
\begin{equation}
\po({\zeta}_{A_{\a,\b},\O})\subseteq\left\{1,2-\frac{1+\a}{1+\b}\right\}\cup\left\{D_m:m\in\eN\right\},
\end{equation}
where $D_m:=2-\frac{1+\a+m\b}{1+\b}$ and all of the (potential) complex dimensions are simple.
Furthermore, by letting $D:=2-\frac{1+\a}{1+\b}>1$ it can be shown directly that $\dim_B(A,\O)=D$ and also, it is always true that $1$ and $D$ are always complex dimensions of $(A,\O)$, i.e., $1,D\in\po({\zeta}_{A_{\a,\b},\O})$.
Moreover, it follows from \cite{ftf_A} that ${\zeta}_{A_{\a,\b},\O}$ is $d$-languid for any vertical line through $\{\re s=\sigma\}$ with $\sigma\in(1,D)$.
Therefore, the assumptions of Theorem \ref{criterion} are satisfied and since $D$ is the only pole on the critical line $\{\re s=D\}$ and is simple, we conclude that the RFD $(A_{\a,\b},\O)$ is Minkowski measurable.

Also, by \cite[Example 6.15]{ftf_A} and in light of \eqref{zetaab}, we have that both $1$ and $D$ are simple poles of $\zeta_{A_{\a,\b}}$.
Furthermore, we have that $D>1$ and, consequently, $\dim_B(A_{\a,\b},\O)=D$ with
\begin{equation}
\begin{aligned}
\mathcal{M}^D(A_{\a,\b},\O)&=\frac{\res({\zeta}_{A_{\a,\b},\O},D)}{2-D}=\frac{2^{2-D}}{(2-D)(D-1)}\frac{\b^{\frac{1+\a}{1+\b}}}{1+\b}\\
&=\frac{(2\b)^{2-D}}{(2-D)(D-1)(1+\b)}.
\end{aligned}
\end{equation}
\end{example}

\subsection{Minkowski measurability criteria for self-similar sprays}\label{subsec_self_similar_sp}

We conclude this section by showing how the results of the present chapter may also be applied to recover and significantly extend, as well as place within a general conceptual framework, the tube formulas for  self-similar sprays generated by an arbitrary open set $G\subset\eR^N$ of finite $N$-dimensional Lebesgue measure.
(See, especially, [LapPe2--3] 
extended to a significantly more general setting in \cite{lappewi1}, along with the exposition of those results given in \cite[\S13.1]{lapidusfrank12}; see also \cite{DeKoOzUr} for an alternative, but related, proof of some of those results.)

Let us recall that a self-similar spray (with a single generator $G$, assumed to be bounded and open) is defined as a collection $(G_k)_{k\in\eN}$ of pairwise disjoint (bounded) open sets $G_k\subset\eR^N$, with $G_0:=G$ and such that for each $k\in\eN$, $G_k$ is a scaled copy of $G$ by some factor $\lambda_k>0$.
(We let $\lambda_0:=1$.)
Then, the associated scaling sequence $(\lambda_k)_{k\in\eN}$ is obtained from a ratio list $\{r_1,r_2,\ldots,r_J\}$, with $0<r_j<1$ for each $j=1,\ldots,J$ and such that $\sum_{j=1}^{J}r_j^N<1$, by considering all possible words built out of the scaling ratios $r_j$.
Here, $J\geq 2$ and the scaling ratios $r_1,\ldots,r_J$ are repeated according to their multiplicities

We now assume that $(A,\O)$ is the self-similar spray considered as a relative fractal drum  and defined as $A:=\partial(\sqcup_{k=0}^{\ty} G_k)$ and $\O:=\sqcup_{k=0}^{\ty} G_k$, with $\ov{\dim}_B(\partial G,G)<N$.
The associated distance zeta function $\zeta_{A,\O}$ of these kinds of relative fractal drums satisfies an important {\em factorization formula} (see \cite[Theorem 3.36]{refds} or \cite[Theorem 4.2.17]{fzf}), expressed in terms of the distance zeta function of the boundary of the generator (relative to the generator), $\zeta_{\pa G,G}$, and the scaling ratios $\{r_j\}_{j=1}^{J}$:
\begin{equation}\label{5.5.102..}
\zeta_{A,\O}(s)=\frac{\zeta_{\partial G,G}(s)}{1-\sum_{j=1}^{J}r_j^s}.
\end{equation}

\begin{remark}\label{5.5.20.1/2R}
Note that we can rewrite Equation \eqref{5.5.102..} as follows:
\begin{equation}\label{5.5.105.1/2E}
\zeta_{A,\O}(s)=\zeta_{\mathfrak{S}}(s)\cdot\zeta_{\pa G,G}(s),
\end{equation}
where the geometric zeta function $\zeta_{\mathfrak{S}}$ of the associated self-similar string (with scaling ratios $\{r_j\}_{j=1}^J$ and a single gap length, equal to one, in the terminology of \cite[Chapters 2 and 3]{lapidusfrank12}) is meromorphic in all of $\Ce$ and given for all $s\in\Ce$ by
\begin{equation}\label{5.5.105.3/4E}
\zeta_{\mathfrak{S}}(s)=\frac{1}{1-\sum_{j=1}^Jr_j^s}.
\end{equation}
In light of the above discussion, we see that generally, given a connected open set $U\subseteq\Ce$, $\zeta_{A,\O}$ is meromorphic in $U$ if and only if $\zeta_{\pa G,G}$ is.
Also, in that case, the {\em factorization formula} \ref{5.5.105.1/2E} then holds for all $s\in U$.
Furthermore, in the sequel and following [LapPe2--3], 
[LapPeWi1--2] 
and [LapRa\v Zu1,8], we will often refer to $\zeta_{\mathfrak{S}}$ as the {\em scaling zeta function} of the self-similar spray $(A,\O)$ and to its poles in $\Ce$ (composing the multiset $\mathfrak{D}$)\label{matfrD} as the {\em scaling complex dimensions} of $(A,\O)$; the latter are the solutions (counting multiplicities) of the {\em complexified Moran equation} $\sum_{j=1}^Jr_j^s=1$.
We will also sometimes write $\mathfrak{D}_{\mathfrak{S}}$ instead of $\mathfrak{D}$, so that $\mathfrak{D}_{\mathfrak{S}}:=\mathfrak{D}$; hence, similarly, $\mathfrak{D}_{\mathfrak{S}}\cap \bm{W}=\mathfrak{D}\cap \bm{W}$, the set of {\em visible scaling complex dimensions} of $(A,\O)$, denotes the set of poles of $\zeta_{\mathfrak{S}}$ visible through the window~$\bm{W}$.
\end{remark}

Let us define the {\em similarity dimension} of $(A,\O)$ as the unique real solution of the Moran equation $\sum_{j=1}^Jr_j^s=1$ and denote it by $\sigma_0$.
Furthermore, observe that from \eqref{5.5.102..} we always have that 
\begin{equation}\label{max_sigma}
\ov{\dim}_B(A,\O)=D(\zeta_{A,\O})=\max\{\sigma_0,D(\zeta_{\pa G,G})\},
\end{equation}
where the last equality follows from the factorization formula \eqref{5.5.20.1/2R}.
Let us now assume that $\dim_B(A,\O)$ exists and also that the generator $G$ is such that $D_G:=\dim_B(\pa G,G)=D(\zeta_{\pa G,G})$ exists and is a simple pole of $\zeta_{\pa G,G}$.
Additionally, assume that $\zeta_{\pa G,G}$ is $d$-languid for some window as is the case when $G$ is, for instance, monophase or pluriphase in the sense of \cite{lappe1,lappe2}
and [LapPeWi1--2].
(Note that a large class of polytopes is monophase, as was conjectured in \cite{lappe1} and later proved in \cite{KoRati}.)
Observe that $D_G\in[N-1,N]$ since $\pa G$ is the boundary of a bounded open set in $\eR^N$.
Namely, consider the orthogonal projection $\pi(\pa G)$ of $\pa G$ onto $\eR^{N-1}$.
Since $G$ is bounded and open, then $\pi(G)\subseteq\pi(\pa G)$ and hence, $\pi(\pa G)_\d\cap\pi(G)=\pi(G)$ for every $\d>0$.
From the fact that $\pi(G)$ is bounded and open in $\eR^{N-1}$, we conclude that $\dim_B(\pa G,G)\geq\dim_B(\pi(\pa G),\pi(G))=N-1$.  

Now, the following three cases naturally arise: $(i)$ $D_G<\sigma_0$; $(ii)$ $D_G=\sigma_0$, and $(iii)$ $D_G>\sigma_0$.

\medskip

{\em Case} $(i)$:\ \ $D_G<\sigma_0$.
Then, by \eqref{max_sigma}, $D=\sigma_0$ and all of the poles of $\zeta_{\pa G,G}$ are located strictly to the left of $D$. 
Therefore, the factorization formula \eqref{5.5.105.1/2E} implies that the principal complex dimensions of $(A,\O)$ coincide with the complex solutions of the complexified Moran equation $\sum_{j=1}^{J}r_j^s=1$ (i.e., with the scaling complex dimensions of $(A,\O)$).
These solutions were extensively investigated in \cite{lapidusfrank12} and here we recall from \cite[Theoem~3.6]{lapidusfrank12} that $\sigma_0$ is always a simple pole of $\zeta_{\mathfrak{S}}$ and that, in the nonlattice case, it is the only pole of $\zeta_{\mathfrak{S}}$ on the vertical line $\{\re s=\sigma_0\}$.
On the other hand, in the lattice case, the poles of $\zeta_{\mathfrak{S}}$ form an infinite subset of $\sigma_0+\mathbf{p}\I\Ze$, where $\mathbf{p}=2\pi/\log(r^{-1})$ is the {\em oscillatory period}, with $r\in(0,1)$ being the single generator of the multiplicative group (of rank $1$) generated by the distinct values of the scaling ratios $r_1,\ldots,r_J$.\footnote{Note that the lattice-nonlattice dichotomy of self-similar sprays or sets is defined in terms of the nature of the multiplicative group generated by the distinct values of the scaling ratios $r_1,\ldots,r_J$. In the lattice case the generating set consists of exactly one element, while in the nonlattice case it has more than one element.}

In particular, since $D$ is simple, it follows that the RFD $(A,\O)$ has a nonreal complex dimension with real part $D\ (=\sigma_0)$ if and only if we are in the lattice case, and hence, in light of Theorems \ref{mink_char} and \ref{criterion}, if and only if $(A,\O)$ is Minkowski measurable.
More specifically, we reason exactly as in the proof of Corollary \ref{5.4.22.1/4C} (which corresponds to the case when $N=1$).
Namely, if $(A,\O)$ is lattice, then it satisfies the hypotheses of Theorem \ref{necess} concerning the languidity and the screen.
Furthermore, in this case we can choose the screen to be strictly to the right of all of the poles of $\zeta_{\pa G,G}$ and also strictly to the right of all of the poles of $\zeta_{\mathfrak{S}}$ having real part strictly less than $\sigma_0$.
A result from \cite{lapidusfrank12} implies that $\zeta_{\mathfrak{S}}$ satisfies the strong $d$-languidity conditions (after a possible scaling) and hence, by the factorization formula \eqref{5.5.105.1/2E}, this is also true for $\zeta_{A,\O}$ and our chosen screen (under the assumption of $d$-languidity we made about $\zeta_{\pa G,G}$ above).

Consequently, since in the lattice case, in addition to $D$ there are other poles with real part $D$, we conclude that $(A,\O)$ cannot be Minkowski measurable.
On the other hand, if $(A,\O)$ is nonlattice, then $D$ is the only pole with real part $D$ and $(A,\O)$ satisfies the hypotheses of Theorem \ref{mink_char}, which implies that $(A,\O)$ is Minkowski measurable.\footnote{For a detailed analysis of the structure of the scaling complex dimensions in the lattice and nonlattice cases, see \cite{lapidusfrank12}, Chapters 2 and 3, especially, Theorem 2.16 and Theorem 3.6.
Note that in \cite[Chapter 3]{lapidusfrank12}, no assumption is made about the underlying scaling ratios (and gaps), so that the corresponding results (and hence also, \cite[Theorems 2.16 and 3.6]{lapidusfrank12}) can be applied to the scaling complex dimensions of  self-similar sprays in $\eR^N$, for an arbitrary dimension $N\geq 1$.}

{\em Therefore, in case $(i)$, $(A,\O)$ is Minkowski measurable if and only if $D$ is its only principal complex dimension and also, if and only if the self-similar spray $(A,\O)$ is nonlattice.}

This proves (for the case of self-similar sprays) the geometric part of a conjecture of the first author in \cite[Conjecture~3, p.\ 163]{Lap3}, in case $(i)$.
Note that in the case of self-similar strings (i.e., when $N=1$), we are always in case $(i)$ and therefore, we have reproved the characterization of the Minkowski measurability for self-similar strings obtained in \cite[\S8.4, esp., Theorems 8.23 and 8.36]{lapidusfrank12}.

\medskip

{\em Case} $(ii)$:\ \ $D_G=\sigma_0$.
In this case, $D_G$ and $\sigma_0$ are simple poles of $\zeta_{\pa G,G}$ and $\zeta_{\mathfrak{S}}$, respectively.
Hence, from the factorization formula (3.39), $\zeta_{A,\O}=\zeta_{\mathfrak{S}}\cdot\zeta_{\pa G,G}$, we conclude that $D$ is a double pole of $\zeta_{A,\O}$.
Hence, it follows directly from Theorem \ref{pole1} that $(A,\O)$ {\em is not Minkowski measurable}, regardless of whether or not the self-similar spray is lattice or nonlattice.

\medskip

{\em Case} $(iii)$:\ \ $D_G>\sigma_0$.
In light of \eqref{max_sigma}, we then have $D=D_G$.
Now, according to \cite[Theorem~1.10]{lapidusfrank12}, all of the poles of $\zeta_{\mathfrak{S}}$ have real part $\leq\sigma_0$ and thus, have real part strictly less than $D$.
Furthermore, the factorization formula \eqref{5.5.102..} implies that the principal complex dimension of $(A,\O)$  depend only on the generating relative fractal drum $(\pa G,G)$.
In most cases, we then have that $\dim_{PC}(A,\O)=\{D\}$ and that $D=D_G=N-1$ is simple (since $D_G$ is a simple pole of  $\zeta_{\pa G,G}$).
However, we emphasize here that other cases are theoretically possible, depending on the fractal properties of the boundary $\pa G$.
It is an interesting (but difficult) open problem to investigate all the possible cases of the distributions of the complex dimensions for relative fractal drums of the type $(\pa G,G)$, where $G$ is a bounded open set in $\eR^N$.
On the other hand, in the case when $D=D_G$ is the only principal complex dimension and is simple, we conclude from Theorem \ref{criterion} (the Minkowski measurability criterion) that the RFD $(A,\O)$ is Minkowski measurable in both the lattice and nonlattice cases.

\medskip

To sumarize, in case $(i)$, the RFD $(A,\O)$ is Minkowski measurable if and only if it is nonlattice, implying that the geometric part of Conjecture $3$ of \cite[p.\ 163]{Lap3} is true in this case.
(See the comments below about the results of \cite{KomPeWi}.)

Furthermore, in case $(ii)$, when $D_G$ is a simple (or possibly even multiple) pole of $\zeta_{\pa G,G}$, then $\sigma_0=D_G\ (=D)$; hence, $(A,\O)$ is not Minkowski measurable since $D$ is then a pole of $\zeta_{A,\O}$ which is (at least) of second order. 

Finally, in case $(iii)$, i.e., when $\sigma_0<D_G\ (=D)$, in general, the Minkowski measurability of $(A,\O)$ will depend only on the Minkowski measurability of the generating RFD $(\pa G,G)$.
Clearly, in light of the above, the conclusion of the analog of the geometric part of \cite[Conjecture 3]{Lap3} fails when we are not in case $(i)$.

For each of the cases $(i)$--$(iii)$, there is a natural realization for general self-similar sprays (or RFDs), as will be illustrated in Examples \ref{1/2-tube_formula} and \ref{1/3-tube_formula} below.
On the other hand, for classical self similar sets $F$ (satisfying the open set condition), which corresponds to the original conjecture (in \cite[Conjecture 3]{Lap3}), only cases $(i)$ or $(ii)$ can occur.
Namely, it is well-known that for such sets, we have $\dim_BF=\sigma_0$.

Recall that the geometric part of \cite[Conjecture 3]{Lap3} has been proved for self-similar sets (rather than for general self-similar sprays or RFDs) satisfying the open set condition, first when $N=1$ in [Lap--vFr1--3] 
(see \cite[\S8.4]{lapidusfrank12} and the earlier books [Lap--vFr1--2]) 
(by using the fractal tube formulas for fractal strings) and then, when $N\geq 1$ in \cite{KomPeWi} (by using the renewal theorem), thereby extending a variety of results previously obtained in \cite{Lap3} (when $N=1$ and by using the renewal theorem), in \cite{Fal2} (also when $N=1$ and by using the renewal theorem) and then, in \cite{gatzouras} (when $N\geq 2$, and also by using the renewal theorem) and in [Lap--vFr1--3] 
(as mentioned above), as well as later, under some relatively restrictive hypotheses, for self-similar sprays in \cite{lappewi2} (when $N\geq 1$ is arbitrary and by using the fractal tube formulas for self-similar sprays of \cite{lappewi1} along with techniques from \cite{lapidusfrank12}).
What was missing to the results of \cite{Lap3,Fal2,gatzouras} (but not of [Lap--vFr1--3] 
and of \cite{lappewi2}) was to show that lattice self-similar sets are {\em not} Minkowski measurable (as was the content of part of the conjecture of \cite{Lap3}), which is now known to be true when $D$ is not an integer.
We point out that case $(ii)$ was also considered in \cite{KomPeWi}, with the same conclusion as above.

Finally, we expect that the above results about self-similar sprays in cases $(i)$ and $(ii)$ can also be proved by analogous methods for self-similar sets in $\eR^N$ (satisfying the open set condition), by finding a suitable functional equation that connects the fractal zeta functions of the self-similar set and the associated self-similar tiling (or spray).

\medskip

The next three examples illustrate interesting phenomena that may occur in the setting of self-similar sprays (or RFDs).
We point out that these examples can also be viewed as inhomogeneous self-similar sets (in the sense of [BarDemk] and \cite{Hat}).
Furthermore, we will also reflect on the corresponding fractal tube formulas obtained in \cite{ftf_A} (see also \cite{fzf}).

\begin{figure}[ht]
\begin{center}
\includegraphics[width=5cm]{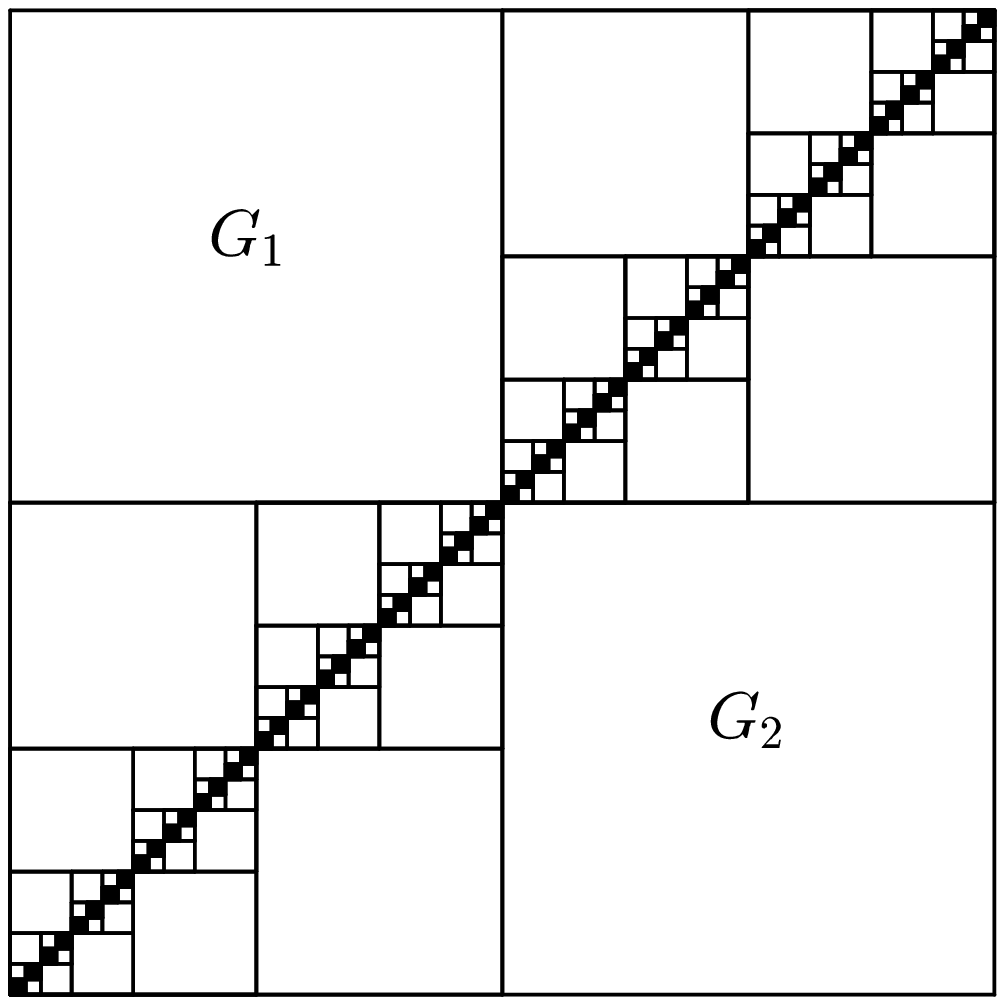}
\includegraphics[trim=0 0.8cm 0 0,clip,width=5.7cm]{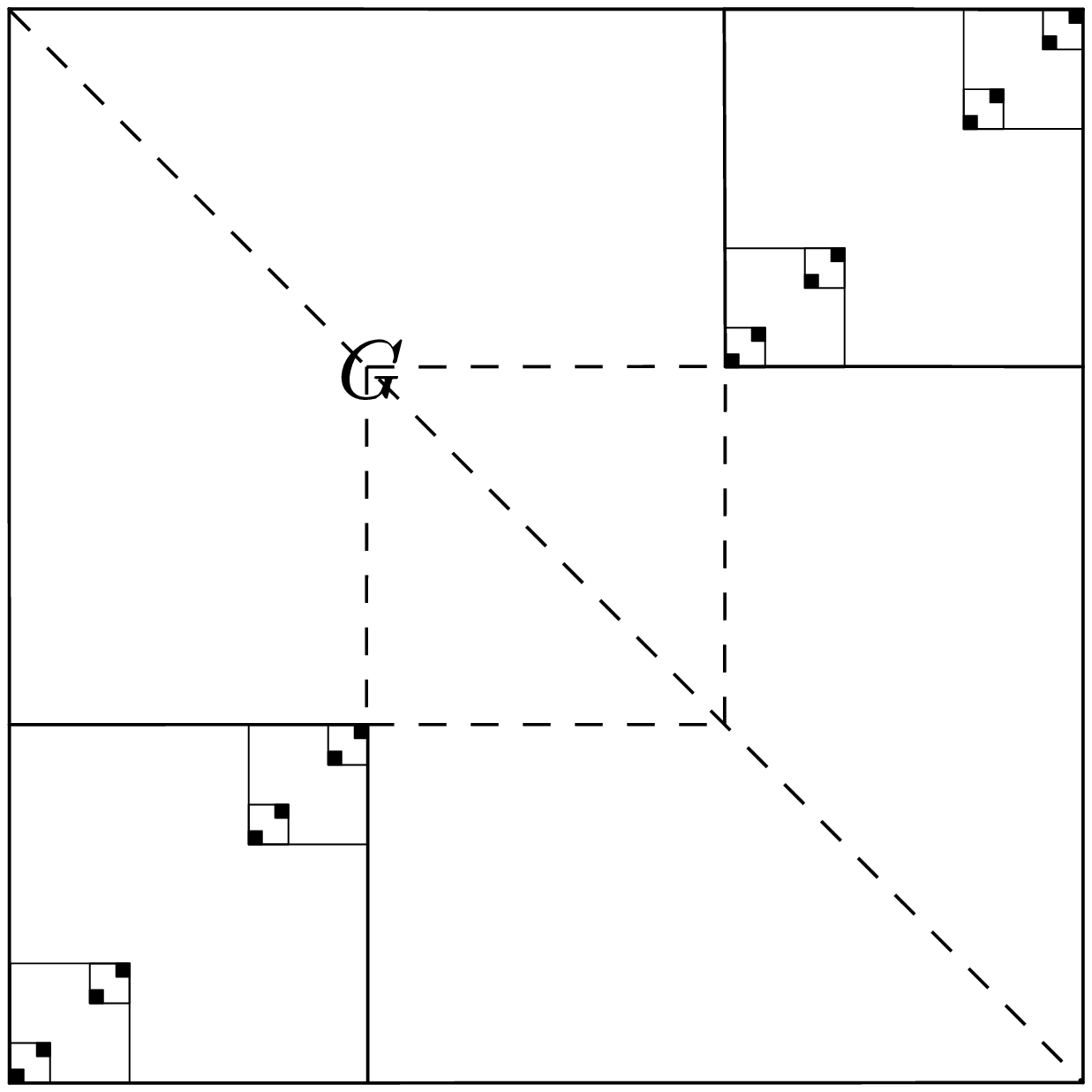}
\caption{{\bf Left:} The $1/2$-square fractal $A$ from Example \ref{1/2-tube_formula}. We start with a unit square $[0,1]^2$ and, in the first step, remove the open squares $G_1$ and $G_2$. In the next step, we repeat this construction with the remaining closed squares $[1/2,1]^2$ and $[0,1/2]^2$; we then continue this process ad infinitum. The $1/2$-{\em square fractal} $A$ is then the set which remains behind.
The first 6 iterations are shown.
Note that $G:=G_1\cup G_2$ is the single generator of the corresponding self-similar spray or RFD $(A,\O)$, where $\O=(0,1)^2$. 
{\bf Right:} The $1/3$-square fractal $A$ from Example \ref{1/3-tube_formula}. We start with a unit square $[0,1]^2$ and, in the first step, remove the open polygon $G$. In the next step, we repeat this construction with the remaining closed squares $[1/3,1]^2$ and $[0,1/3]^2$; 
we then continue this process ad infinitum. 
The $1/3$-{\em square fractal} $A$ is then the set which remains behind. 
The first 4 iterations are shown. 
Note that $G$ is the single generator of the corresponding self-similar spray or RFD $(A,\O)$, 
where $\O:=(0,1)^2$.
}\label{kv_0.5}
\end{center}
\end{figure}

\begin{example}({\em The $1/2$-square fractal}).\label{1/2-tube_formula}
Let us consider the $1/2$-square fractal $A$ from \cite[Example 3.38]{refds} and depicted in Figure \ref{kv_0.5}, left.
This is a case $(ii)$ example from the above discussion about self-similar sprays.
The distance zeta function of $A$ was obtained in \cite{refds}, where it was shown to be meromorphic on all of $\Ce$ and given by
\begin{equation}\label{dist_1/2_1}
\zeta_A(s)=\frac{2^{-s}}{s(s-1)(2^s-2)}+\frac{4}{s-1}+\frac{2\pi}{s},
\end{equation}
for every $s\in\Ce$.
Furthermore, as was discussed in \cite{refds}, it follows at once from \eqref{dist_1/2_1} that
\begin{equation}
\begin{aligned}
D(\zeta_A)=1,\q
\po(\zeta_A):=\po(\zeta_A,\Ce)=\{0\}\cup\left(1+\mathbf{p}\I\Ze\right)
\end{aligned}
\end{equation}
and 
\begin{equation}\label{PCsq1}
\dim_{PC} A:=\po_c(\zeta_A)=\{1\}.
\end{equation}
Here, the oscillatory period $\mathbf{p}$ of $A$ is given by $\mathbf{p}:={2\pi}/{\log 2}$ and all of the complex dimensions in $\po(\zeta_A)$ are simple, except for $\omega_0:=1$ which is a double pole of $\zeta_A$.
Theorem \ref{pole1} immediately implies that $A$ cannot be Minkowski measurable. 

We can now obtain the fractal tube formula for $A$ directly from its distance zeta function.
Namely, by \cite[Example 6.18]{ftf_A}, we have that for all $t\in(0,1/2)$,
\begin{equation}\label{1/2_tube_formula_eq}
\begin{aligned}
|A_t|&=\sum_{\omega\in\po({\zeta}_{A})}\res\left(\frac{t^{2-s}}{2-s}{\zeta}_A(s),\omega\right)\\
&=\frac{1}{4\log 2}t\log t^{-1}+t\,G\left(\log_2(4t)^{-1}\right)+\frac{1+2\pi}{2}t^2.
\end{aligned}
\end{equation}
Here, $G$ is a {\em nonconstant} $1$-periodic function on $\eR$ which is bounded away from zero and infinity and is given by the following convergent Fourier series:
\begin{equation}
G(x):=\frac{29\log 2-4}{8\log 2}+\frac{1}{4}\sum_{k\in\Ze\setminus\{0\}}\frac{\E^{2\pi\I kx}}{(2-\omega_k)(\omega_k-1)\omega_k},\quad\textrm{for all}\ x\in\eR, 
\end{equation}
where we have let $\omega_k:=1+\I\mathbf{p}k$ for each $k\in\Ze$.

In conclusion, it is now also clear from the fractal tube formula \eqref{1/2_tube_formula_eq} for the $1/2$-square fractal that $\dim_BA=1$ and that $A$ is actually Minkowski degenerate with $\mathcal{M}^1(A)=+\ty$.
Furthermore, it is also clear that $A$ is $h$-Minkowski measurable with $h(t):=\log t^{-1}$ (for all $t\in(0,1)$) and with $h$-Minkowski content given by
$
\mathcal{M}^1(A,h)=(4\log 2)^{-1}.
$\footnote{Recall that the notions of $h$-Minkowski measurability and $h$-Minkowski content will be precisely defined towards the beginning of $\S$\ref{gauge_mink}.}
Finally, we emphasize that although $D:=\dim_BA=1$ and hence, $A$ would not be considered fractal in the classical sense, we can also see from \eqref{1/2_tube_formula_eq} that the nonreal complex dimensions of $A$ with real part equal to $D$ give rise to (intrinsic) geometric oscillations of order $O(t^{2-D})$ in its fractal tube formula.
Hence, according to our new proposed definition of fractality given in \cite{fzf,ftf_A}, $A$ is {\em critically fractal} in dimension $d:=D=\dim_BA=1$.
\end{example}

\begin{example}({The $1/3$-{\em square fractal}}).\label{1/3-tube_formula}
We next consider the $1/3$-square fractal $A$ from \cite[Example 3.39]{refds} and illustrated in Figure \ref{kv_0.5}, right.
This is a case $(iii)$ example of self-similar sprays discussed above since here the similarity dimension $\sigma_0=\log_32$ is strictly less than the dimension $D_G=1$ of the generating relative fractal drum $(\pa G,G)$.
The distance zeta function of $A$ was obtained in \cite{refds}, where it was shown to be meromorphic on all of $\Ce$ and given by
\begin{equation}\label{zeta_1/3_square_2}
\zeta_{A}(s)=\frac{2}{s(3^s-2)}\left(\frac{6}{s-1}+Z(s)\right)+\frac{4}{s-1}+\frac{2\pi}{s},
\end{equation}
for all $s\in\Ce$.
The function $Z(s)$ is entire and given by
$
Z(s):=\int_0^{\pi/2}(\cos\varphi+\sin\varphi)^{-s}\di\varphi.
$
Furthermore, we have that
$
D(\zeta_A)=1
$
and
\begin{equation}\label{c_dim_A1/3_2}
\po(\zeta_A):=\po(\zeta_A,\Ce)\subseteq\{0\}\cup\left(\log_32+\mathbf{p}\I\Ze\right)\cup\{1\}.
\end{equation}
Here, the oscillatory period $\mathbf{p}$ of $A$ is given by $\mathbf{p}:={2\pi}/{\log 3}$ and all of the complex dimensions in $\po(\zeta_A)$ are simple.
In Equation \eqref{c_dim_A1/3_2}, we have only an inclusion since, in theory, some of the complex dimensions with real part $\log_32$ may be canceled by the zeros of the term $6/(s-1)+Z(s)$.
However, in addition to $D=\log_32$, there are nonreal complex dimensions with real part $\log_32$ in $\po(\zeta_A)$ which are not canceled out.
This fact can be checked numerically and we conjecture that, in fact, we have an equality in \eqref{c_dim_A1/3_2}.
Furthermore, it now follows from Theorem \ref{mink_char} below that $A$ is Minkowski measurable with $\mathcal{M}^1(A)=\res(\zeta_A,1)=16$.

We can next state the fractal tube formula of $A$ from its distance zeta function given in \eqref{zeta_1/3_square_2}.
This formula was obtained in \cite[Example 4.2.34]{ftf_A} and is given as follows:
\begin{equation}\label{racun_1/3}
\begin{aligned}
|A_t|&=\sum_{\omega\in\po({\zeta}_{A})}\res\left(\frac{t^{2-s}}{2-s}{\zeta}_A(s),\omega\right)\\
&=16t+t^{2-\log_32}G\left(\log_3(3t)^{-1}\right)+\frac{12+\pi}{2}t^2,
\end{aligned}
\end{equation}
valid for every $t\in(0,1/\sqrt{2})$.
Here, the function $G$ is a {\em nonconstant} $1$-periodic function on $\eR$, which is bounded and is given by the following Fourier series:
\begin{equation}
G(x):=\frac{1}{\log 3}\sum_{k=-\ty}^{+\ty}\frac{\E^{2\pi\I kx}}{(2-\omega_k)\omega_k}\left(\frac{6}{\omega_k-1}+Z(\omega_k)\right),\quad\textrm{for all}\ x\in\eR, 
\end{equation}
where we have let $\omega_k:=\log_32+\I\mathbf{p}k$ for each $k\in\Ze$.

In conclusion, we can see from the fractal tube formula \eqref{racun_1/3} that $\dim_BA=1$ and $A$ is Minkowski measurable, with Minkowski content given by
$
\mathcal{M}^1(A)=16.
$

Finally, we emphasize that since $D:=\dim_BA=1$, the set $A$ would not be considered fractal in the classical sense.
On the other hand, it is clear from \eqref{racun_1/3} that the nonreal complex dimensions of $A$ with real part equal to $\log_32$ give rise to (intrinsic) geometric oscillations of order $O(t^{2-\log_32})$ in its fractal tube formula.
Again, according to our new proposed definition of fractality given in \cite{fzf,ftf_A}, the $1/3$-square fractal $A$ is fractal; more precisely, it is {\em strictly subcritically fractal} in dimension $d:=\log_32$.
\end{example}

\section{Sketch of the proof of the Minkowski measurability criterion}\label{main_proof}

In this section, we give a sketch of the proof of our main result, i.e., Theorem \ref{criterion_p}.
Recall that, under suitable hypotheses, an RFD with Minkowski dimension $D$ is Minkowski measurable if and only if its only complex dimension with real part $D$ is equal to $D$ itself, and $D$ is simple.
(See Theorems \ref{criterion} and \ref{tilde_criterion}, along with Remark \ref{ekv_mink_krt}.)
In Theorem \ref{mink_char} we obtain a sufficient condition (with weaker hypotheses imposed on the RFD in comparison to the Minkowski measurability criterion of Theorem \ref{criterion}) for the Minkowski measurability of a relative fractal drum.
Finally, we also establish an upper bound for the upper Minkowski content of an RFD in terms of the residue at $s=\ov{D}$ of its fractal zeta function, where $\ov{D}$ denotes the upper Minkowski dimension of the RFD; see Theorem \ref{mink_bound}.
Note that all of theses results also apply, in particular, to bounded subsets of $\eR^N$, with $N\geq 1$ arbitrary.

Complete proofs of the main results presented in this section can be found in \cite[Chapter 5, esp., $\S$5.4.1 and $\S$5.4.3]{fzf}.

\subsection{A sufficient condition for Minkowski measurability}

In this subsection, we obtain a sufficient condition for an RFD $(A,\O)$ to be Minkowski measurable
in terms of its relative tube (or distance) zeta function.
This theorem is a consequence of a well-known Tauberian theorem due to Wiener and Pitt (see~\cite{PitWie}) and which is a generalization of the famous Ikehara Tauberian theorem.
For the proof of the Wiener--Pitt Tauberian theorem, we also refer the interested reader to~\cite[Chapter~III, Lemma~9.1 and Proposition~4.3]{Kor} or to~\cite[\S6.1]{Pit} and to~\cite{Dia}, where a different proof using a technique of Bochner is given.
For the sake of completeness we state this theorem here.

\begin{theorem}[The Wiener--Pitt Tauberian theorem, {\rm cited from~\cite{Kor}}]\label{korevaar}
Let $\sigma\colon\eR\to\eR$ be such that $\sigma(t)$ vanishes for all $t<0$, is nonnegative for all $t\geq 0$, and such that its Laplace transform
\begin{equation}\label{laplace}
F(s):=\{\mathfrak{L}\sigma\}(s):=\int_0^{+\ty}\E^{-st}\sigma(t)\di t
\end{equation}
exists for all $s\in\Ce$ such that $\re s>0$.
Furthermore, suppose that for some constants $A>0$ and $\lambda>0$, the function
\begin{equation}
H(s)=F(s)-\frac{A}{s},\quad s=x+\I y,
\end{equation}
converges in $L^1(-\lambda,\lambda)$ to a boundary function $H(\I y)$ as $x\to0^+$.
Then, for every real number $h\geq 2\pi/\lambda$, we have that
\begin{equation}\label{statement_3}
\sigma_h(u):=\frac{1}{h}\int_u^{u+h}\sigma(t)\di t\leq CA+o(1)\quad\mathrm{as}\quad u\to +\ty,
\end{equation}
for some positive constant $C<3$.

Moreover, if the above constant $\lambda$ can be taken to be arbitrarily large, then for every $h>0$,
\begin{equation}\label{tauber_limit}
\sigma_h(u)\to A\quad\mathrm{as}\quad u\to +\ty.
\end{equation}
\end{theorem}

We are now ready to state and sketch a proof of the aforementioned sufficiency result.

\begin{theorem}[Sufficient condition for Minkowski measurability]\label{mink_char}
Let $(A,\O)$ be a relative fractal drum in $\eR^N$ and let $\ov{D}:=\ov{\dim}_B(A,\O)$.
Furthermore, suppose that the relative tube zeta function $\widetilde{\zeta}_{A,\O}$ of $(A,\O)$ can be meromorphically extended to a connected open neighborhood $U\subseteq\Ce$ of the critical line $\{\re s=\ov{D}\}$, with a single pole $\ov{D}$, which is assumed to be simple.
Then $D:=\dim_B(A,\O)$ exists, $D=\ov{D}$ and $(A,\O)$ is Minkowski measurable with Minkowski content given by
\begin{equation}
\M^{D}(A,\O)=\res(\widetilde{\zeta}_{A,\O},D).
\end{equation}

Moreover, if we assume, in addition, that $\ov{D}<N$, then the theorem is also valid if we replace the relative tube zeta function $\widetilde{\zeta}_{A,\O}$ by the relative distance zeta function ${\zeta}_{A,\O}$ of $(A,\O)$, and in that case, we have
\begin{equation}
\M^{D}(A,\O)=\frac{\res({\zeta}_{A,\O},D)}{N-D}.
\end{equation}
\end{theorem}

\begin{proof}
Without loss of generality, for the tube zeta function $\widetilde{\zeta}_{A,\O}(\,\cdot\,;\delta)$ we may choose $\delta=1$ and change the variable of integration by letting $u:=1/t$:
\begin{equation}
\begin{aligned}
\widetilde{\zeta}_{A,\O}(s+\ov{D})&=\int_0^{1}t^{s+\ov{D}-1-N}|A_t\cap\O|\di t\\
&=\int_1^{+\ty}u^{-s-1-\ov{D}+N}|A_{1/u}\cap\O|\di u\\
&=\int_0^{+\ty}\E^{-sv}\E^{v(N-\ov{D})}|A_{\E^{-v}}\cap\O|\di v=\{\mathfrak{L}\sigma\}(s),
\end{aligned}
\end{equation}
where we have made another change of variable in the second to last equality, namely, $v:=\log u$, and we have also let $\sigma(v):=\E^{v(N-\ov{D})}|A_{\E^{-v}}\cap\O|$.
It is clear from the definition of the relative tube zeta function of $(A,\O)$ that the residue of $\widetilde{\zeta}_{A,\O}(s)$ at $s={\ov{D}}$ must be real and positive.
Furthermore, since $s=\ov{D}$ is the only pole of $\widetilde{\zeta}_{A,\O}$ in $U$, we conclude that
\begin{equation}
H(s):=\widetilde{\zeta}_{A,\O}(s+\ov{D})-\frac{\res(\widetilde{\zeta}_{A,\O},\ov{D})}{s}
\end{equation}
is holomorphic in the neighborhood $U_{\ov{D}}:=\{s\in\Ce\,:\,s+\ov{D}\in U\}$ of the vertical line $\{\re s= 0\}$.
Hence, we can apply Theorem~\ref{korevaar} (for arbitrarily large $\lambda>0$, in the notation of that theorem) and conclude that
\begin{equation}
\sigma_h(u)=\frac{1}{h}\int_u^{u+h}\sigma(v)\di v\to\res(\widetilde{\zeta}_{A,\O},\ov{D})\quad\mathrm{as}\quad u\to +\ty,
\end{equation}
for every $h>0$.
In particular, since $v\mapsto|A_{\E^{-v}}\cap\O|$ is nonincreasing, we next consider the following two cases:

\medskip

{\em Case} $(a)$:
We assume that $\ov{D}<N$.
Hence, we have
$$
\begin{aligned}
\frac{1}{h}\int_u^{u+h}\E^{v(N-\ov{D})}|A_{\E^{-v}}\cap\O|\di v&\leq\frac{|A_{\E^{-u}}\cap\O|}{h}\int_u^{u+h}\E^{v(N-\ov{D})}\di v\\
&=\frac{|A_{\E^{-u}}\cap\O|}{\E^{-u(N-\ov{D})}}\,\frac{\E^{h(N-\ov{D})}-1}{(N-\ov{D})h}.
\end{aligned}
$$
By taking the lower limit of both sides as $u\to +\ty$, we obtain that
\begin{equation}
\res(\widetilde{\zeta}_{A,\O},\ov{D})\leq{\mathcal{M}}_{*}^{\ov{D}}(A,\O)\frac{\E^{h(N-\ov{D})}-1}{(N-\ov{D})h}.
\end{equation}
Since this is true for every $h>0$, we deduce by letting $h\to 0^+$ that
\begin{equation}\label{liminf_nej}
\res(\widetilde{\zeta}_{A,\O},\ov{D})\leq{\mathcal{M}}_*^{\ov{D}}(A,\O).
\end{equation}
On the other hand, we have
\begin{equation}\label{ocjena_integrala}
\begin{aligned}
\frac{1}{h}\int_u^{u+h}\E^{v(N-\ov{D})}|A_{\E^{-v}}\cap\O|\di v&\geq\frac{|A_{\E^{-(u+h)}}\cap\O|}{h}\int_u^{u+h}\E^{v(N-\ov{D})}\di v\\
&=\frac{|A_{\E^{-(u+h)}}\cap\O|}{\E^{-(u+h)(N-\ov{D})}}\,\frac{1-\E^{-h(N-\ov{D})}}{(N-\ov{D})h}
\end{aligned}
\end{equation}
and, similarly as before, by taking the upper limit of both sides as $u\to +\ty$, we obtain that
\begin{equation}
\res(\widetilde{\zeta}_{A,\O},\ov{D})\geq{\mathcal{M}}^{*\ov{D}}(A,\O)\frac{1-\E^{-h(N-\ov{D})}}{(N-\ov{D})h}.
\end{equation}
Since this is true for every $h>0$, we can let $h\to 0^+$ and deduce that
\begin{equation}
\res(\widetilde{\zeta}_{A,\O},\ov{D})\geq{\mathcal{M}}^{*\ov{D}}(A,\O).
\end{equation}
From this latter inequality, combined with~\eqref{liminf_nej}, we conclude that $(A,\O)$ is $\ov{D}$-Minkowski measurable.
Of course, this fact implies, a fortiori, that $D=\dim_B(A,\O)=\ov{D}$.
Furthermore, we also conclude that $\res(\widetilde{\zeta}_{A,\O},D)={\mathcal{M}}^D(A,\O)$, the Minkowski content of $(A,\O)$.

\medskip

{\em Case} $(b)$:
We will now assume that $\ov{D}=N$.
Therefore, in this case we have
$$
|A_{\E^{-(u+h)}}\cap\O|=\frac{|A_{\E^{-(u+h)}}\cap\O|}{\E^{-(u+h)(N-N)}}\leq\frac{1}{h}\int_u^{u+h}|A_{\E^{-v}}\cap\O|\di v\leq\frac{|A_{\E^{-u}}\cap\O|}{\E^{-u(N-N)}}=|A_{\E^{-u}}\cap\O|.
$$
Then, by taking, respectively, the lower and upper limits as $u\to +\ty$, we obtain that
\begin{equation}
{\mathcal{M}}^{*N}(A,\O)\leq\res(\widetilde{\zeta}_{A,\O},N)\leq{\mathcal{M}}_*^{N}(A,\O).
\end{equation}

Finally, if $D<N$, then the part of the theorem dealing with the distance (instead of the tube) zeta function of $(A,\O)$ follows at once from case $(a)$ of the proof for $\widetilde{\zeta}_{A,\O}$.
This fact is clear in light of the functional equation~\eqref{equ_tilde}, or more precisely, of the relation between the residues at $s=D$ of the two zeta functions which follows from it (namely, $\res({\zeta}_{A,\O},D)=(N-D)\res(\widetilde{\zeta}_{A,\O},D)$).
This concludes the proof of the theorem.
\end{proof}

\begin{remark}\label{integral_remark}
It is clear from Theorem~\ref{korevaar} that the assumptions of Theorem~\ref{mink_char} can be somewhat weakened.
We leave it to the interested reader to formulate and prove Theorem~\ref{mink_char} under these weaker assumptions.
\end{remark}

When besides $\ov{D}$, there are other singularities on the critical line $\{\re s=\ov{D}\}$ of the relative fractal drum $(A,\O)$, it is possible to use Theorem~\ref{korevaar} in order to derive a bound for the upper $\ov{D}$-dimensional Minkowski content of $(A,\O)$ in terms of the residue of its relative tube (or distance) zeta function at $s=\ov{D}$.

\begin{theorem}[Upper bound for the upper Minkowski content]\label{mink_bound}
Let $(A,\O)$ be a relative fractal drum in $\eR^N$ and let $\ov{D}:=\ov{\dim}_B(A,\O)$.
Furthermore, assume that the relative tube zeta function $\widetilde{\zeta}_{A,\O}$ of $(A,\O)$ can be meromorphically extended to a connected open neighborhood $U$ of the critical line $\{\re s=\ov{D}\}$ and that $\ov{D}$ is a simple pole of its meromorphic continuation to $U$.
Also assume that the critical line $\{\re s=\ov{D}\}$ contains another pole of $\widetilde{\zeta}_{A,\O}$, different from $\ov{D}$.
Furthermore, let
\begin{equation}
\lambda_{A,\O}:=\inf\left\{|\ov{D}-\omega|\,:\,\omega\in\dim_{PC}(A,\O)\setminus\left\{\ov{D}\right\}\right\}
\end{equation}

Then, if $\ov{D}<N$, we have the following upper bound for the upper $\ov{D}$-dimensional Minkowski content of $(A,\O)$, expressed in terms of the residue at $s:=\ov{D}$ of the relative tube zeta function of $(A,\O)$$:$
\begin{equation}\label{le_claim}
{\mathcal{M}}^{*\ov{D}}(A,\O)\leq\frac{C\lambda_{A,\O}(N-\ov{D})}{2\pi\left(1-\E^{-{2\pi(N-\ov{D})}/{\lambda_{A,\O}}}\right)}\res(\widetilde{\zeta}_{A,\O},\ov{D});
\end{equation}
moreover, in the case when $\ov{D}=N$, we have
\begin{equation}\label{le_claim_1}
{\mathcal{M}}^{*N}(A,\O)\leq C\res(\widetilde{\zeta}_{A,\O},N),
\end{equation}
where $($both in \eqref{le_claim}, \eqref{le_claim_1} just above and in \eqref{le_claim_2} just below$)$ $C$ is a positive constant such that $C<3$.

Finally, if $\ov{D}<N$, we have the following upper bound for the upper $\ov{D}$-dimensional Minkowski content of $(A,\O)$, expressed in terms of the residue at $s:=\ov{D}$ of the relative distance zeta function of $(A,\O)$$:$
\begin{equation}\label{le_claim_2}
{\mathcal{M}}^{*\ov{D}}(A,\O)\leq\frac{C\lambda_{A,\O}}{2\pi\left(1-\E^{-{2\pi(N-\ov{D})}/{\lambda_{A,\O}}}\right)}\res({\zeta}_{A,\O},\ov{D}).
\end{equation}
\end{theorem}

\begin{proof}
We use exactly the same reasoning as in the proof of Theorem~\ref{mink_char}, with the only difference being the fact that we can now only use the weaker statement~\eqref{statement_3} of Theorem~\ref{korevaar} since by hypothesis, there is another pole on the critical line $\{\re s=\ov{D}\}$, besides $\ov{D}$ itself.
More specifically, if $\ov{D}<N$ and $\lambda<\lambda_{A,\O}$, then by using~\eqref{ocjena_integrala} and~\eqref{statement_3}, we show that for every $h\geq 2\pi/\lambda$, we have
\begin{equation}\label{poomm}
C\res(\widetilde{\zeta}_{A,\O},\ov{D})\geq{\mathcal{M}}^{*\ov{D}}(A,\O)\frac{1-\E^{-h(N-\ov{D})}}{(N-\ov{D})h}.
\end{equation}
The right-hand side of \eqref{poomm} just above is a decreasing function of $h$; hence, we obtain the best estimate when $h=2\pi/\lambda$.
Furthermore, since this is true for every $\lambda<\lambda_{A,\O}$, we obtain~\eqref{le_claim} by letting $\lambda\to\lambda_{A,\O}^-$.
Moreover,~\eqref{poomm} is also valid if $\ov{D}=N$, but without the factor that depends on $h$, by a similar argument as in case~$(b)$ of the proof of Theorem~\ref{mink_char}.
Finally, the part of the theorem dealing with the relative distance zeta function $\zeta_{A,\O}$ follows by the same argument as in case $(a)$ of the proof of Theorem~\ref{mink_char}.
\end{proof}

\begin{remark}\label{weak}
Similarly as in the case of Theorem~\ref{mink_char} (see Remark~\ref{integral_remark}), the hypotheses of Theorem~\ref{mink_bound} can be somewhat weakened.
However, the above form of Theorem~\ref{mink_bound} is good enough for the most common situations which are encountered in the applications.
We again leave it to the interested reader to try to formulate and prove Theorem \ref{mink_bound} under these weaker assumptions.
\end{remark}

\subsection{Characterization of Minkowski measurability}\label{subsec_crit}

We now proceed to obtain a necessary condition for Minkowski measurability (Theorem \ref{necess} below) which will then be combined with Theorem \ref{mink_char} from \S4.1 in order to yield the desired criterion of Theorem \ref{criterion}. This necessary condition will follow from the distributional tube formula obtained in \cite{ftf_A}.
In order to accomplish this, we need the distributional tube formula stated in terms of the {\em Mellin zeta function} of a given relative fractal drum.
The Mellin zeta function is closely related to the tube zeta function and was introduced in \cite{ftf_A} in order to obtain a distributional tube formula valid for a larger class of test functions, which is needed in order to prove Theorem \ref{necess} below.\footnote{The general fractal tube formulas expressed in terms of the distance and the tube zeta functions are stated and established in \cite{ftf_A}; see also \cite[Chapter 5; $\S$5.1--$\S$5.3]{fzf}.}

For the sake of completeness, we recall here the definition of the Mellin zeta function along with some of its basic properties (see \cite[\S5.4]{ftf_A} for the details).

\begin{definition}\label{mellin_zeta_def}
Let $(A,\O)$ be an RFD in $\eR^N$ such that $\overline{\dim}_B(A,\O)<N$. The {\em Mellin zeta function} $\zeta_{A,\O}^{\mathfrak M}$ of $(A,\O)$ is defined by
\begin{equation}\label{mellin_zeta_1}
{\zeta}^{\mathfrak{M}}_{A,\O}(s):=\int_0^{+\ty}t^{s-N-1}|A_{t}\cap\O|\di t,
\end{equation}
for all $s\in\Ce$ with $\re s\in(\overline{\dim}_B(A,\O),N)$, where the integral is taken in the Lebesgue sense.
\end{definition}

We point out that the Mellin zeta function of $(A,\O)$ is actually equal to the Mellin transform of $f(t):=t^{-N}|A_t\cap\O|$, defined for all $t>0$; i.e., ${\zeta}^{\mathfrak{M}}_{A,\O}(s)=\{\mathfrak{M}f\}(s)$, where we recall that
\begin{equation}
\{\mathfrak{M}f\}(s):=\int_0^{+\ty}t^{s-1}f(t)\di t.
\end{equation}

The next theorem provides an explicit connection between the Mellin, tube and distance zeta functions (see \cite[Theorem 5.22]{ftf_A} for the proof).

\begin{theorem}\label{mellin_an}
Let $(A,\O)$ be an RFD in $\eR^N$ such that $\overline{\dim}_B(A,\O)<N$.
Then, the Mellin zeta function ${\zeta}^{\mathfrak{M}}_{A,\O}$ is holomorphic on the open vertical strip $\{\overline{\dim}_B(A,\O)\allowbreak<\re s<N\}$ and
\begin{equation}\label{mellin_der}
\frac{\di}{\di s}{\zeta}^{\mathfrak{M}}_{A,\O}(s)=\int_0^{+\ty}t^{s-N-1}|A_{t}\cap\O|\log t\di t,
\end{equation}
for all $s$ in $\{\overline{\dim}_B(A,\O)<\re s<N\}$.
Furthermore,  $\{\overline{\dim}_B(A,\O)<\re s<N\}$ is the largest $($horizontally bounded$)$ vertical strip on which the integral \eqref{mellin_zeta_1} is absolutely convergent $($i.e., is a convergent Lebesgue integral$)$.

Moreover, for all $s\in\Ce$ such that $\overline{\dim}_B(A,\O)<\re s<N$ and for any fixed $\d>0$ such that $\O\subseteq A_\d$, ${\zeta}^{\mathfrak{M}}_{A,\O}$ satisfies the following functional equations$:$
\begin{equation}\label{mellin_tube_t}
{\zeta}^{\mathfrak{M}}_{A,\O}(s)=\widetilde{\zeta}_{A,\O}(s;\d)+\frac{\d^{s-N}|\O|}{N-s}
\end{equation}
and
\begin{equation}\label{mellin_dist}
{\zeta}^{\mathfrak{M}}_{A,\O}(s)=\frac{{\zeta}_{A,\O}(s;\d)}{N-s}.
\end{equation}
\end{theorem}

The principle of analytic continuation implies that the functional equations \eqref{mellin_dist}, \eqref{mellin_tube_t} continue to hold on any connected open set to which any of the involved zeta functions has an analytic continuation.
Hence, for an RFD $(A,\O)$ in $\eR^N$ and provided $\ov{\dim}_B(A,\O)<N$, it is equivalent to define the complex dimensions of an RFD as the poles of its Mellin zeta function.

\begin{theorem}\label{pole_mellin}
Assume that $(A,\O)$ is a nondegenerate RFD in $\eR^N$, 
that is, $0<\M_*^D(A,\O)\le{\M}^{*D}(A,\O)<\ty$ $($in particular, $D:=\dim_B(A,\O)$ exists$)$, 
and $D<N$.
If ${\zeta}^{\mathfrak{M}}_{A,\O}$ can be extended meromorphically to a connected open neighborhood of $s=D$,
then $D$ is necessarily a simple pole of ${\zeta}^{\mathfrak{M}}_{A,\O}$ and 
\begin{equation}\label{res_mellin}
\M_*^D(A,\O)\le\res({\zeta}^{\mathfrak{M}}_{A,\O},D)\le{\M}^{*D}(A,\O).
\end{equation}
Furthermore, if $(A,\O)$ is Minkowski measurable, then 
\begin{equation}\label{pole1minkg1_mellin}
\res({\zeta}^{\mathfrak{M}}_{A,\O}, D)=\M^D(A,\O).
\end{equation}
\end{theorem} 

Moreover, in order to be able to state the distributional tube formula for relative fractal drums in terms of the Mellin zeta function, we need to recall the notion of distributional asymptotics (see~\cite{EsKa,jm,PiStVi} and also~\cite[Definition~5.29]{lapidusfrank12}).

Let $\mathcal{D}(0,+\ty):=C_c^{\ty}(0,+\ty)$ be the Schwartz space of infinitely differentiable functions with compact support contained in $(0,+\ty)$.
For a test function $\varphi\in\mathcal{D}(0,+\ty)$ and $a>0$, we let
\begin{equation}\label{varphi_a}
\varphi_a(t):=\frac{1}{a}\varphi\left(\frac{t}{a}\right).
\end{equation}

\begin{definition}\label{dist_order_dis}
Let $\mathcal{R}$ be a distribution in $\mathcal D'(0,\d)$ and let $\alpha\in\eR$.
We say that $\mathcal R$ is of {\em asymptotic order} at most $t^{\alpha}$ (resp., less than $t^{\alpha}$) as $t\to 0^+$ if applied to an arbitrary test function $\varphi_a$ in $\mathcal{D}(0,\d)$, we have that\footnote{In this formula, the implicit constant may depend on the test function $\varphi$.}
\begin{equation}
\langle\mathcal R,\varphi_a\rangle=O(a^{\alpha})\quad(\textrm{resp., }\ \langle\mathcal R,\varphi_a\rangle=o(a^{\alpha})),\quad\textrm{as}\quad a\to 0^+.
\end{equation}
We then write that $\mathcal R(t)=O(t^{\alpha})$ (resp., $\mathcal R(t)=o(t^{\alpha})$),\quad as\quad $a\to 0^+$.
\end{definition}
Note that if a continuous function $f$ satisfies the classical pointwise asymptotics, then $f$ also satisfies the same asymptotics, in the distributional sense of Definition~\ref{dist_order_dis}.

\begin{theorem}[Distributional fractal tube formula with error term, via $\zeta_{A,\O}^{\mathfrak M}$; level $k=0$]\label{dist_error_mellin}
Let $(A,\O)$ be a relative fractal drum in $\eR^N$ such that $\overline{\dim}_B(A,\O)<N$.
Furthermore, assume that $\zeta_{A,\O}^{\mathfrak M}$ satisfies the languidity conditions for some $\kappa\in\eR$ and $\d>0$.
Then, the regular distribution $\mathcal{V}_{A,\O}(t):=|A_t\cap\O|$ in $\mathcal{D}'(0,+\ty)$ is given by the following distributional identity in $\mathcal{D}'(0,+\ty)$$:$
\begin{equation}\label{dist_form_error_mellin}
|A_t\cap\O|=\sum_{\omega\in\po({\zeta}_{A,\O}^{\mathfrak{M}},\bm{W})}\res\left({t^{N-s}}{\zeta}_{A,\O}^{\mathfrak{M}}(s),\omega\right)+\mathcal{R}^{{\mathfrak M}}_{A,\O}(t).
\end{equation}
That is, the action of $\mathcal{V}_{A,\O}$ on an arbitrary test function $\varphi\in\mathcal{D}(0,+\ty)$ is given by
\begin{equation}\label{error_action_mellin}
\begin{aligned}
\big\langle\mathcal{V}_{A,\O},\varphi\big\rangle&=\sum_{\omega\in\po({\zeta}_{A,\O}^{\mathfrak{M}},\bm{W})}\res\Big({\{\mathfrak{M}\varphi\}(N\!-\! s\!+\! 1)\,{\zeta}_{A,\O}^{\mathfrak{M}}(s)},\omega\Big)+\big\langle\mathcal{R}^{{\mathfrak M}}_{A,\O},\varphi\big\rangle.
\end{aligned}
\end{equation}
Here, the distributional error term $\mathcal{R}^{{\mathfrak M}}_{A,\O}$ is the distribution in $\mathcal{D}'(0,+\ty)$ given for all $\varphi\in\mathcal{D}(0,+\ty)$ by
\begin{equation}\label{R_distr_mellin}
\big\langle\mathcal{R}^{{\mathfrak M}}_{A,\O},\varphi\big\rangle=\frac{1}{2\pi\I}\int_{\bm{S}}{\{\mathfrak{M}\varphi\}(N\!-\! s\!+\! 1)\,{\zeta}_{A,\O}^{\mathfrak{M}}(s)}\di s.
\end{equation}
Furthermore, the distribution $\mathcal{R}^{{\mathfrak M}}_{A,\O}(t)$ is of asymptotic order at most $t^{N-\sup S}$ as $t\to 0^+$; i.e.,
\begin{equation}
\mathcal{R}^{{\mathfrak M}}_{A,\O}(t)=O(t^{N-\sup S})\quad\mathrm{ as }\quad t\to0^+,
\end{equation}
in the sense of Definition~\ref{dist_order_dis}.

Moreover, if $S(\tau) < \sup S$ for all $\tau\in\eR$ $($that is, if the screen $\bm{S}$ lies strictly
to the left of the vertical line $\{\re s =\sup S\}$$)$, then $\mathcal{R}^{{\mathfrak M}}_{A,\O}(t)$ is of asymptotic order less than $t^{N-\sup S}$; i.e.,
\begin{equation}\label{dist_estimate_o_mellin}
\mathcal{R}^{{\mathfrak M}}_{A,\O}(t)=o(t^{N-\sup S})\quad\mathrm{ as }\quad t\to0^+,
\end{equation} 
again in the sense of Definition~\ref{dist_order_dis}.
\end{theorem}

We are now ready to state and prove the aforementioned necessary condition for Minkowski measurability of relative fractal drums.
We point out that in the following theorem, the phrase that the Mellin zeta function $\zeta_{A,\O}^{\mathfrak{M}}$ is languid means that it satisfies the languidity conditions of Definition \ref{languid} for some {\em languidity exponent} $\kappa\in\eR$, with the exception that in condition {\bf L1}, we now assume that $c\in(\ov{\dim}_B(A,\O),N)$.

\begin{theorem}[Necessary condition for Minkowski measurability]\label{necess}
Let $(A,\O)$ be a relative fractal drum in $\eR^N$ such that $D:=\dim_B(A,\O)$ exists, $D<N$ and $(A,\O)$ is Minkowski measurable.
Furthermore, assume that its Mellin zeta function $\zeta_{A,\O}^{\mathfrak{M}}$ is languid for some screen $\bm{S}$ passing strictly to the left of the critical line $\{\re s=D\}$ and strictly to the right of all the complex dimensions of $(A,\O)$ with real part strictly less than $D$.
 
Then, $D$ is the only pole of $\zeta_{A,\O}^{\mathfrak{M}}$ located on the critical line $\{\re s=D\}$ and it is simple.
\end{theorem}

\begin{proof}
Since $(A,\O)$ is languid, the hypotheses of Theorem~\ref{pole_mellin} are satisfied and, therefore, $s=D$ is a simple pole of $\zeta_{A,\O}^{\mathfrak{M}}$.
Furthermore, also by Theorem \ref{pole_mellin}, we have that $\mathcal{M}:=\mathcal{M}^D(A,\O)=\res(\zeta_{A,\O}^{\mathfrak{M}},D).$
All we need to show is that this is the only pole located on the critical line.
Firstly, we can see clearly from the definition of the Mellin zeta function given in Equation \eqref{mellin_zeta_1} that $|\zeta_{A,\O}^{\mathfrak{M}}(s)|\leq\zeta_{A,\O}^{\mathfrak{M}}(\re s)$, for all $s\in\{D<\re s<N\}$.
This inequality implies that if $\xi$ is another pole of $\zeta_{A,\O}^{\mathfrak{M}}$ with $\re\xi=D$, then it has to be simple.

Next, let us denote by $\xi_n=D+\I\gamma_n$, with $\gamma_n\in\eR$ and $n\in J$, the potentially infinite but at most counable sequence of poles of $\zeta_{A,\O}^{\mathfrak{M}}$ with real part $D$ (i.e., of principal poles of $\zeta_{A,\O}^{\mathfrak{M}}$).
Here, $J\subseteq\eN_0$ is a finite or infinite subset of $\eN_0$, $0\in J$, and we use the convention according to which $\gamma_0:=0$ and hence, $\xi_0:=D$.
Note also that, by hypothesis, we have that $\gamma_n\neq 0$ for all $n\in J\setminus\{0\}$.

Observe also that in light of the argument given in the first part of the proof, each principal pole $\xi_n$ is then also {\em simple}, for every $n\in J\setminus\{0\}$; so that we can let $a_n:=\res(\zeta_{A,\O}^{\mathfrak M},\xi_n)$, for every $n\in J$. 
Furthermore, recall that $a_0=\res(\zeta_{A,\O}^{\mathfrak M},D)=\mathcal M$, the Minkowski content of $(A,\O)$ since we assumed that $(A,\O)$ is Minkowski measurable.

Let us now assume that $J\setminus\{0\}$ is nonempty and see why this leads to a contradiction.
Namely, in light of Theorem \ref{dist_error_mellin} applied with the stronger error estimate given by \eqref{dist_estimate_o_mellin} and for the same choice of the screen $\bm{S}$ as assumed to exist in the statement of that theorem (and which also exists by hypothesis of the present theorem), we have that
\begin{equation}\label{A_t_mellin}
\begin{aligned}
|A_t\cap\O|&=\sum_{n\in J}a_nt^{N-\xi_n}+o(t^{N-D})\\
&=\mathcal{M}t^{N-D}+t^{N-D}\sum_{n\in J\setminus\{0\}}a_nt^{-\I\gamma_n} +o(t^{N-D})\quad \textrm{as}\quad t\to 0^+,
\end{aligned} 
\end{equation}
in the distributional sense since, by assumption, the screen $\bm{S}$ lies strictly to the left of the critical line $\{\re s=D\}$.

On the other hand, since $(A,\O)$ is assumed to be Minkowski measurable, we have that
\begin{equation}\label{A_t_func_m}
|A_t\cap\O|=\mathcal{M}t^{N-D}+o(t^{N-D})\quad\textrm{as}\quad t\to 0^+,
\end{equation}
in the usual pointwise sense and hence, also in the distributional sense.
Combining~\eqref{A_t_mellin} with~\eqref{A_t_func_m} together yields that
\begin{equation}\label{almost_per_a}
\sum_{n\in J\setminus\{0\}}a_nt^{-\I\gamma_n}=o(1)\quad \textrm{as}\quad t\to 0^+,
\end{equation}
in the distributional sense.
In light of the uniqueness theorem for almost periodic distributions (see \cite[\S{VI.9.6}, p.\ 208]{Schw}), Equation  \eqref{almost_per_a} can only be true if $a_n=0$ for all $n\in J\setminus\{0\}$, which contradicts our assumption.
We therefore conlude that there are no other poles on the critical line, except for $s=D$, as we needed to show. 
\end{proof}

\begin{remark}\label{tube_dist_nec}
We point out that the above theorem can also be stated in terms of the relative distance zeta function of $(A,\O)$.
This follows from the fact that Proposition \ref{propB} combined with the functional equation \eqref{mellin_dist} connecting the relative distance zeta function and the Mellin zeta function of $(A,\O)$ imply that if the $d$-languidity conditions {\bf L1} and {\bf L2} are satisfied by the distance zeta functions, then they are also satisfied by the Mellin zeta function with a possibly different languidity exponent. 
\end{remark}

\begin{remark}\label{5.4.16.1/4}
It is clear from the proof of Theorem \ref{necess} that it would suffice to assume in the statement of that theorem that the RFD $(A,\O)$ is Minkowski measurable {\em in the distributional sense} (which specifically means in the present context that Equation \eqref{A_t_func_m} above holds as a distributional identity in $\mathcal{D}'(0,+\ty)$, with $\mathcal{M}\in(0,+\ty)$).
\end{remark}

In light of the previous remark, we introduce the following definition.

\begin{definition}({\em Weak vs. strong Minkowski measurability}).\label{5.4.16.1/2}

\medskip 

$(i)$\ \ A relative fractal drum $(A,\O)$ such that $D:=\dim_{B}(A,\O)$ exists is said to be {\em Minkowski measurable, in the distributional sense} (or {\em weakly Minkowski measurable}, in short) if there exist a constant $\mathcal{M}\in(0,+\ty)$ such that, in the sense of distributional asymptotics (see Definition \ref{dist_order_dis}),
\begin{equation}\label{5.4.47.1/4}
\lim_{t\to 0^+}t^{-(N-D)}|A_t\cap\O|=\mathcal{M},\ \ \textrm{in}\ \ \mathcal{K}'(0,+\ty);
\end{equation}
i.e., for every $\varphi\in\mathcal{K}(0,+\ty)$,

\begin{equation}\label{5.4.47.1/2}
\begin{aligned}
\lim_{a\to 0^+}\int_0^{+\ty}t^{-(N-D)}|A_t\cap\O|\varphi_{a}(t)\di t&=\mathcal{M}\lim_{a\to 0^+}\int_0^{+\ty}\varphi_{a}(t)\di t\\
&=\mathcal{M}\int_0^{+\ty}\varphi(t)\di t.
\end{aligned}
\end{equation}
Here, as before, $\varphi_{a}$ is the scaled version of $\varphi$ defined by \eqref{varphi_a}.
Then, $\mathcal{M}$ is called the {\em weak Minkowski content} of the RFD $(A,\O)$.

\medskip

$(ii)$\ \ Much as in part $(i)$ of this definition, we can say that a relative fractal drum $(A,\O)$ is {\em strongly Minkowski measurable} (with {\em strong Minkowski content $\mathcal{M}$}) if it is Minkowski measurable in the usual (pointwise) sense; i.e., if there exists a constant $\mathcal{M}\in(0,+\ty)$ such that
\begin{equation}\label{5.4.47.3/4}
\lim_{t\to 0^+}t^{-(N-D)}|A_t\cap\O|=\mathcal{M},\ \ \textrm{in}\ \ \eR.
\end{equation}
\end{definition}

It is clear that if $(A,\O)$ is strongly Minkowski measurable, it is also weakly Minkowski measurable and then, the strong and weak Minkowski contents of $(A,\O)$ coincide.
Note also that we could similarly distinguish between weak and strong (or ordinary) Minkowski nondegeneracy, for example, although this definition will not be needed in the sequel.

\medskip

We point out here that the notion of Minkowski measurability being characterized in all of the criteria stated in this article (namely, Theorem 3.1, Theorem \ref{criterion}, Theorem \ref{tilde_criterion} and Corollary \ref{5.4.20.1/4}) {\em is always the notion of strong $($or ordinary$)$ Minkowski measurability}, 
in the sense of part $(ii)$ of Definition \ref{5.4.16.1/2} above.

\medskip

Finally, we are now ready to state the announced Minkowski measurability criterion, which follows directly from Theorems \ref{mink_char} and \ref{necess}.

\begin{theorem}[Minkowski measurability criterion in terms of $\zeta_{A,\O}$]\label{criterion}
Let $(A,\O)$ be a relative fractal drum in $\eR^N$ such that $D:=\dim_B(A,\O)$ exists and $D<N$.
Furthermore, assume that $(A,\O)$ is $d$-languid for a screen passing {\rm strictly} between the critical line $\{\re s=D\}$ and all the complex dimensions of $(A,\O)$ with real part strictly less than $D$.
Then the following statements are equivalent$:$

\medskip

$(a)$ The RFD $(A,\O)$ is $($strongly$)$ Minkowski measurable.

\medskip

$(b)$ $D$ is the only pole of the relative distance zeta function ${\zeta}_{A,\O}$ located on the critical line $\{\re s=D\}$ and it is simple.
\end{theorem}

\begin{remark}\label{ekv_mink_krt}
In part $(b)$ of the above criterion, we can replace ${\zeta}_{A,\O}$ with the relative tube zeta function $\widetilde{\zeta}_{A,\O}$ or the Mellin zeta function ${\zeta}_{A,\O}^{\mathfrak{M}}$.
Of course, in that case, one then assumes that the chosen fractal zeta function satisfies the usual languidity conditions (along with the condition from Theorem \ref{criterion} about the existence of a suitable screen).
\end{remark}

Next, we explicitly state the counterpart of Theorem \ref{criterion}, but now in terms of the tube zeta function $\widetilde{\zeta}_{A,\O}$ (instead of the distance zeta function ${\zeta_{A,\O}}$) since in this case the restriction $\dim_B(A,\O)<N$ can be removed.

\begin{theorem}[Minkowski measurability criterion in terms of $\widetilde{\zeta}_{A,\O}$]\label{tilde_criterion}
Let $(A,\O)$ be a relative fractal drum in $\eR^N$ such that $D:=\dim_B(A,\O)$ exists.
Furthermore, assume that $(A,\O)$ is languid for a screen $\bm{S}$ passing strictly to the left of the critical line $\{\re s=D\}$ and strictly to the right of all the complex dimensions of $(A,\O)$ with real part strictly less than $D$.
Then the following statements are equivalent$:$

\medskip

$(a)$ The RFD $(A,\O)$ is $($strongly$)$ Minkowski measurable.

\medskip

$(b)$ $D$ is the only pole of the relative distance zeta function $\widetilde{\zeta}_{A,\O}$ located on the critical line $\{\re s=D\}$, and it is simple.
\end{theorem}

\begin{proof}
Firstly, if $D=\dim_B(A,\O)<N$, then the conclusion follows from Theorems \ref{mink_char} and \ref{necess} together with Remark \ref{tube_dist_nec}.

In the case when $D=N$, we embed $(A,\O)$ into $\eR^{N+1}$, as was done in \cite[\S3.2]{mefzf} (see also \cite{fzf}), and then use the results of the present section as we now explain.
%
%
%
%
The fact that $(b)$ implies $(a)$ is a direct consequence of Theorem \ref{mink_char} since there are no restrictions of the type ${\dim}_B(A,\O)<N$ in the hypotheses of that theorem.
In fact, one sees directly from the definition of the relative Minkowski content that $\dim_B(A,\O)=N$ implies that $\mathcal{M}^{N}(A,\O)$ exists and $\mathcal{M}^{N}(A,\O)=|\ov{A}\cap\O|$, where $\ov A$ denotes the closure of $A$ in~$\eR^N$.

In order to prove that $(a)$ implies $(b)$, we embed $(A,\O)$ into $\eR^{N+1}$ as 
$$
(A,\O)_1:=(A\times\{0\},\O\times(-\d,\d)),\ \textrm{ for some suitable }\d>0,
$$
and then, by \cite[Theorem 3.9]{mefzf} or \cite[Theorem 4.7.9]{fzf}, the relative tube zeta functions of the RFDs $(A,\O)$ and $(A,\O)_1$ are connected by the following approximate functional equation:
\begin{equation}\label{d_d}
\widetilde{\zeta}_{A\times\{0\},\O\times(-\d,\d)}(s;\delta)= \frac{\sqrt{\pi}\,\Gamma\left(\frac{N-s}{2}+1\right)}{\Gamma\left(\frac{N+1-s}{2}+1\right)}\widetilde{\zeta}_{A,\O}(s;\delta)+E(s;\delta).
\end{equation}
Here, $\d>0$ is chosen such that $\widetilde{\zeta}_{A,\O}(s;\delta)$ satisfies the languidity hypothesis of the theorem.
Next, we show that $\widetilde{\zeta}_{A\times\{0\},\O\times(-\d,\d)}(\,\cdot\,;\delta)$ satisfies the needed languidity conditions and use Theorem \ref{necess} to conclude the proof.
The error function $E(\,\cdot\,;\d)$ is known to be holomorphic on the open left half-plane $\{\re s<N+1\}$ and bounded by ${2\delta^{\re s-N}|A_{\delta}\cap\O|_N}\left(\frac{\pi}{2}-1\right)$ (see the proof of \cite[Theorem 3.9]{mefzf} or \cite[Theorem 4.7.9]{fzf}).
This fact implies that $E(\,\cdot\,;\d)$ is languid (with a languidity exponent equal to $0$).
Furthermore, for any $a,b\in\Ce$ such that $\re (b-a)>0$, the following asymptotic expansion is well known:
\begin{equation}
\frac{\Gamma(z+a)}{\Gamma(z+b)}\sim z^{a-b}\sum_{n=0}^{\ty}\frac{(-1)^nB_n^{(a-b+1)}(a)}{n!}\frac{\Gamma(b-a+n)}{\Gamma(b-a)}\frac{1}{z^n}\q\textrm{as}\q |z|\to+\ty,
\end{equation}
valid for $z$ such that $|\arg z|<\pi$.\footnote{Here, $B_n^{(\sigma)}(x)$ is the generalized Bernoulli polynomial (see, e.g., \cite{SriTod} for the exact definition and an explicit formula). See also \cite[\S3.6.2]{tem} for this result on asymptotics of the ratio of gamma functions.}
We next substitute $z=\frac{N-s}{2}+1$, $a=0$ and $b=1/2$ to obtain that
\begin{equation}
\frac{\Gamma\left(\frac{N-s}{2}\! +\! 1\right)}{\Gamma\left(\frac{N+1-s}{2}\! +\! 1\right)}\sim(N-s+2)^{-\frac{1}{2}}\sum_{n=0}^{\ty}\frac{(2n)!\sqrt{2}(-1)^n B_n^{(1/2)}(0)}{2^{n}(n!)^2(N-s+2)^{n}}\q \textrm{as}\q |s|\to+\ty,
\end{equation}
for all $s\in\Ce\setminus[N+2,+\ty)$.\footnote{We have used here the identities $\Gamma(1/2)=\sqrt{\pi}$ and $\Gamma(1/2+n)=\frac{(2n)!}{4^nn!}\sqrt{\pi}$.}
In particular, we have that
\begin{equation}
\frac{\Gamma\left(\frac{N-s}{2}+1\right)}{\Gamma\left(\frac{N+1-s}{2}+1\right)}=O(|s|^{-1/2})\q \textrm{as}\q |s|\to+\ty,
\end{equation}
for all $s\in\Ce\setminus[N+2,+\ty)$, from which we conclude that the product of this ratio of gamma functions with the relative tube zeta function $\widetilde{\zeta}_{A,\O}(\,\cdot\,;\delta)$ is languid with a languidity exponent not greater than $\kappa-1/2$, where $\kappa$ is the languidity exponent of $\widetilde{\zeta}_{A,\O}(\,\cdot\,;\delta)$.
This fact, together with the languidity of $E(\,\cdot\,;\d)$ and Equation \eqref{d_d}, implies that $\widetilde{\zeta}_{A\times\{0\},\O\times(-\d,\d)}(\,\cdot\,;\delta)$ is languid with the same choice of a double sequence $(T_n)_{n\in\Ze\setminus\{0\}}$ and the screen $\bm{S}$ as for $\widetilde{\zeta}_{A,\O}(\,\cdot\,;\delta)$ and with a languidity exponent not greater than $\max\{\kappa-1/2,0\}$.

On the other hand, if $(A,\O)$ is Minkowski measurable, then this is also true for the embedded RFD $(A\times\{0\},\O\times(-\d,\d))$.
This fact follows in a completely analogous way as in the case of bounded subsets of $\eR^N$ which was proven in \cite{maja} and extended to RFDs in \cite[$\S$4.7]{fzf}.
We now conclude the proof by invoking Theorem~\ref{necess}, or rather, its analog in terms of the relative tube zeta function (see Remark \ref{tube_dist_nec}).
\end{proof}

In the next corollary of Theorem \ref{criterion} and \ref{tilde_criterion}, and in light of \cite[Lemma 5.25]{ftf_A}\footnote{This lemma states that, as long as $\ov{\dim}_B(A,\O)<N$, the complex dimensions of a given RFD $(A,\O)$ in $\eR^N$ do not depend on the choice of the associated fractal zeta function in terms of which they are defined.}  (see also \cite[Lemma 5.4.11]{fzf}) and Remark \ref{tube_dist_nec}, we can indifferently interpret the (principal) complex dimensions of the RFD $(A,\O)$ as being the (principal) poles of either the distance, tube, or Mellin zeta function of $(A,\O)$.

\begin{corollary}[Characterization of Minkowski measurability in terms of the complex dimensions]\label{5.4.20.1/4}
Let $(A,\O)$ be a relative fractal drum in $\eR^N$, with $N\geq 1$ arbitrary, such that $D:=\dim_B(A,\O)$ exists and $D<N$.
Assume also that any of its fractal zeta functions $($specifically, $\zeta_{A,\O}$ or $\widetilde{\zeta}_{A,\O}$, respectively$)$ satisfies the hypotheses of Theorem \ref{criterion} $($or of Theorem \ref{tilde_criterion}, respectively$)$ concerning the languidity and the screen.
Then, the following statements are equivalent$:$

\medskip

$(a)$ The RFD $(A,\O)$ is $($strongly$)$ Minkowski measurable.

\medskip

$(b)$ $D$ is the only complex dimension of the RFD $(A,\O)$ with real part equal to $D$ $($i.e., located on the critical line $\{\re s=D\}$$)$, and it is simple.
\end{corollary}

\section{Complex dimensions and gauge-Minkowski measurability}\label{gauge_mink}

In this section, we will obtain a very general result (Theorem \ref{degg}) which enables us to generate $h$-Minkowski measurable RFDs, where $h(t):=(\log t^{-1})^{m-1}$ for all $t\in(0,1)$ and $m$ is a positive integer.
In order to do this, we need only some information about the principal poles and their multiplicities.
Its proof follows from the pointwise fractal tube formula with or without error term; \cite[Theorem 3.2]{ftf_A} and \cite[Theorem 3.4]{ftf_A}, respectively (see also \cite[Theorem 5.1.14]{fzf}).
Especially important is the asymptotic expansion of the tube function in Equation \eqref{tubem}, from which it is possible to deduce the optimal tube function asymptotic expansion for a class of $h$-Minkowski measurable RFDs, stated in Theorem \ref{tubealpha}.

Let us recall the definitions of gauge functions and gauge Minkowski content.
Note that the latter notion is motivated, geometrically and physically, by the study of non power law scaling behavior which arises in many natural examples. (See [HeLap] and the relevant references therein.) All of these definitions can easily be extended to RFDs $(A,\O)$ in $\eR^N$ but for notational simplicity, we will only state them in the case of bounded sets.

If $A$ is (Minkowski) degenerate and such that $D:=\dim_BA$ exists, we assume that
\begin{equation}\label{F1}
|A_t|=t^{N-D}(F(t)+o(1))\quad\mbox{as\qs$t\to0^+$},
\end{equation}
where $F:(0,\e_0)\to(0,+\ty)$, for some sufficiently small $\e_0>0$.

\medskip

 Let $A$ be a degenerate set in $\eR^N$. Then:


\medskip

     $\bullet$ $A$ is {\em weakly degenerate}\label{wdeg} 
		if $D=\dim_BA$ exists and either $\M_*^D(A)=0$ (i.e., we have that $\liminf_{t\to0^+} F(t)=0$) or $\M^{*D}(A)=+\ty$ (i.e., $\limsup_{t\to0^+} F(t) =+\ty$); see Equation~(\ref{F1}).
		
		\medskip
		
     $\bullet$ $A$ is {\em strongly degenerate}\label{sdeg} 
		if $\underline\dim_BA<\overline\dim_BA$. Note that here, (\ref{F1}) cannot hold for any $D\ge0$
     (otherwise, $D$ would here be the box dimension of $A$).

\bigskip

Furthermore, weakly degenerate sets can be classified with respect to their {\em gauge functions} $h$, if they exist\label{wds} 
(see Definition~\ref{gaugef}).
We assume that the function
$F(t)$ appearing in Equation~(\ref{F1}) is of the form
\begin{equation}\label{Gh}
F(t)=h(t)\quad\mbox{or\q$\ F(t)=\frac{1}{h(t)}$,}
\end{equation}
where $h:(0,\e_0)\to(0,+\ty)$, for some small $\e_0>0$, $h(t)\to+\ty$ as $t\to0^+$ and
\begin{equation}\label{O0}
h(t)=O(t^{0)})\q\mbox{as\qs$t\to0^+$}, 
\quad\mbox{with}\q O(t^{0)}):=\bigcap_{\beta<0}O(t^\beta).
\end{equation}
Note that we need to assume that $h(t)=O(t^{0)})$ as $t\to 0^+$ in order to fix the value $D=\dim_BA$; see Equation~(\ref{F1}).

\begin{definition}\label{gaugef}
If a function $h:(0,\e_0)\to(0,+\ty)$ is of class $O(t^{0)})$ (as defined in Equation \eqref{O0}) and converges to infinity as $t\to0^+$, we then say that $h$ is {\em of slow growth to infinity} as $t\to0^+$. Analogously, a function $g:(0,\e_0)\to(0,+\ty)$ is said to be {\em of slow decay to zero} as $t\to0^+$ if it is of the form $g(t)=1/h(t)$, for some function $h$ which is of slow growth to infinity
as $t\to0^+$. Such functions $h$ and $g$ are called {\em gauge functions}.
\end{definition}

It is not difficult to see that a function $g:(0,\e_0)\to(0,+\ty)$ is of slow decay to zero as $t\to0^+$ if and only if for every $\b>0$,  $t^\b=O(g(t))$ as $t\to0^+$.

\begin{example}\label{gaugeexa}
If we let $h_1(t):=\log t^{-1}$, $h_2(t):=\log\log t^{-1}$, and more generally, $h_3(t):=(\log t^{-1})^{a}$, $h_4(t):=(\log^kt^{-1})^{a}$ (here, $a>0$,
 $k\in\eN$, and $\log^k$ denotes the $k$-fold composition of logarithms), then all of these functions are of slow growth to infinity as $t\to0^+$. Furthermore, their reciprocals
 are functions of slow decay to $0$ as $t\to0^+$.
\end{example}

Since for a weakly degenerate set $A$ we have $\M^{*D}(A)=+\ty$ or $\M_*^D(A)=0$, it will be useful to define (as in \cite{lapidushe}) the {\em upper} and {\em lower} $D$-dimensional {\em Minkowski contents of $A$
with respect to a given gauge function $h$}\label{mink_cont_gauge}, as follows (with $g:=1/h$):
\begin{equation}
\begin{aligned}
\M^{*D}(A,h)&=\limsup_{t\to0^+}\frac{|A_t|}{t^{N-D}h(t)},\\
\M_*^D(A,h)&=\liminf_{t\to0^+}\frac{|A_t|}{t^{N-D}g(t)}.
\end{aligned}
\end{equation}
The aim here is to find gauge functions $h$ and $g$ so that the upper and lower Minkowski contents of $A$ with respect to $h$ are 
nondegenerate, that is, belong to $(0,+\ty)$.
If $\M_*^D(A,h)=\M^{*D}(A,h)$, this common value is denoted by $\M^D(A,h)$ and called the {\em $h$-Minkowski content} of $A$.
We then say that $A$ is {\em $h$-Minkowski measurable}.


\begin{definition}\label{gauge_def} Assume that $A$ is a bounded subset of $\eR^N$ such that (\ref{F1}) holds under one of the conditions stated in (\ref{Gh}) and  that, in addition, (\ref{O0}) is satisfied.
We then say that $h=h(t)$ or $g=1/h(t)$ is a {\em gauge function} of $A$.\footnote{In the case when $F(t)=g(t)$, we also assume that the implied function $o(1)$ appearing in \eqref{F1} satisfies $o(1)/g(t)\to0$ as $t\to0^+$.}
We also say that the set $A$ is {\em weakly degenerate, of type $h$ or $1/h$}, respectively.
\end{definition} 

Note that in the first case of (\ref{Gh}), we have 
\begin{equation}
\M^{*D}(A)=+\ty,\quad
0<\M^{*D}(A,h)<\infty,\nonumber
\end{equation}
while in the second case of (\ref{Gh}), we have
\begin{equation}
\M_*^D(A)=0,\quad
0<\M_*^D(A,1/h)<\infty.\nonumber
\end{equation}

 Let $A$ be a weakly degenerate set in $\eR^N$ of type $h$, in the sense of Definition \ref{gauge_def}. We will say that
\medskip


     $\bullet$ $A$ is a {\em constant weakly degenerate set} of type $h$ (or a {\em constant $h$-degenerate set}, or in the sequel, an {\em $h$-Minkowski measurable
		set}), if $\M^{D}(A,h)$ exists and
     belongs to $(0,+\ty)$. We adopt a similar terminology in the case of the gauge function $1/h$ instead of $h$; see Definition \ref{gauge_def}.
		
		\medskip
		
     $\bullet$ $A$ is a {\em nonconstant weakly degenerate set of type $h$} (or a {\em nonconstant $h$-degenerate set}, or an {\em $h$-Minkowski 
		nonmeasurable set}) if
     $$0<\M_*^{D}(A,h)<\M^{*D}(A,h)<\ty.$$ We adopt a similar terminology for the gauge function $1/h$.

\begin{theorem}\label{degg} $(${\rm Generating $h$-Minkowski measurable RFDs}$).$
Let $(A,\O)$ be a relative fractal drum in $\eR^N$ which is languid with $\kappa<-1$ or such that $(\lambda A,\lambda\O)$ is strongly languid for some $\lambda>0$ with $\kappa<0$, for a screen $\bm{S}$ passing strictly between the critical line $\{\re s=\ov{\dim}_B(A,\O)\}$ and all the complex dimensions of $(A,\O)$ with real part strictly less than $\ov{D}:=\ov{\dim}_B(A,\O)$.
Furthermore, suppose that $\ov{D}$ is the only pole of its relative tube zeta function $\widetilde{\zeta}_{A,\O}$ with real part equal to $\ov{D}$ of multiplicity $m\geq1$ ($m\in\eN$) and, additionally, that there exists $($at most$)$ finitely many nonreal poles of $\widetilde{\zeta}_{A,\O}$ with real part $\ov{D}$.
Also, assume that the multiplicity of each of those nonreal poles is of order strictly less than $m$.
Then, $\dim_B(A,\O)$ exists and is equal to $D:=\ov{D}$.
Moreover, $\mathcal{M}^{D}(A,\O)$ exists and is equal to $+\ty;$ hence, $(A,\O)$ is Minkowski degenerate.

Finally, an appropriate gauge function for $(A,\O)$ is $h(t):=(\log t^{-1})^{m-1}$, for all $t\in(0,1)$, and we have that, {\rm relative to} $h$, the RFD $(A,\O)$ is not only Minkowski nondegenerate but is also Minkowski measurable with Minkowski content given by
\begin{equation}\label{hML}
{\mathcal M}:=
\mathcal{M}^{{D}}(A,\O,h)=\frac{\widetilde{\zeta}_{A,\O}[D]_{-m}}{(m-1)!},
\end{equation}
where $\widetilde{\zeta}_{A,\O}[D]_{-m}$ denotes the coefficient corresponding to $(s-D)^{-m}$ in the Laurent expansion of $\widetilde{\zeta}_{A,\O}$ around $s=D$. Moreover, if there is at least one nonreal complex dimension on the critical line $\{\re=D\}$, then the tube function $t\mapsto |A_t\cap\O|$ has the following asymptotic estimate$:$
\begin{equation}\label{tubemlog}
|A_t\cap\O|=t^{N-D}h(t)\big({\mathcal M}+O((\log t^{-1})^{-1})\big)\q\mbox{\rm as\q $t\to0^+$,}
\end{equation}
while if $D$ is the unique pole of $\widetilde\zeta_{A,\O}$ on the critical line $($i.e., the unique principal complex dimension of $(A,\O)$$)$, we have the following sharper estimate$:$
\begin{equation}\label{tubem}
|A_t\cap\O|=t^{N-D}h(t)\big({\mathcal M}+O(t^{D-\sup S})\big)\q\mbox{{\rm as}\q $t\to0^+$.}
\end{equation}
\end{theorem}

\begin{proof}
Let $\omega_0:=\ov{D}:=\ov{\dim}_B(A,\O)$,\footnote{It will follow from the proof that $\dim_B(A,\O)$ exists, and that $\ov{D}=\dim_B(A,\O)$ and hence, is equal to $D$.} and let $\omega_j:=\ov{D}+\I\gamma_j$, where $\gamma_j\in\eR\setminus\{0\}$ for $j\in J$ and $J$ is a finite subset of $\Ze\setminus\{0\}$.
That is, $\{\omega_j\}_{j\in J}$ is the (finite) set of all the other poles of $\widetilde{\zeta}_{A,\O}$ located on the critical line $\{\re s=\ov{D}\}$, i.e., with real part $\ov{D}$.
We also let $\gamma_0:=0$ and $m_0:=m$, in order to be consistent with the notation introduced just below.
Furthermore, for each $j\in J$, let $m_j$ be the multiplicity of $\omega_j$ and then, by hypothesis of the theorem, we have that $m_j<m$ for every $j\in J$.
By \cite[Theorem 3.2]{ftf_A} (or \cite[Theorem 5.1.14]{fzf}) and since the screen $\bm{S}$ is strictly to the right of all the other complex dimensions of $(A,\O)$ with real part strictly less than $\ov{D}$ and strictly to the left of the critical line $\{\re s=\ov{D}\}$, we obtain a pointwise tube formula for $|A_t\cap\O|$ with an error term that is of strictly higher asymptotic order as $t\to0^+$ than the term corresponding to the residue at $s=\ov{D}$; more specifically, we have the following pointwise tube formula with error term:
\begin{equation}\label{degene}
|A_t\cap\O|=\sum_{j\in J\cup\{0\}}\res(t^{N-s}\widetilde{\zeta}_{A,\O}(s),\omega_j)+O(t^{N-\sup S})\quad \textrm{as}\ t\to0^+.
\end{equation}
We next consider the Taylor expansion of $t^{N-s}$ around $s=\omega_j$ (for each $j\in J\cup\{0\}$),
\begin{equation}\label{tns}
t^{N-s}=t^{N-\omega_j}\,\E^{(s-\omega_j)\log t^{-1}}=t^{N-\omega_j}\sum_{n=0}^{\ty}\frac{(\log t^{-1})^n}{n!}(s-\omega_j)^{n},
\end{equation}
multiply it by the Laurent expansion of $\widetilde{\zeta}_{A,\O}(s)$ around $s=\omega_j$ and extract the residue of this product in order to deduce that
\begin{equation}\label{res_logt}
\res(t^{N-s}\widetilde{\zeta}_{A,\O}(s),\omega_j)=t^{N-\omega_j}\sum_{n=0}^{m_j-1}\frac{(\log t^{-1})^n}{n!}\widetilde{\zeta}_{A,\O}[\omega_j]_{-n-1}.
\end{equation}
From the above identity and from \eqref{degene}, we conclude that $\dim_B(A,\O)$ exists and is equal to $\ov{D}$.
Furthermore, since $m_j<m_0=m$, we have that the highest power of $\log t^{-1}$ appearing in the fractal tube formula \eqref{degene} is $m-1$, and that it appears only in the sum \eqref{res_logt} when $j=0$.
Therefore, if we choose $h(t):=(\log t^{-1})^{m-1}$ for $t\in(0,1)$ as our gauge function, the statements about the Minkowski content and the gauge Minkowski content now also follow from the fractal tube formula \eqref{degene}.

The interested reader can now check that Equations \eqref{tubemlog} and \eqref{tubem} follow easily from \eqref{degene} by rewriting \eqref{res_logt} as follows:
\begin{equation}\label{res_logt2}
\res(t^{N-s}\widetilde{\zeta}_{A,\O}(s),\omega_j)=t^{N-D}h(t)\sum_{n=0}^{m_j-1}t^{D-\o_j}\frac{(\log t^{-1})^{n-m+1}}{n!}\widetilde{\zeta}_{A,\O}[\omega_j]_{-n-1}.
\end{equation}
\vskip2mm
\noindent Indeed, for $j=0$, we have that $m_0=m$ and $\o_0=D$; so that the term on the right-hand side of \eqref{res_logt2} corresponding to $n=m-1$ is equal to $\frac{\widetilde{\zeta}_{A,\O}[\omega_j]_{-m}}{(m-1)!}$
(i.e., to $\mathcal M$, the $h$-Minkowski content of $(A,\O)$; see Equation \eqref{hML} above), while for any $n\in\{0,\dots,m-2\}$ (if this set is nonempty, i.e., if $m\ge2$),
we are left with a function which is $O((\log t^{-1})^{-1})$ as $t\to0^+$ (if $m=1$, the corresponding function is absent; i.e., it is equal to zero). Equations \eqref{tubemlog} and \eqref{tubem} then follow because for $j\ne0$ we have that $|t^{D-\o_j}|=1$ (since $D-\o_j$ is a purely imaginary complex number) and $(\log t^{-1})^{n-m+1}=O((\log t^{-1})^{-1})$ as $t\to0^+$. 

This concludes the proof of the theorem.
\end{proof}

\begin{remark}
In the statement and in light of the proof of Theorem \ref{degg},
the remainder term $O((\log t^{-1})^{-1})$ in Equation \eqref{tubemlog} can be slightly improved to $O((\log t^{-1})^{n-m+1})$,
where $n$ is the largest positive integer such that $n<m-1$ and for which there exists $j\in J$ such that $\widetilde{\zeta}_{A,\O}[\omega_j]_{-n-1}\ne0$.
\end{remark}

We point out that the further to the left we can meromophically extend the tube zeta function $\widetilde\zeta_{A,\O}$ of a given RFD $(A,\O)$
(meaning, the smaller the value of $\sup \bm{S}$), the sharper the estimate \eqref{tubem} in Theorem \ref{degg}. In other words, 
if we denote by ${\mathcal S}(A,\O)$ the family of all possible screens $\bm{S}$ such that the tube zeta function $\widetilde\zeta_{A,\O}$ admits a meromorphic extension to an open connected neighborhood of the corresponding window $W=W(S)$, it is natural to define the following (extended) real number:
\begin{equation}\label{alphaAO}
{\rm or}(A,\O):=\inf_{S\in{\mathcal S}(A,\O)}\sup\, S\in[-\ty,\ov\dim_B(A,\O)],
\end{equation}
which we call the {\em order}\label{orderAO} of the RFD $(A,\O)$. 
It is clear that 
\begin{equation}\label{sigmamer}
{\rm or}(A,\O)\le D_{\rm mer}(\widetilde\zeta_{A,\O}),
\end{equation}
where $D_{\rm mer}(\widetilde\zeta_{A,\O})$ is the abscissa of meromorphic continuation of the tube zeta function $\widetilde\zeta_{A,\O}$.
In light of Equation \eqref{tubem} in Theorem \ref{degg}, we can now deduce the following significant conclusion.

\begin{theorem}\label{tubealpha} $(${\rm Optimal tube function asymptotic expansion}$).$
Let $(A,\O)$ be a relative fractal drum such that the conditions of Theorem \ref{degg} are satisfied, with $D:=D(\widetilde\zeta_{A,\O})$ being the unique pole of $\widetilde\zeta_{A,\O}$ on the open right half-plane $\{\re s>D_{\rm mer}(\widetilde\zeta_{A,\O})\}$,
of multiplicity $m\geq1$.
Then, for any positive real number $\e$, the tube function $t\mapsto|A_t\cap\O|$ has the following asymptotic expansion, with an error term$:$
\begin{equation}\label{tubealpha2}
|A_t\cap\O|=t^{N-D}h(t)\big({\mathcal M}+O(t^{D-D_{\rm mer}(\widetilde\zeta_{A,\O})-\e})\big)\q\mbox{\rm as\q $t\to0^+$,}
\end{equation}
where $h(t):=(\log t^{-1})^{m-1}$ for all $t\in(0,1)$.
Furthermore, the asymptotic formula in Equation \eqref{tubealpha2} is optimal; that is, the constant $D-D_{\rm mer}(\widetilde\zeta_{A,\O})$ cannot be replaced by a larger number. Moreover, the order of the RFD $(A,\O)$ 
is equal to the abscisssa of meromorphic continuation of the corresponding tube zeta function $\widetilde\zeta_{A,\O}$; i.e.,
\begin{equation}\label{ormer}
{\rm or}(A,\O)=D_{\rm mer}(\widetilde\zeta_{A,\O}).
\end{equation}
\end{theorem}

\begin{proof}
Firstly, using Theorem \ref{degg} and since for any $\e>0$ there exists $S\in\mathcal S$ 
such that $\sup S>{\rm or}(A,\O)+\e$, we have that
\begin{equation}\label{tubealpha1}
|A_t\cap\O|=t^{N-D}h(t)\big({\mathcal M}+O(t^{D-{\rm or}(A,\O)-\e})\big)\q\mbox{as\, $t\to0^+$.}
\end{equation}
Equation \eqref{tubealpha2} then follows from \eqref{sigmamer}.

In order to obtain the {\em optimality} of the asymptotic formula in Equation \eqref{tubealpha2}, we reason by contradiction.
Assume that we have
\begin{equation}
|A_t\cap\O|=t^{N-D}h(t)\big({\mathcal M}+O(t^{\a})\big)\q\mbox{as\q $t\to0^+$,}
\end{equation}
for some real number
\begin{equation}\label{consta}
\a>D-D_{\rm mer}(\widetilde\zeta_{A,\O}). 
\end{equation}
(Note that in the statement of \eqref{consta}, we can omit $\e$ since, without loss of generality, we can always choose a smaller real number $\a$ satisfying the same strict inequality.) 
We now use \cite[Theorem~4.5.1]{fzf} (applied with $m-1$ instead of $m$) or [LapRa\v Zu3] to conclude that $\widetilde\zeta_{A,\O}$ can be meromophically extended (at least) to the open right half-plane $\{\re s>D-\a\}$.
It then follows directly from the definition of the abscissa of meromorphic continuation that $D_{\rm mer}(\widetilde\zeta_{A,\O})\le D-\a$.
On the other hand, we can deduce from \eqref{consta} that $D-\a<D_{\rm mer}(\widetilde\zeta_{A,\O})$,
which contradicts the above inequality in \eqref{consta}. 

Finally, in order to obtain Equation \eqref{ormer}, it suffices to note that, due to the optimality proved just above, we must have that $D-{\rm or}(A,\O)\le D-D_{\rm mer}(\widetilde\zeta_{A,\O})$ (see Equations \eqref{tubealpha1} and \eqref{tubealpha2}); i.e., ${\rm or}(A,\O)\ge D_{\rm mer}(\widetilde\zeta_{A,\O})$, which together with \eqref{sigmamer}, implies the desired equality \eqref{ormer}.
This concludes the proof of the theorem.
\end{proof}

We can rephrase this result, assuming that the conditions of Theorem \ref{tubealpha} are satisfied.
Namely, we then conclude that, the larger the difference 
\begin{equation}\label{aAO}
\a(A,\O):=D(\widetilde\zeta_{A,\O})-D_{\rm mer}(\widetilde\zeta_{A,\O})
\end{equation}
between the abscissa of (absolute) convergence and the abscissa of meromorphic continuation of the tube zeta function $\widetilde\zeta_{A,\O}$ of a given RFD $(A,\O)$, the better the asymptotic estimate \eqref{tubealpha2} of the tube function $t\mapsto|A_t\cap\O|$ when $t\to0^+$.


\medskip
\medskip

The next result, Theorem \ref{converse-log}, can be actually viewed as the converse of \cite[Theorem 4.5.1]{fzf}.

\begin{theorem}\label{converse-log}
Let $(A,\O)$ be an RFD such that the conditions of Theorem \ref{degg} are satisfied. 
Furthermore, assume that there exists a positive real number $\a$ such that the relative tube zeta function $\widetilde\zeta_{A,\O}$ can be meromorphically extended to the open right half-plane $\{\re s> D-\a\}$,
with $D:=D(\widetilde\zeta_{A,\O})$ being the unique pole of $\widetilde\zeta_{A,\O}$ in this right half-plane, of multiplicity $m\geq1$.
Then, the tube function $t\mapsto|A_t\cap\O|$ has the following asymptotic expansion, with error term of order $\a$$:$
\begin{equation}\label{tubealpha3}
|A_t\cap\O|=t^{N-D}h(t)\big({\mathcal M}+O(t^{\a})\big)\q\mbox{\rm as\q $t\to0^+$,}
\end{equation}
where $h(t):=(\log t^{-1})^{m-1}$ for all $t\in(0,1)$.
Moreover, if we let
\begin{equation}\label{rAO}
r(A,\O):=\sup\{\re s:s\in\po(\widetilde\zeta_{A,\O})\setminus\{D\}\},
\end{equation}
then for any positive real number $\e$, the tube function $t\mapsto|A_t\cap\O|$ has the following asymptotic expansion, with error term$:$
\begin{equation}\label{tubealpha4}
|A_t\cap\O|=t^{N-D}h(t)\big({\mathcal M}+O(t^{D-r(A,\O)-\e})\big)\q\mbox{\rm as\q $t\to0^+$.}
\end{equation}
\end{theorem}

\begin{proof}
Equation \eqref{tubealpha3} follows from Equation \eqref{tubemlog} of Theorem \ref{degg}, by choosing the screen $\bm{S}$ to be the vertical line $\{\re s=D-\a\}$; that is, $S(x)=D-\a$ for all $x\in\eR$. Indeed, in this case, we have $D-\sup S=D-(D-\a)=\a$.

Equation \eqref{tubealpha4} follows easily from \eqref{tubealpha3}, by letting $\a:=D-r(A,\O)-\e$, since
$s=D$ is the only pole in the open right half-plane $\{\re s>r(A,\O)+\e=D-\a\}$, with $\e>0$ small enough.
\end{proof}

Furthermore, from the functional equation \eqref{equ_tilde}, one can see clearly that when $\ov\dim_BA<N$, the value of the order $\a(A,\O)$
can be analogously defined by using the relative distance zeta function $\zeta_{A,\O}$ instead of the relative tube zeta function $\widetilde\zeta_{A,\O}$ in Equation \eqref{aAO}. Moreover,
$D(\zeta_{A,\O})=D(\widetilde\zeta_{A,\O})=\ov\dim_B(A,\O)$, $D_{\rm mer}(\zeta_{A,\O})=D_{\rm mer}(\widetilde\zeta_{A,\O})$, and the analog of Theorem \ref{tubealpha} can be easily stated and proved in the case of the relative distance zeta function of a given RFD $(A,\O)$.

\begin{remark}
We point out that it is possible that the conclusion of Theorem \ref{degg} (and of Theorem \ref{log-mink} below, respectively) is also true in the case when there exists an infinite sequence of nonreal complex dimensions of $(A,\O)$ with real part $\ov{D}$ such that each of them has multiplicity strictly less than that of $\ov{D}$.\footnote{See Example \ref{1/2-tube_formula} above which illustrates such a situation.}
Note that the fractal tube formula \eqref{degene} also holds pointwise in this case, but in order to obtain the conclusion about $\dim_B(A,\O)$ and the $h$-Minkowski measurability
of $(A,\O)$, the problem lies in the justification of the interchange of the limit as $t\to0^+$ and the infinite sum which appears in this case in Equation \eqref{degene}.
Of course, a priori, we do not have such a justification to our disposal without making additional assumptions on the nature of the convergence of the sum in \eqref{degene}.
\end{remark}

\medskip

It would be interesting to try to extend the above result and obtain a type of gauge Minkowski measurability criterion, in the likes of Theorem \ref{criterion}.\footnote{Note that in \cite{lapidushe} a gauge Minkowski measurability criterion was obtained for fractal strings, extending to the case of non power laws the one obtained (when $h\equiv1$) in [LapPo1,2] although this criterion does not involve the notion of complex dimensions and is stated only in terms of the underlying gauge function $h$ and the asymptotic behavior of the lengths of the string.}
For instance, see \cite[Theorem 4.5.1]{fzf} for a partial converse of the above theorem in the case when the relative tube function satisfies the following pointwise asymptotic expansion, with error term:
\begin{equation}
|A_t\cap\O|=t^{N-D}(\log t^{-1})^{m-1}(\mathcal{M}+O(t^{\alpha}))\quad\textrm{as}\quad t\to0^+,
\end{equation}
where, $m\in\eN$, $m\geq 1$ and $\alpha>0$. 

\medskip

We conclude this section by reformulating Theorem \ref{degg} in terms of the distance (instead of the tube) zeta function but we omit the proof and refer instead the interested reader to \cite[Theorem 5.4.32.]{fzf}.
\medskip

\begin{theorem}\label{log-mink}
Let $(A,\O)$ be a relative fractal drum in $\eR^N$ such that $\ov{\dim}_B(A,\O)<N$.
Also assume that $(A,\O)$ is $d$-languid with $\kappa_d<0$ or is such that $(\lambda A,\lambda\O)$ is strongly $d$-languid for some $\lambda>0$ with $\kappa_d<1$, for a screen passing strictly between the critical line $\{\re s=\ov{\dim}_B(A,\O)\}$ and all the complex dimensions of $(A,\O)$ with real part strictly less than $\ov{D}:=\ov{\dim}_B(A,\O)$.
Furthermore, suppose that $\ov{D}$ is the only pole of the relative distance zeta function $\zeta_{A,\O}$ with real part equal to $\ov{D}$ of multiplicity $m\geq1$ and, additionally, that there exists $($at most$)$ finitely many  nonreal poles of ${\zeta}_{A,\O}$ with real part $\ov{D}$.
Moreover, assume that the multiplicity of each of those nonreal poles is of order strictly less than $m$.
Then, $\dim_B(A,\O)$ exists and is equal to $D:=\ov{D}$.
Finally, $\mathcal{M}^{D}(A,\O)$ exists and is equal to $+\ty$; hence, $(A,\O)$ is Minkowski degenerate.

In addition, an appropriate gauge function for $(A,\O)$ is $h(t):=(\log t^{-1})^{m-1}$ for all $t\in(0,1)$ and we have that, {\rm relative to} $h$, the RFD $(A,\O)$ is not only $h$-Minkowski nondegenerate but is also $h$-Minkowski measurable, with $h$-Minkowski content given by
\begin{equation}
\mathcal{M}^{{D}}(A,\O,h)=\frac{\zeta_{A,\O}[D]_{-m}}{(N-D)(m-1)!},
\end{equation}
where $\zeta_{A,\O}[D]_{-m}$ denotes the coefficient corresponding to $(s-D)^{-m}$ in the Laurent expansion of ${\zeta}_{A,\O}$ around $s=D$.

Finally, the exact same conclusions as in Theorem \ref{degg} hold concerning the asymptotic expansion of 
$|A_t\cap\O|$ in either \eqref{tubemlog} or \eqref{tubem}, but with $\zeta_{A,\O}$ in place of $\widetilde\zeta_{A,\O}$
in the respective hypotheses.
\end{theorem}

\section{Essential singularities of fractal zeta functions (and beyond)}

The theory of meromorphic functions involves, among other things, various constructions of their meromorphic extensions to suitable open connected domains in the complex plane $\Ce$.  Meromorphic functions (of a single complex variable) and their meromorphic extensions can have at most countably many poles (since the latter are necessarily isolated). As shown in [Lap-vFr1,2] and \cite{fzf}, a large class of meromophic functions can be generated by various bounded fractal strings, fractal sets and RFDs, via the corresponding fractal zeta functions (more precisely, by geometric zeta functions $\zeta_{\mathcal L}$ of fractal strings $\mathcal L$, by distance and tube zeta functions $\zeta_A$ and $\tilde\zeta_A$ of bounded subsets $A$ of Euclidean spaces, as well as by relative distance and relative tube zeta functions $\zeta_{A,\O}$ and $\tilde\zeta_{A,\O}$ associated with RFDs $(A,\O)$ in $\eR^N$). 

It turns out, however, that it is both geometrically and analytically interesting to consider a much larger class of functions than just meromorphic functions. Namely, it has first been noticed by the second author that it is possible to construct a bounded fractal string $\mathcal L$ such that the corresponding geometric zeta function $\zeta_{\mathcal L}$ (or any of its meromorphic extensions) does not possess any pole, while it has a unique essential singularity. See \cite{ra} and \cite[pp.\ 214--215, 281--282, 288--289]{fzf}.

This motivated us to consider a class of complex-valued functions of a single complex variable, called {\em paramorphic functions}, which may possess not only poles but also essential singularities; see \cite{laprazu9}. Much as in the case of meromorphic extensions of a given meromorphic function, it is possible to consider paramorphic continuations of paramorphic functions.
In this manner, we can extend the definition of complex dimensions (of bounded fractal strings $\mathcal L$, bounded subsets $A$ of Euclidean spaces and RFDs $(A,\O)$ in $\eR^N$) to include not only the poles, but also the essential singularities of the corresponding fractal zeta functions.
For the detailed results, see the forthcoming paper \cite{laprazu9} where we construct a bounded fractal string $\mathcal L$ such that the associated geometric zeta function does not possess any pole, and has an arbitrary number of essential singularities, either finite or countably infinite.

Much as in the case of meromorphic extensions of a given fractal zeta function $f$, where we have defined the `abscissa of meromorphic continuation' $D_{\mathrm{mer}}(f)$ (see \cite[p.\ 85]{fzf}), it is possible to define the {\em abscissa of paramorphic continuation}  $D_{\mathrm{par}}(f)$ of $f$. In \cite{laprazu9}, we constructed a bounded fractal string $\mathcal L$ with arbitrarily prescribed values of $D_{\mathrm{par}}(\zeta_{\mathcal L})<D_{\mathrm{mer}}(\zeta_{\mathcal L})<D(\zeta_{\mathcal L})$ contained in $(0,1]$.
Furthermore, by adapting the methods described in \cite{fzf}, one can extend this result to bounded sets $A$ and RFDs $(A,\O)$ in $\eR^N$, for any $N\ge1$, and with arbitrarily prescribed values of $D_{\mathrm{par}}(f)<D_{\mathrm{mer}}(f)<D(f)$ in the interval $(0,N]$ for $f:=\zeta_A$, and in the interval $(-\ty,N]$ for $f:=\zeta_{A,\O}$.

Along similar lines, but also in a somewhat different and significantly more general situation, it would be interesting to develop a theory of complex dimensions valid for analytic fractal zeta functions defined on suitable Riemann surfaces (and even on complex analytic varieties or more generally, on analytic sets); see [Lap7, \S4] for motivations for such an extension. In the geometric setting of bounded sets and RFDs in $\eR^N$, and for the case of Riemann surfaces, the beginning of such a theory is developed in the work in preparation [LapRa\v Zu10] in certain concrete examples for which the usual power scaling laws do not hold and hence, the notion of generalized Minkowski contents (as in [HeLap], \cite[\S6.1.12]{fzf} and in \S\ref{gauge_mink} above) must be used, with the underlying gauge functions suitably adapted to the types of the singularities of the fractal zeta functions involved. We stress, however, that a lot more needs to be done in this research direction, well beyond the initial results obtained in~[LapRa\v Zu10].

\end{document}